\documentclass{article}

\usepackage[preprint]{neurips_2026}

\usepackage[utf8]{inputenc} 
\usepackage[T1]{fontenc}    
\usepackage{hyperref}       

\usepackage{url}            
\usepackage{booktabs}       
\usepackage{amsfonts}       
\usepackage{nicefrac}       
\usepackage{microtype}      
\usepackage{xcolor}         

\usepackage[final]{graphicx}
\usepackage{subcaption}
\usepackage{booktabs} 
\usepackage{comment}

\usepackage{amsmath}
\usepackage{amssymb}
\usepackage{mathtools}
\usepackage{amsthm}
\usepackage{algorithm}
\usepackage{algpseudocode}

\usepackage[capitalize,noabbrev]{cleveref}

\usepackage{framed}

\usepackage{enumitem}

\usepackage{tikz}
\usetikzlibrary{arrows.meta, positioning, calc, shapes.misc, decorations.pathreplacing,fit, backgrounds}
\definecolor{bordercolor}{RGB}{180, 180, 180}

\usepackage{float}
\usepackage{bm}
\usepackage{multirow}
\usepackage{wrapfig}

\usepackage{pifont}

\newcommand{\sn}[1]{\ensuremath{\times 10^{#1}}}

\theoremstyle{plain}
\newtheorem{theorem}{Theorem}[section]
\newtheorem{proposition}[theorem]{Proposition}
\newtheorem{lemma}[theorem]{Lemma}
\newtheorem{corollary}[theorem]{Corollary}
\theoremstyle{definition}
\newtheorem{definition}[theorem]{Definition}
\newtheorem{assumption}[theorem]{Assumption}
\theoremstyle{remark}
\newtheorem{remark}[theorem]{Remark}

\crefname{definition}{Definition}{Definitions}
\Crefname{definition}{Definition}{Definitions}

\usepackage[textsize=tiny]{todonotes}

\tikzset{
    flow/.style={->, >={Stealth[length=1.8mm, width=1.3mm]}, line width=0.6pt, black},
    flowsmall/.style={->, >={Stealth[length=1.5mm, width=1.1mm]}, line width=0.4pt, black},
    brace/.style={decorate, decoration={brace, amplitude=2.5pt, mirror}, line width=0.4pt},
    panel/.style={inner sep=0pt},
    sectitle/.style={font=\fontsize{9}{10}\selectfont\rmfamily},
    ptitle/.style={font=\fontsize{8}{10}\selectfont\rmfamily},
    psub/.style={font=\fontsize{6}{9}\selectfont\rmfamily},
}

\title{HyCOP: Hybrid Composition Operators for Interpretable Learning of PDEs}

\author{%
  Jinpai Zhao\thanks{Equal contribution. Corresponding author: \texttt{nishpan@lanl.gov}. LA-UR-26-23260.} \\
  Oden Institute\\
  University of Texas at Austin\\
  Austin, TX, USA \\
  \texttt{max.zhao@utexas.edu} \\
  \And
  Nishant Panda\footnotemark[1] \\
  Information Sciences, CAI-3\\
  Los Alamos National Laboratory\\
  Los Alamos, NM 87545, USA \\
  \texttt{nishpan@lanl.gov} \\
  \AND
  Yen Ting Lin \\
  Information Sciences, CAI-3\\
  Los Alamos National Laboratory\\
  Los Alamos, NM 87545, USA \\
  \texttt{yentingl@lanl.gov} \\
  \And
  Eirik Valseth \\
  Norwegian University of Life Sciences\\
  \AA{}s, Norway / Simula Research Lab.\\
  Oslo, Norway \\
  \texttt{eirik.valseth@nmbu.no} \\
  \AND
  Diane Oyen \\
  Information Sciences, CAI-3\\
  Los Alamos National Laboratory\\
  Los Alamos, NM 87545, USA \\
  \texttt{doyen@lanl.gov} \\
  \And
  Clint Dawson \\
  Oden Institute\\
  University of Texas at Austin\\
  Austin, TX, USA \\
  \texttt{clint.dawson@austin.utexas.edu} \\
}

\begin{document}

\maketitle

\begin{abstract}
  We introduce HyCOP, a modular framework that learns parametric PDE solution operators by composing simple modules (advection, diffusion, learned closures, boundary handling) in a query-conditioned way. Rather than learning a monolithic map, HyCOP learns a policy over short programs---which module to apply and for how long---conditioned on regime features and state statistics. Modules may be numerical sub-solvers or learned components, enabling hybrid surrogates evaluated at arbitrary query times without autoregressive rollout. Across diverse PDE benchmarks, HyCOP produces interpretable programs, delivers order-of-magnitude OOD improvements over monolithic neural operators, and supports modular transfer through dictionary updates (e.g., boundary swaps, residual enrichment). Our theory characterizes expressivity and gives an error decomposition that separates composition error from module error and doubles as a process-level diagnostic.
\end{abstract}

\section{Introduction}
\label{sec:intro}
 
Scientific machine learning (SciML) is moving beyond monolithic replacement of numerical solvers.
Across fluid dynamics, climate modeling, and multi-physics simulation, practitioners increasingly build \emph{hybrid} pipelines that embed pretrained neural surrogates alongside classical numerical modules within larger workflows~\citep{jakeman2026vvsciml,singh2017augmented}.
Verification and validation frameworks now treat such hybrid models as a first-class category~\citep{jakeman2026vvsciml}, recognizing that trustworthy SciML requires not only accurate components but principled ways to combine them.
A central question therefore remains: \emph{given a heterogeneous collection of numerical and learned components, how should a scientist compose them into a surrogate that is robust, interpretable, and modular?}
 
This paper provides one answer for a broad and practically important class of problems: \emph{spatiotemporal PDEs that admit a meaningful decomposition into constituent physical mechanisms}, for example, advection--diffusion--reaction equation, shallow-water equation, Navier--Stokes equation, and many multi-physics systems with advective, viscous, and forcing splits.
Traditionally, such systems have been modeled incrementally: a modeler starts with known or postulated physical mechanisms to compose processes that explains data; when the composition fails to explain the data, modelers use the \emph{pattern of failure} to identify, hypothesize, or discover the missing mechanism.
Monolithic neural operators discard this workflow: they learn a single black-box map from paired data, and when they fail under distribution shift, they provide little guidance on \emph{why} or \emph{what to fix}.
 
\textbf{The myth of expensive numerical solvers.}
A common motivation for monolithic surrogates is that numerical PDE solvers are slow.
But this conflates the cost of the \emph{full coupled solver} with the cost of its \emph{constituent processes}.
The full solver is expensive because it resolves the non-commutative interaction of multiple mechanisms at fine scales.
By contrast, individual process solvers---a spectral diffusion step, an upwind advection scheme, an explicit reaction update---benefit from decades of algorithmic and hardware optimization, and often run faster than a neural operator forward pass at comparable resolution.
The bottleneck is not the primitives; it is knowing how to compose them.

\subsection{Scientific Machine Learning Regimes}
These observations motivate the central thesis: \emph{encoding a PDE's process decomposition as a structural prior---and learning only the composition policy---yields surrogates that are more robust under shift, more interpretable, and more modular for transfer than monolithic alternatives.}
HyCOP is designed to address two distinct SciML needs.

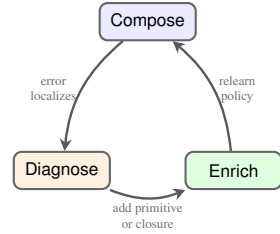
\begin{wrapfigure}{r}{0.32\columnwidth}
\vspace{-1.2em}
\centering
\begin{tikzpicture}[
  scale=0.82, every node/.style={transform shape},
  box/.style={rounded corners=3pt, draw=black!60, line width=0.8pt,
              minimum width=1.55cm, minimum height=0.6cm,
              font=\fontsize{8}{9.5}\selectfont\sffamily, align=center,
              inner sep=2pt},
  arr/.style={->, >={Stealth[length=1.6mm, width=1.6mm]}, line width=0.9pt,
              black!65},
  annot/.style={font=\fontsize{6}{7}\selectfont\rmfamily, text=black!55,
                align=center, inner sep=1pt},
]
  \node[box, fill=blue!8]    (compose)  at (90:1.6)  {Compose};
  \node[box, fill=orange!12] (diagnose) at (210:1.6) {Diagnose};
  \node[box, fill=green!12]  (enrich)   at (330:1.6) {Enrich};

  \draw[arr] (compose)  to[bend right=22]
    node[annot, pos=0.5, left=1pt]  {error\\localizes}          (diagnose);
  \draw[arr] (diagnose) to[bend right=22]
    node[annot, pos=0.5, below=1pt] {add primitive\\or closure} (enrich);
  \draw[arr] (enrich)   to[bend right=22]
    node[annot, pos=0.5, right=1pt] {relearn\\policy}           (compose);
\end{tikzpicture}
\vspace{-0.3em}
\caption{\textbf{Compose--diagnose--enrich.}
HyCOP's workflow for incomplete or hybrid physics (\S\ref{sec:exp:regime_b}).}
\label{fig:cde_loop}
\vspace{-1.0em}
\end{wrapfigure} 
 
 

\textbf{Regime~A: surrogate fitting (physics is known).}
When the governing equations are fully specified, the scientist wants a fast surrogate for downstream tasks that remains accurate as initial conditions, boundary conditions, parameters, or resolution shift at test time.
HyCOP encodes each known process as a dictionary primitive and learns only the composition policy (${\sim}$50--100 parameters), yielding short programs adapted to each query's physical regime.
Because each primitive generalizes independently, the composition degrades gracefully under shift.
For example, a HyCOP policy trained on smooth shallow-water equations with periodic boundaries transfers zero-shot to dam-break shocks, wall boundaries, discontinuities, and parameter ranges never seen in training, outperforming monolithic baselines by $10{\times}$.
Swapping the boundary primitive for a wall-boundary module, without retraining, yields a further $4{\times}$ improvement (\S\ref{sec:exp:regime_a}).
 
\textbf{Regime~B: adaptation and discovery (physics is incomplete or heterogeneous).}
When the process decomposition is only partially known, or when components come from different sources, HyCOP supports a \emph{compose--diagnose--enrich} loop (Figure~\ref{fig:cde_loop}).
We demonstrate two complementary use cases on the same AD$\to$ADR benchmark:
\begin{enumerate}[leftmargin=1.5em,itemsep=1pt,topsep=2pt,label=(\alph*)]
\item \emph{Unknown missing physics.}
A scientist models advection--diffusion (AD) but encounters ADR data.
Composing the AD dictionary produces structured failure: errors concentrate where reaction dominates.
A residual closure (UNO) trained on the discrepancy is added to the dictionary, and relearning only the policy recovers accuracy while revealing where the closure activates.
\item \emph{Known missing physics, heterogeneous primitives.}
A lab already has a pretrained AD surrogate, and new data reveals missing reaction that is well characterized.
HyCOP composes the pretrained AD surrogate (learned primitive) with a textbook reaction solver (numerical primitive), learning a policy over this \emph{hybrid} dictionary without retraining either component.
\end{enumerate}
Both paths arrive at a hybrid dictionary---numerical and learned components orchestrated by a small policy---but from different scientific starting points. This is why ``Hybrid'' appears in the title.
 
Monolithic neural operators are therefore not replaced by this framework, but absorbed into it: as expensive constituent processes in Regime~A, or as learned primitives modeling missing physics in Regime~B. HyCOP orchestrates heterogeneous building blocks rather than competing with them.

\paragraph{Our contributions (overview in Figure~\ref{fig:pipeline}).}
\begin{enumerate}[leftmargin=*,itemsep=2pt,topsep=2pt]
\item \textbf{Framework.}
We introduce HyCOP, a modular framework that learns PDE solution operators as query-conditioned compositions of reusable primitives---numerical, neural, or learned closures---supporting hybrid dictionaries and multi-time prediction without autoregressive rollout (\S\ref{sec:method}).
 
\item \textbf{Theory.}
We introduce \emph{compositional operator flows} as a new hypothesis class and provide the first learning-theoretic analysis of learned query-conditioned compositions, including an error decomposition that separates \emph{splitting error} from \emph{primitive error} and doubles as a diagnostic (\S\ref{sec:theory}).
 
\item \textbf{Experiments.}
Across five benchmarks including 2D Navier--Stokes, HyCOP delivers order-of-magnitude OOD improvements over state-of-the-art monolithic baselines, with more than $25{\times}$ fewer training forward passes and over $10{\times}$ shorter wall-clock training time (\S\ref{sec:experiments}).
 
\item \textbf{Modularity.}
HyCOP supports zero-shot adaptation through dictionary updates: on the dam-break experiment, the learned policy alone handles OOD shock initial conditions and a $100{\times}$ larger grid, while swapping the boundary primitive resolves the remaining boundary mismatch.
For missing physics (AD$\to$ADR), learned residuals and hybrid numerical--neural dictionaries capture the absent reaction mechanism without retraining existing primitives (\S\ref{sec:experiments}, \S\ref{sec:ablations}).
\end{enumerate}

\begin{figure}[t]
  \centering
  \includegraphics[width=\textwidth]{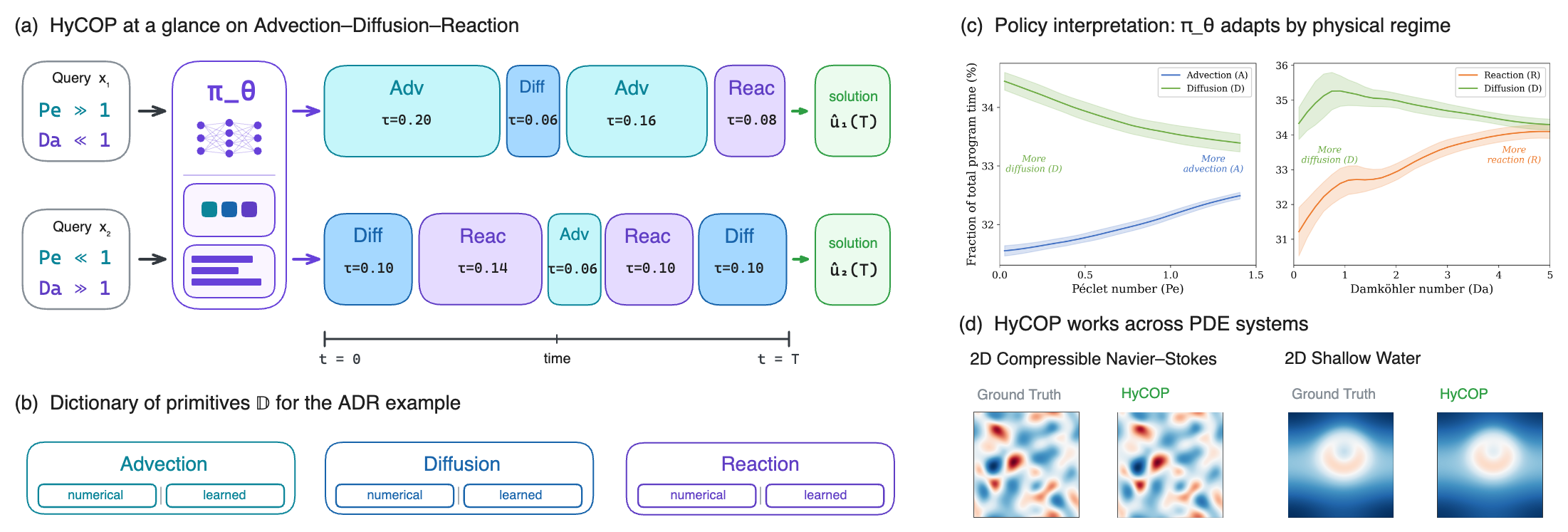}
  \caption{\textbf{HyCOP at a glance.}
  \textbf{(a)}~A small learned policy $\pi_\theta$ (${\sim}$50--100 parameters) maps a query $x=(u_0,\mu,T)$ to a program of (primitive, duration) pairs from dictionary $\mathbb{D}$, composed over $[0,T]$ to produce $\hat{u}(T)$; programs adapt to the regime (walkthrough in \S\ref{sec:method:example}).
  \textbf{(b)}~Each primitive can be numerical, learned, or both, enabling hybrid dictionaries (\S\ref{sec:exp:transfer}).
  \textbf{(c)}~$\pi_\theta$ smoothly reallocates time with regime numbers, consistent with dominant-process physics.
  \textbf{(d)}~The framework applies across PDE systems (coverage: Table~\ref{tab:benchmark_coverage}; results: \S\ref{sec:experiments}, Appendix~\ref{app:exp_configs}).}
  \label{fig:pipeline}
\end{figure}

\section{Background and Positioning}
\label{sec:background}

\textbf{Operator splitting.}
When a PDE generator decomposes as $\mathcal{F}=\sum_{i=1}^n \mathcal{F}_i$, operator splitting approximates the coupled flow $e^{t\mathcal{F}}$ by composing simpler sub-flows.
Lie--Trotter~\citep{trotter1959product} gives first-order accuracy; Strang~\citep{strang1968construction} second-order; higher-order schemes also exist~\citep{suzuki1990fractal,HairerEtAl2006}.
The error is governed by commutators $[\mathcal{F}_i,\mathcal{F}_j]$, but the schedule is fixed \emph{a priori} and cannot adapt to each query.
\textbf{Neural operators.}
Neural operators learn monolithic maps between function spaces, including FNO~\citep{li2020fourier}, DeepONet~\citep{lu2021deeponet}, Loc.\ Int.\ Diff.\ FNO~\citep{liuschiaffini2024localfno}, CNO~\citep{raonic2023cno}, PINO~\citep{li2021pino}, PDE-Refiner~\citep{lippe2023pderefiner}, and Poseidon~\citep{herde2024poseidon}.
These methods are powerful in distribution, but under shift the learned map degrades globally and the architecture does not localize failure to a specific process.
Related approaches such as operator inference~\citep{kramer2024opinf}, PINNs~\citep{raissi2019pinn}, GP/kernel surrogates~\citep{chen2021gppde,batlle2024kernelop}, and Koopman methods~\citep{bevanda2021koopman} address adjacent goals, but none learns an explicit, query-conditioned composition over reusable process primitives.
\textbf{The gap.}
Classical splitting encodes process structure but does not learn from data; neural operators learn from data but discard process structure.
HyCOP bridges this gap by learning a query-conditioned sequential composition---which process to apply and for how long---preserving the physics of each sub-flow rather than using fixed schedules or weighted blends.
\textbf{Concurrent work.}
Recent papers validate compositional structure but keep the schedule fixed:
Serrano et al.~\citep{serrano2026neuralsplitting} compose DISCO~\citep{morel2025disco} operators via fixed Strang splitting;
LegONet~\citep{zhang2026legonet} composes spectral blocks with fixed Strang and a related error decomposition;
Gopakumar et al.~\citep{gopakumar2026learning} learn physical operators under a fixed schedule;
and Koch et al.~\citep{koch2024neuralDAE} apply splitting to neural DAEs with structurally fixed decompositions.
HyCOP differs in that the learned object is the \emph{composition itself}: it learns \emph{what} to compose, \emph{for how long}, and \emph{conditioned on the regime}.\footnote{GEPS~\citep{koupai2024geps}, VENICE~\citep{wilhelm2024venice}, HINTS~\citep{zhang2024hints}, and scale-consistent training~\citep{li2025scaleconsistentlearningpartialdifferential} pursue complementary directions.}
\textbf{Positioning.}
Table~\ref{tab:positioning} contrasts the three paradigms; the capability unique to HyCOP among learned surrogates is \emph{process-level failure diagnosis}.

\section{HyCOP: Hybrid Composition Operators}
\label{sec:method}


We introduce HyCOP through a worked example (\S\ref{sec:method:example}), then present the general framework (\S\ref{sec:method:general}).

\subsection{A Worked Example: Advection--Diffusion--Reaction}
\label{sec:method:example}
Figure~\ref{fig:pipeline}a illustrates HyCOP on the 2D advection--diffusion--reaction equation $\partial_t u + \mathbf{c}\cdot\nabla u = D\,\Delta u + r\,u(1-u)$ with advection velocity $\mathbf{c}$, diffusivity $D$, and reaction rate $r$.
The dictionary contains three primitives---$\mathcal{O}_{\mathrm{adv}}$, $\mathcal{O}_{\mathrm{diff}}$, $\mathcal{O}_{\mathrm{react}}$---each implemented numerically, by a learned surrogate, or a mix of both (Figure~\ref{fig:pipeline}b).
Each can be evaluated for any duration $\tau>0$; none solves the coupled problem alone.
Given $u_0$, $|\mathbf{c}|{=}2.0$, $D{=}0.1$, $r{=}1.0$, and query time $t{=}0.5$, the policy $\pi_\theta$ maps dimensionless features (P\'eclet number, Damk\"ohler number, variance, gradient variance) to
\[
\underbrace{[\mathcal{O}_{\mathrm{adv}},\, \tau{=}0.20]}_{\text{step 1}}
\;\to\;
\underbrace{[\mathcal{O}_{\mathrm{diff}},\, \tau{=}0.06]}_{\text{step 2}}
\;\to\;
\underbrace{[\mathcal{O}_{\mathrm{adv}},\, \tau{=}0.16]}_{\text{step 3}}
\;\to\;
\underbrace{[\mathcal{O}_{\mathrm{react}},\, \tau{=}0.08]}_{\text{step 4}}.
\]
The durations sum to $t{=}0.5$, so the program allocates prediction time across processes rather than advancing with a fixed step.
Execution composes the four sub-flows to produce $u(0.5)$.
When Pe increases at test time, the policy reallocates time toward advection; each primitive remains valid, so the composition adapts where a monolithic network cannot.
If the dictionary is incomplete---e.g., the true system includes a process absent from the dictionary---HyCOP's errors localize where the missing process dominates, enabling the compose--diagnose--enrich loop described in \S\ref{sec:intro} and demonstrated in \S\ref{sec:exp:transfer}.

\subsection{General Framework: Training and Inference}
\label{sec:method:general}

\textbf{Setup and query-conditioned programs.}
We consider $\partial_t u=\mathcal{F}(u,\mu)$, $u(0){=}u_0$, on $\Omega$ with boundary condition $b$.
For a query $x{=}(\mu,u_0,b,\Omega)$ and time $t$, the solution operator is $\mathcal{S}(t;x){=}\Phi_t^{\mathcal{F}}(u_0)$.
When $\mathcal{F}\approx\sum_{j=1}^n \mathcal{F}_j$, HyCOP collects implemented sub-flows into a dictionary $\mathbb{D}{=}\{\widehat{\mathcal{O}}_1,\ldots,\widehat{\mathcal{O}}_n\}$, where each primitive may be a numerical sub-solver, a neural operator, or a learned closure.
The policy $\pi_\theta$ predicts a program $(j_1,\tau_1),\ldots,(j_k,\tau_k)$ and evaluates
\begin{equation}
\widehat{\mathcal{S}}(t;x)
= \widehat{\Psi}_\theta(x,t)
:= \widehat{\Phi}^{(j_k)}_{\tau_k}\circ\cdots\circ\widehat{\Phi}^{(j_1)}_{\tau_1}(u_0).
\label{eq:hycop_comp}
\end{equation}
HyCOP is therefore a \emph{learned integrator family}: a policy-selected composition of flows, not a weighted blend.
Multi-time queries are answered without autoregressive rollout.
Dictionary specifications and policy sizes appear in Table~\ref{tab:dictionary} (Appendix~\ref{app:exp_configs}).

\textbf{Training objective.}
We minimize the expected prediction error
\begin{equation}
\label{eq:objective}
J(\theta):=\mathbb{E}_{(x,t)\sim\rho}\big[\mathcal{L}(x,t;\theta)\big],
\qquad
\mathcal{L}(x,t;\theta)=\big\|u(t;x)-\widehat{\Psi}_{\theta}(x,t)\big\|_{L^2},
\end{equation}
where reference targets $u(t;x)$ come from a trusted solver.
The policy conditions on $(x,t)$ and scale-free features $f(x)$---dimensionless regime numbers (P\'eclet, Damk\"ohler, Froude) and coarse state statistics---rather than raw grid values, promoting resolution-invariant scheduling (\S\ref{sec:ablations}).

\textbf{Policy and optimization.}
At each program position $r\in\{1,\ldots,K_{\max}\}$, the policy outputs
(i) logits $z_r\in\mathbb{R}^n$ selecting a primitive via $\sigma_r=\mathrm{softmax}(z_r)$,
(ii) a positive duration $\tau_r=\mathrm{softplus}(a_r)$,
and (iii) an effective program length $k\in[2,K_{\max}]$.
We train with Evolution Strategies (ES)~\citep{Salimans2017} because $\theta\mapsto\mathcal{L}(x,t;\theta)$ is generally non-differentiable when primitives are legacy solvers or other black-box modules.
ES is practical here because the policy is low-dimensional (${\sim}$50--100 parameters) and population rollouts parallelize trivially.
The same hyperparameters ($M{=}500$, $\sigma{=}0.02$, 200 generations) are used across all benchmarks without per-problem tuning.
Algorithm~\ref{alg:hycop_es} and further details are in Appendix~\ref{app:training}; Figure~\ref{fig:hycop_training_es} illustrates the pipeline.

\begin{algorithm}[t]
\small
\setlength{\baselineskip}{0.95\baselineskip}
\caption{Training HyCOP with Evolution Strategies (ES)}
\label{alg:hycop_es}
\begin{algorithmic}[1]
\Statex \textbf{Inputs:} dictionary $\mathbb{D}{=}\{\mathcal{O}_1,\dots,\mathcal{O}_n\}$;\, query distribution $\rho$ over $(x,t)$;\, policy $\pi_\theta$;\, ES settings $(M,\sigma,\eta,\lambda)$.
\Statex \textbf{Outputs:} trained parameters $\theta^\star$;\, at inference, $\pi_{\theta^\star}(x,t)$ emits per-step primitive logits $z_r$, durations $\tau_r{=}\mathrm{softplus}(a_r)$, and program length $k$ --- together prescribing which $\mathcal{O}_j \in \mathbb{D}$ to apply, for how long, and in what order.
\Statex \hrulefill
\For{generation $=1,2,\dots$}
  \State Sample minibatch $\{(x_b,t_b)\}_{b=1}^B \sim \rho$;\, sample $\epsilon_i \sim \mathcal{N}(0,I)$ for $i{=}1,\dots,M$
  \State $L_i^\pm \gets \tfrac{1}{B}\sum_{b=1}^B \mathcal{L}\bigl(x_b,t_b;\,\theta\pm\sigma\epsilon_i\bigr)$ \quad // $\mathcal{L}$ runs $\widehat{\Psi}_\theta$ from Eq.~\eqref{eq:hycop_comp}
  \State $\{w_i^+,w_i^-\}_{i=1}^M \gets \mathrm{rank\text{-}shape}(\{L_i^+,L_i^-\})$;\, $g \gets \tfrac{1}{2M\sigma}\sum_{i=1}^M (w_i^+ - w_i^-)\,\epsilon_i$
  \State $\theta \gets (1-\lambda)\,(\theta - \eta\, g)$ \quad // ES step with weight decay
\EndFor
\State \Return $\theta^\star \gets \theta$
\end{algorithmic}
\end{algorithm}


\textbf{What HyCOP is not.}
For a query $x$ and dictionary $\mathbb{D}$, HyCOP produces an operator $\widehat{\Psi}_\theta(x,\cdot){:}\,u_0\mapsto\hat{u}(t)$ via Eq.~\eqref{eq:hycop_comp}---\emph{not} a mixture-of-experts: MoE blends outputs ($\hat{u}{=}\sum g_i f_i(u_0)$), HyCOP \emph{composes} flows sequentially, so order matters and durations are physical integration times.
It is \emph{not} autoregressive rollout: programs operate at the operator level over learned durations rather than fixed-$\Delta t$ marching.
And it is \emph{not} fixed splitting: HyCOP learns the schedule from data, conditioned on regime features, whereas concurrent work~\citep{serrano2026neuralsplitting,zhang2026legonet} uses fixed Strang.


\section{A Learning Theory for Compositional Operator Flows}
\label{sec:theory}

Neural operator theory gives universal approximation results for monolithic maps~\citep{kovachki2023neuraloperator,lu2021deeponet}.
Classical splitting theory analyzes fixed schedules~\citep{HairerEtAl2006,strang1968construction}.
Neither addresses a \emph{learned, query-conditioned composition} of approximate sub-flows---a distinct regime where the schedule is not fixed and the hypothesis class is not unstructured.
We develop the first learning-theoretic analysis of \emph{compositional operator flows} as a hypothesis class for PDE surrogates.

\textbf{Composite flows as a hypothesis class}
\label{sec:theory:expressivity}

Consider a stable split system (Definition~\ref{def:stable_split}, Appendix~\ref{sec:proof}): $\partial_t u{=}\mathcal{F}(u,\mu)$ with $\mathcal{F}{=}\sum_{i=1}^n \mathcal{F}_i$, each sub-problem well-posed.
A $k$-step composite flow is $\Psi^{(k)}{=}\Phi^{(j_k)}_{\tau_k}{\circ}\cdots{\circ}\Phi^{(j_1)}_{\tau_1}$; let $\mathcal{C}{=}\bigcup_{k\ge 1}\mathcal{C}_k$ be the class of all finite composite flows, and $\widehat{\mathcal{C}}$ the corresponding class with implemented primitives.

Unlike the hypothesis classes of FNO or DeepONet which are general continuous operators without process structure, HyCOP's hypothesis class $\mathcal{C}$ is restricted to sequential compositions of process-specific sub-flows, making every element interpretable and modular.
$\mathcal{C}$ is dense in the space of PDE solution operators on compact sets (Theorem~\ref{thm:existence}, Appendix~\ref{sec:proof}): no entanglement of processes is needed.
The more consequential results, a structured error decomposition and policy existence, follow. The overall theory roadmap is summarized in Figure~\ref{fig:theory_roadmap}.

\begin{figure}[h]
\centering
\begin{tikzpicture}[
  node distance=0.7cm and 0.4cm,
  box/.style={draw, rounded corners=3pt, thick, align=center, inner sep=3pt,
              minimum width=2.5cm, minimum height=0.6cm, font=\footnotesize},
  arr/.style={->, thick, >=Latex, black!70}
]
\node[box, fill=orange!12] (thm4) {Thm.~\ref{thm:fitting}: Universality\\[-1pt]\scriptsize with structure};
\node[box, fill=gray!12, below left=0.55cm and 1.6cm of thm4] (thm1) {Thm.~\ref{thm:existence}: Expressivity\\[-1pt]\scriptsize Lie algebra / comm.\ bounds};
\node[box, fill=blue!8, below=0.55cm of thm4] (thm2) {Thm.~\ref{thm:error_decomp}: Error decomp.\\[-1pt]\scriptsize splitting + primitive};
\node[box, fill=green!8, below right=0.55cm and 1.6cm of thm4] (thm3) {Thm.~\ref{thm:optimal_policy}: Policy existence\\[-1pt]\scriptsize $\varepsilon$-optimal $\theta^\star$};
\draw[arr] (thm1.north) -- (thm4.south west);
\draw[arr] (thm2.north) -- (thm4.south);
\draw[arr] (thm3.north) -- (thm4.south east);
\end{tikzpicture}
\caption{Expressivity (Thm.~\ref{thm:existence}) and error decomposition (Thm.~\ref{thm:error_decomp}) provide the approximation-theoretic foundation; policy existence (Thm.~\ref{thm:optimal_policy}) ensures learnability; together they yield universality with structure (Thm.~\ref{thm:fitting}) within the compositional class $\widehat{\mathcal{C}}$.}
\label{fig:theory_roadmap}
\end{figure}
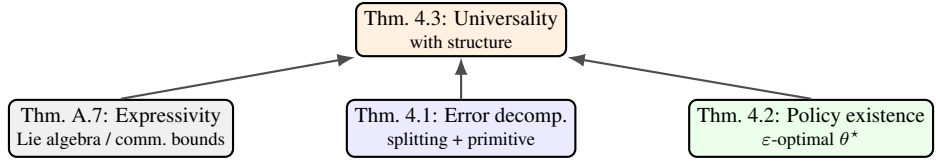

\subsection{Structured error decomposition}
\label{sec:theory:decomp}

When exact sub-flows are replaced by implemented primitives, the error separates:
\begin{equation}
\underbrace{\|u(t){-}\widehat{\Psi}_\theta\|}_{\text{total}}
\;\le\;
\underbrace{\|u(t){-}\Psi_\theta\|}_{\substack{\text{splitting err.}\\\text{(policy)}}}
\;+\;
\underbrace{\|\Psi_\theta{-}\widehat{\Psi}_\theta\|}_{\substack{\text{primitive err.}\\\text{(modules)}}},
\label{eq:error_decomp}
\end{equation}
where $\Psi_\theta$ uses exact sub-flows and $\widehat{\Psi}_\theta$ uses implementations.

\begin{theorem}[Error decomposition]
\label{thm:error_decomp}
Under regularity/stability assumptions (Appendix~\ref{sec:proof}), splitting error is $\mathcal{O}(h^p)$ for order-$p$ schedules (constant depends on $[\mathcal{F}_i,\mathcal{F}_j]$); primitive error satisfies $\|\Psi_\theta{-}\widehat{\Psi}_\theta\|_{L^2} \le C_{\mathrm{sol}} e^{\bar{\omega}T}(\sum_j \tau_j^{q+1})\|u_0\|_{H^s}$ with $q$ the sub-solver order.
Both terms admit practical estimators (Appendix~\ref{sec:proof}).
\end{theorem}

Under shift, each sub-flow remains well-posed, so primitive error is controlled independently of the regime; the policy adapts the schedule.
When a mechanism is missing, a \emph{dictionary mismatch} term prepends the bound, yielding a three-term decomposition (mismatch $+$ splitting $+$ primitive) where each term maps to a distinct intervention: enrich the dictionary, refine the policy, or improve a module.
This is the theoretical basis for the compose--diagnose--enrich loop (\S\ref{sec:intro}).

\subsection{Near-optimal policies and universality with structure}
\label{sec:theory:policy}

\begin{theorem}[$\varepsilon$-optimal policy]
\label{thm:optimal_policy}
Let $\Lambda$ be compact, $\Theta$ a compact policy class (programs up to $K_{\max}$ steps), and assume primitive stability.
Then $J(\theta){=}\mathbb{E}[\mathcal{L}(x,t;\theta)]$ attains its infimum on $\Theta$; an $\varepsilon$-optimal $\theta^\star$ exists; and $x{\mapsto}\theta^\star(x)$ can be approximated by a neural network (Appendix~\ref{sec:proof}).
\end{theorem}

\noindent Combining Theorem~\ref{thm:existence} (Appendix), Theorem~\ref{thm:error_decomp}, and Theorem~\ref{thm:optimal_policy}:

\begin{theorem}[Universality for compositional surrogates]
\label{thm:fitting}
Under the above conditions, for any $\varepsilon{>}0$ and compact $\Lambda$, $\exists$ neural policy $\pi_\phi$ with $\sup_{x\in\Lambda}\|u(t;x){-}\widehat{\Psi}_{\pi_\phi(x)}(x,t)\|_{L^2}{<}\varepsilon$.
\end{theorem}


\textbf{Universality with structure.}
This is \emph{not} a corollary of standard universality, which guarantees only that \emph{some} black-box operator approximates the map.
Theorem~\ref{thm:fitting} says universality holds within $\widehat{\mathcal{C}}$, where every step has physical meaning and errors decompose as in~\eqref{eq:error_decomp}---universality with built-in modularity, interpretability, and process-level diagnostics.
LegONet~\citep{zhang2026legonet} has a related decomposition for fixed Strang; our results are the first for non-fixed, learned, query-conditioned policies.


\section{Experiments}
\label{sec:experiments}

The theory predicts graceful OOD degradation, successful transfer via dictionary updates, and interpretable programs.
We test these predictions across two regimes mirroring the scientific workflows introduced in \S\ref{sec:intro}.

\textbf{Setup.}
We benchmark on 2D compressible Navier--Stokes (NS, PDEBench~\citep{PDEBench2022}), 2D shallow-water equations (SWE), and 2D advection--diffusion--reaction (ADR, Fisher--KPP).
Baselines include FNO~\citep{li2020fourier}, DeepONet~\citep{lu2021deeponet}, Loc.\ Int.\ Diff.\ FNO~\citep{liuschiaffini2024localfno}, U-Net~\citep{ronneberger2015unet}, PINO~\citep{li2021pino}, Poseidon~\citep{herde2024poseidon}, and classical Strang splitting on the same primitives as HyCOP.
We report HyCOP in two dictionary configurations:
HyCOP, with an all-numerical primitive dictionary, and HyCOP-Hyb, in which one or more primitives are pretrained FNO surrogates with FiLM~\citep{perez2018film} conditioning on query time and the remaining primitives are numerical.
HyCOP-Hyb appears on the NS benchmark (Table~\ref{tab:regime_a}a) and in the Regime~B adaptation experiments (\S\ref{sec:exp:regime_b}); all other results use the all-numerical HyCOP.
All methods train on the same data; HyCOP uses identical ES hyperparameters ($M{=}500$, $\sigma{=}0.02$, 200 generations) for all from-scratch training across benchmarks, without per-problem tuning. Adaptation experiments (\S\ref{sec:exp:regime_b})
use shorter ES runs from a warm-started policy; details in Appendix~\ref{app:ad_adr}.
Full PDE specifications, parameter ranges, and OOD definitions are in Appendix~\ref{app:exp_configs}; additional 1D benchmark results are in Appendix~\ref{app:1d_systems}.
\textbf{Fairness.}
All methods train on the same data distribution and sample count.
HyCOP is trained with variable query time $T$; fixed-time baselines are evaluated at the matching target time; autoregressive baselines use free rollout.
Primitives have zero learnable parameters (numerical variant); HyCOP-Hyb's FNO-FiLM primitives are pretrained on single-process data and frozen.

\subsection{Regime A: Surrogate Fitting}
\label{sec:exp:regime_a}

When the physics is fully known, HyCOP learns compositional programs that generalize across parameter regimes, initial conditions, boundary conditions, and resolutions.

 
\begin{table}[t]
\caption{\textbf{Regime~A: Surrogate fitting (known physics).}
Relative $L^2$ error ($\downarrow$) unless noted.
$^\dagger$PDEBench-reported numbers.
HyCOP uses an all-numerical primitive dictionary throughout; HyCOP-Hyb (panel a) uses a mixed dictionary of numerical and pretrained learned primitives (NS setup: Appendix~\ref{app:ns}). Baseline training details: Appendix~\ref{app:common_setup}.
All models trained on the same data; HyCOP uses identical ES hyperparameters ($M{=}500$, $\sigma{=}0.02$, 200 generations) across all benchmarks.}
\label{tab:regime_a}
\centering
 
 
\vspace{0.3em}
\noindent
\begin{minipage}[b]{0.34\textwidth}\centering
{\footnotesize\textbf{(a) 2D Compressible NS}}\\[-1pt]
{\scriptsize (PDEBench, $T{=}0.05$; nRMSE\,/\,cRMSE)}
\end{minipage}%
\hfill
\begin{minipage}[b]{0.60\textwidth}\centering
{\footnotesize\textbf{(b) Fixed-time: 2D SWE and 2D ADR}}\\[-1pt]
\mbox{}
\end{minipage}
 
\vspace{4pt}
\hrule height \heavyrulewidth
\vspace{3pt}
 
\noindent
{\setlength{\tabcolsep}{0pt}%
\begin{tabular}{@{}p{0.34\textwidth}@{\hspace{5pt}}|@{\hspace{5pt}}p{0.59\textwidth}@{}}
 
\centering
\scriptsize
\setlength{\tabcolsep}{3pt}%
\begin{tabular}{@{}lcc@{}}
Method & nRMSE & cRMSE \\
\midrule
U-Net$^\dagger$    & 5.10\sn{0}  & 3.20\sn{-2} \\
PINO               & 6.34\sn{-1} & 1.37\sn{0}  \\
FNO$^\dagger$      & 3.60\sn{-1} & 3.20\sn{-3} \\
DeepONet           & 7.91\sn{-2} & 4.76\sn{-2} \\
Strang             & 7.56\sn{-2} & 1.59\sn{-8} \\
HyCOP-Hyb          & 7.97\sn{-2} & 1.59\sn{-8} \\
\textbf{HyCOP} & \textbf{4.04\sn{-2}} & \textbf{1.59\sn{-8}} \\
\end{tabular}
 
&
 
\centering
\scriptsize
\setlength{\tabcolsep}{3pt}%
\begin{tabular}{@{}l cc cc@{}}
& \multicolumn{2}{c}{2D SWE} & \multicolumn{2}{c}{2D ADR} \\
\cmidrule(lr){2-3}\cmidrule(lr){4-5}
Method & ID & OOD & ID & OOD \\
\midrule
DeepONet              & 3.89\sn{-1} & 5.61\sn{-1} & 1.58\sn{-1} & 2.82\sn{-1} \\
FNO                   & 1.19\sn{-1} & 3.80\sn{-1} & 8.42\sn{-2} & 2.60\sn{-1} \\
PINO                  & 1.17\sn{-1} & 3.83\sn{-1} & 8.42\sn{-2} & 2.36\sn{-1} \\
Loc.\ Int.\ Diff.\ FNO & 7.49\sn{-2} & 3.54\sn{-1} & 3.15\sn{-2} & 1.89\sn{-1} \\
\textbf{HyCOP}    & \textbf{2.40\sn{-2}} & \textbf{5.00\sn{-2}} & \textbf{2.10\sn{-2}} & \textbf{2.87\sn{-2}} \\
\end{tabular}
 
\tabularnewline
\end{tabular}}
 
\vspace{3pt}
\hrule height \heavyrulewidth
 
\vspace{1.0em}
 
 
\noindent
\begin{minipage}[b]{0.46\textwidth}\centering
{\footnotesize\textbf{(c) Long-horizon: 2D SWE}}
\end{minipage}%
\hfill
\begin{minipage}[b]{0.46\textwidth}\centering
{\footnotesize\textbf{(d) Long-horizon: 2D ADR}}
\end{minipage}
 
\vspace{4pt}
\hrule height \heavyrulewidth
\vspace{3pt}
 
\noindent
{\setlength{\tabcolsep}{0pt}%
\begin{tabular}{@{}p{0.47\textwidth}@{\hspace{7pt}}|@{\hspace{7pt}}p{0.43\textwidth}@{}}
 
\centering
\tiny
\setlength{\tabcolsep}{2pt}%
\begin{tabular}{@{}l cc cc@{}}
& \multicolumn{2}{c}{5-step} & \multicolumn{2}{c}{20-step} \\
\cmidrule(lr){2-3}\cmidrule(lr){4-5}
Method & ID & OOD & ID & OOD \\
\midrule
U-Net               & 1.04\sn{0}  & 5.92\sn{-1} & 4.20\sn{0}  & 1.20\sn{0}  \\
LIDFNO              & 4.95\sn{-1} & 6.93\sn{-1} & 5.52\sn{0}  & 1.44\sn{0}  \\
Poseidon            & 3.47\sn{-1} & 6.54\sn{-1} & 8.42\sn{-1} & 9.35\sn{-1} \\
AR-LIDFNO           & 1.72\sn{-1} & 5.33\sn{-1} & 4.86\sn{-1} & 7.55\sn{-1} \\
\textbf{HyCOP}  & \textbf{1.91\sn{-2}} & \textbf{4.54\sn{-2}} & \textbf{6.94\sn{-2}} & \textbf{1.21\sn{-1}} \\
\end{tabular}
 
&
 
\centering
\tiny
\setlength{\tabcolsep}{2pt}%
\begin{tabular}{@{}l cc cc@{}}
& \multicolumn{2}{c}{5-step} & \multicolumn{2}{c}{20-step} \\
\cmidrule(lr){2-3}\cmidrule(lr){4-5}
Method & ID & OOD & ID & OOD \\
\midrule
U-Net               & 1.70\sn{-1} & 5.23\sn{-1} & 2.30\sn{-1} & 1.07\sn{0}  \\
LIDFNO              & 7.78\sn{-2} & 5.65\sn{-1} & 1.51\sn{-1} & 1.30\sn{0}  \\
AR-LIDFNO           & 8.24\sn{-2} & 3.97\sn{-1} & 1.02\sn{-1} & 5.72\sn{-1} \\
\textbf{HyCOP}  & \textbf{1.68\sn{-2}} & \textbf{2.12\sn{-2}} & \textbf{1.96\sn{-2}} & \textbf{3.78\sn{-2}} \\[4pt]
\end{tabular}
 
\tabularnewline
\end{tabular}}
 
\vspace{3pt}
\hrule height \heavyrulewidth
 
\end{table}

\subsubsection{2D Compressible Navier--Stokes (PDEBench)}
\label{sec:exp:ns}

We evaluate single-step prediction on the PDEBench 2D compressible NS benchmark ($M{=}0.1$, $\eta{=}\zeta{=}0.1$, $T{=}0.05$), predicting density, velocity $(V_x, V_y)$, and pressure from initial conditions.
HyCOP's dictionary consists of an RK4 Euler advection primitive and a spectral viscous diffusion primitive---both textbook routines with zero learnable parameters.

Table~\ref{tab:regime_a} reports results.
HyCOP achieves nRMSE $4.04 \times 10^{-2}$, outperforming Strang splitting by 47\%, DeepONet by 49\%, FNO by 89\%, and PINO by 94\%.
The improvement over Strang is the cleanest test of our central claim: same primitives, same dictionary, different schedule---the learned composition policy accounts for the entire gap.
Conservation error (cRMSE) is at machine precision ($1.59 \times 10^{-8}$) for both HyCOP and Strang, inherited directly from the numerical primitives.
No monolithic baseline provides this guarantee; DeepONet's cRMSE is six orders of magnitude worse.

\subsubsection{2D Shallow-Water Equations: OOD and Transfer}
\label{sec:exp:swe}
PDEBench does not support controlled OOD evaluation (parameter extrapolation, long-horizon rollout, boundary shift).
We design a 2D SWE benchmark with explicit OOD splits in physical parameters and initial conditions (details in Appendix~\ref{app:exp_configs}).

\textbf{Fixed-time evaluation.}
Table~\ref{tab:regime_a} (panel b) reports relative $L^2$ error at a single query time.
HyCOP achieves order-of-magnitude OOD improvements: $5.00 \times 10^{-2}$ versus $3.54$--$3.80 \times 10^{-1}$ for the best monolithic baseline on SWE.
Full metric breakdown (fRMSE bands, RMSE, MaxErr, bRMSE, cRMSE) and qualitative comparisons appear in Appendix~\ref{app:2d_swe}.

\textbf{Long-horizon rollout.}
Table~\ref{tab:regime_a} (panel c) shows error at 5-step and 20-step horizons on SWE.
Autoregressive baselines accumulate error rapidly, with most exceeding $\mathrm{Rel.}\ L^2 > 1.0$ at 20 steps OOD---worse than a constant predictor.
HyCOP remains stable because its programs compose at the operator level rather than marching one small step at a time.
Per-horizon trajectory metrics and qualitative rollouts at $1$/$5$/$10$/$20$ steps are in Appendix~\ref{app:2d_swe_trajectory}.

\textbf{Dam-break transfer.}
The source task is 2D SWE with periodic boundaries and smooth ICs; the target is a dam-break with solid-wall boundaries and shock discontinuities---never seen in training.
Table~\ref{tab:dambreak} and Figure~\ref{fig:combined} reveal a progressive story.
HyCOP with periodic primitives already outperforms all baselines zero-shot ($10{\times}$); the policy alone handles the OOD initial conditions and a $100{\times}$ larger grid.
Swapping the boundary primitive for a wall-boundary module---without any dam-break training data---drops error by another $4{\times}$, demonstrating that boundary physics is the dominant residual shift and that modular transfer resolves it.

\begin{figure}[t]
  \centering
  \begin{tikzpicture}[
    flowsmall/.style={-{Latex[length=2mm]}, line width=0.8pt}
  ]
    \node[inner sep=0pt] (img1) {
      \includegraphics[width=0.4\textwidth]{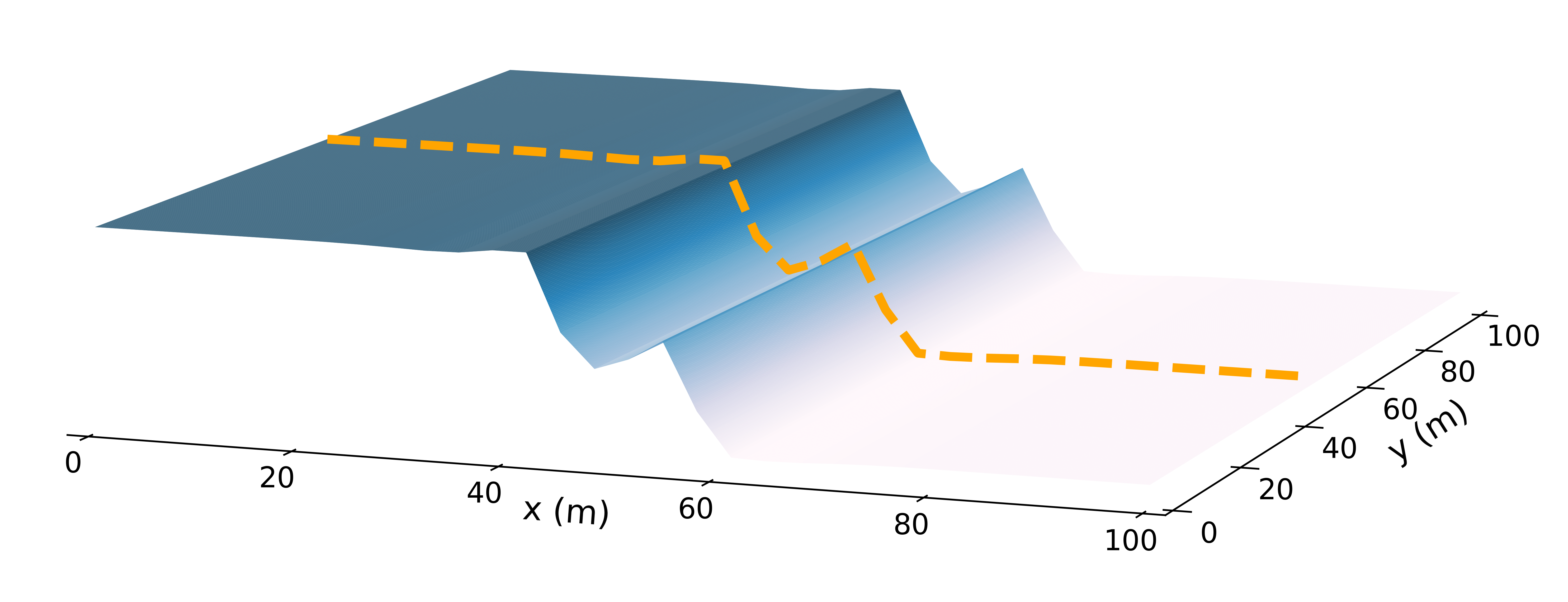}
    };
    \node[inner sep=0pt, right=3pt of img1] (img2) {
      \includegraphics[width=0.55\textwidth]{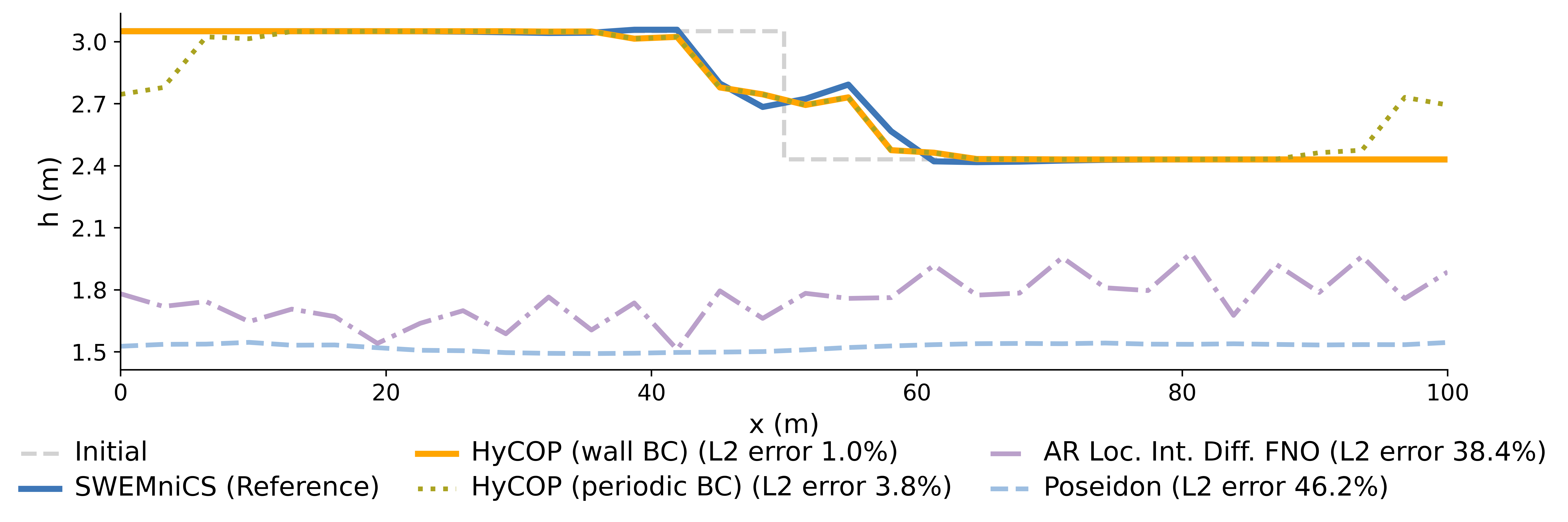}
    };
    \definecolor{goldarrow}{HTML}{feb431}
    \draw[flowsmall, rounded corners=8pt, goldarrow] (2,-0.3) |- (3.0,0.45);
  \end{tikzpicture}
  \caption{\textbf{Dam-break transfer (zero-shot).}
Left: reference SWE height surface with a highlighted cross-section.
Right: 1D slice at $y=50$m comparing baselines vs.\ HyCOP.
Swapping only the boundary primitive (periodic $\rightarrow$ wall) yields a sharp improvement and resolves the shock region (gray box), indicating boundary physics as the dominant shift.}
  \label{fig:combined}
\end{figure}

\begin{table}[t]
\caption{\textbf{SWE $\to$ dam-break transfer} (zero-shot).
Relative $L^2$ error ($\downarrow$) for height $h$, with inference time per sample.
All models trained only on smooth SWE with periodic boundaries.
Swapping the boundary primitive (periodic $\to$ wall) requires no retraining.}
\label{tab:dambreak}
\centering
\footnotesize
\setlength{\tabcolsep}{5pt}
\begin{tabular}{lcccc}
\toprule
& AR-Loc.\ Int.\ Diff.\ FNO & Poseidon & HyCOP (periodic) & \textbf{HyCOP (wall swap)} \\
\midrule
Rel.\ $L^2$ ($h$) & 3.84\sn{-1} & 4.62\sn{-1} & 3.82\sn{-2} & \textbf{9.92\sn{-3}} \\
Time (s)         & 0.345        & 3.391        & 1.397             & 1.360 \\
\bottomrule
\end{tabular}
\end{table}

\textbf{Computational cost.}
HyCOP achieves these results with $5 \times 10^5$ training forward passes---$25{\times}$ fewer than U-Net/AR baselines---and over $10{\times}$ shorter wall-clock training time.
At inference, HyCOP is $2{\times}$ faster than Poseidon at 20-step horizons (Table~\ref{tab:ablations}).

\subsubsection{2D Advection--Diffusion--Reaction}
\label{sec:exp:adr}
We run the same fixed-time and long-horizon experiments on 2D advection--diffusion--reaction (Fisher--KPP); HyCOP shows the same OOD-robustness pattern with comparable order-of-magnitude improvements (Table~\ref{tab:regime_a}, panels b and d). Full setup, metric breakdown, qualitative comparisons, and trajectory rollouts are in Appendices~\ref{app:2d_adr} and~\ref{app:2d_adr_trajectory}.

\subsection{Regime B: Adaptation and Discovery}
\label{sec:exp:regime_b}

When the physics is incomplete or components come from different sources, HyCOP's compositional structure enables systematic model refinement.
We demonstrate both paths introduced in \S\ref{sec:intro} on the AD$\to$ADR benchmark.


\begin{table}[H]
\caption{\textbf{Regime~B: Adaptation and discovery on 2D ADR.}
Relative $L^2$ error ($\downarrow$). Setup details for all configurations in Appendix~\ref{app:ad_adr}.
\emph{Top}: monolithic baselines trained end-to-end on full ADR data, and HyCOP with the complete numerical dictionary (Regime~A reference).
\emph{Middle}: two adaptation paths from incomplete or heterogeneous dictionaries---neither retrains any existing primitive; only the small policy (${\sim}$50--100 parameters) is (re)learned.
\emph{Bottom}: ablation replacing all numerical primitives with per-process FNOs.}
\label{tab:regime_b}
\centering
\footnotesize
\setlength{\tabcolsep}{4pt}
\begin{tabular}{@{} c @{\hspace{3pt}} l l cc @{}}
\toprule
& Method & Dictionary & ID & OOD \\

\midrule
& \multicolumn{4}{l}{\emph{Baselines trained on full ADR data}} \\
& DeepONet                & ---                              & 1.58\sn{-1} & 2.82\sn{-1} \\
& FNO                     & ---                              & 8.42\sn{-2} & 2.60\sn{-1} \\
& PINO                    & ---                              & 8.42\sn{-2} & 2.36\sn{-1} \\
& Loc.\ Int.\ Diff.\ FNO & ---                              & 3.15\sn{-2} & 1.89\sn{-1} \\
& \textbf{HyCOP}      & \{Adv, Diff, React\}             & \textbf{2.10\sn{-2}} & \textbf{2.87\sn{-2}} \\

\midrule
\multirow{6}{*}{\rotatebox[origin=c]{90}{\scriptsize\textbf{Missing physics}}}
& \multicolumn{4}{l}{\emph{Path~(a): compose--diagnose--enrich (unknown missing physics)}} \\
& \quad HyCOP (AD pretrain)     & \{Adv, Diff\}            & \multicolumn{2}{c}{1.81\sn{-1}} \\
& \quad\textbf{HyCOP + residual} & \{Adv, Diff, Resid.\}   & \multicolumn{2}{c}{\textbf{4.42\sn{-2}}} \\
\cmidrule(l){2-5}
& \multicolumn{4}{l}{\emph{Path~(b): hybrid dictionary (known missing physics)}} \\
& \quad HyCOP-Hyb               & \{FNO-AD, React\} & 3.77\sn{-2} & 1.36\sn{-1} \\

\midrule
& \multicolumn{4}{l}{\emph{Ablation: fully learned primitives (per-process FNOs)}} \\
& \quad HyCOP-Learned           & \{FNO-Adv, FNO-Diff, FNO-React\} & 2.38\sn{-2} & 3.90\sn{-1} \\

\bottomrule
\end{tabular}
\end{table}

\subsubsection{Two Paths to a Hybrid Dictionary}
\label{sec:exp:transfer}

The source task is 2D advection--diffusion (AD); the target is ADR with Fisher--KPP reaction.
We compare two scientific workflows arriving at hybrid dictionaries:

\textbf{Path (a): Unknown missing physics (compose--diagnose--enrich).}
The scientist starts with an AD dictionary $\{\mathcal{O}_{\mathrm{adv}}, \mathcal{O}_{\mathrm{diff}}\}$ and observes structured failure on ADR data: errors concentrate where reaction dominates, while advection/diffusion regions remain accurate.
A residual closure (UNO, trained on 120 ADR trajectories) is added as $\mathcal{O}_{\mathrm{res}}$, and only the policy is relearned over the enriched dictionary.
Error drops from $0.181$ to $0.044$ (Table~\ref{tab:regime_b}), and the resulting programs invoke the residual precisely in high-reaction regions (Figure~\ref{fig:adr_residual}).

\textbf{Path (b): Known missing physics, heterogeneous primitives.}
A lab already has a pretrained FNO-FiLM surrogate for AD dynamics (trained on AD data, frozen).
The missing reaction mechanism is well-characterized, so a textbook numerical reaction solver is added.
HyCOP learns a two-operator policy over the hybrid dictionary $\{\text{FNO-AD}, \mathcal{O}_{\mathrm{react}}^{\mathrm{num}}\}$---no retraining of either component.

Table~\ref{tab:regime_b} reveals a diagnostic gradient.
HyCOP (all numerical, Rel.\ $L^2$ OOD $= 2.87 \times 10^{-2}$) barely degrades because every primitive generalizes exactly; the only OOD vulnerability is the policy schedule.
HyCOP-Hyb (Rel.\ $L^2$ OOD $= 1.36 \times 10^{-1}$) degrades more because the FNO-FiLM AD primitive does not generalize as robustly---precisely the primitive-error term in the decomposition~\eqref{eq:error_decomp}.
Yet HyCOP-Hyb still outperforms Loc.\ Int.\ Diff.\ FNO ($1.89 \times 10^{-1}$), a monolithic model trained end-to-end on the full ADR system.
The error decomposition serves as a diagnostic: the scientist can identify the FNO-FiLM as the limiting component and decide whether to replace it with a numerical primitive (recovering HyCOP's robustness) or invest in a more robust learned surrogate.

Both paths arrive at hybrid dictionaries---numerical and learned components orchestrated by a small policy---from different scientific starting points.
Neither requires retraining any existing primitive; only the small policy (${\sim}$50--100 parameters) is (re)learned.

\section{Ablations and Analysis}
\label{sec:ablations}
We systematically test each component of HyCOP (Table~\ref{tab:ablations}).
\textbf{Dictionary robustness.} Adding a redundant reaction primitive to the SWE dictionary does not degrade accuracy: the policy suppresses it to $3.5\%$ time allocation ($\mathrm{Rel.}\ L^2 = 2.30 \times 10^{-2}$ vs.\ $2.40 \times 10^{-2}$ with the correct dictionary). When reaction is missing from the AD dictionary and the target is ADR, zero-shot error is $0.181$; dictionary enrichment recovers to $0.044$ (\S\ref{sec:exp:transfer}).
\textbf{ES sensitivity.} A grid sweep over population size $M \in \{100,250,500\}$ and noise $\sigma \in \{0.01,0.02,0.05\}$ on 2D SWE shows a broad plateau; even $M{=}100$ outperforms all monolithic baselines OOD ($3.5 \times 10^{-2}$ vs.\ $> 3.5 \times 10^{-1}$).
\textbf{Resolution transfer.} A policy trained at $32{\times}32$ transfers zero-shot to $128{\times}128$ on 2D SWE with modest degradation ($2.40 \to 3.37 \times 10^{-2}$); the policy conditions on dimensionless features, so only primitive cost scales with resolution.
\textbf{Feature ablation.} Conditioning on dimensionless regime features (P\'eclet, variance, gradient variance) reduces OOD error by $32\%$ relative to raw IC inputs, confirming their role in regime-aware scheduling.
\textbf{Fully Learned primitives.} Replacing all numerical primitives with per-process FNOs on 2D ADR isolates the primitive-error term of the error decomposition (\S\ref{sec:theory}): OOD error tracks primitive quality as the theory predicts (Table~\ref{tab:regime_b}), validating the decomposition as a diagnostic.
\textbf{Chaotic validation.} On the chaotic Kuramoto--Sivashinsky equation---where our stable-split-system assumptions may not hold---HyCOP maintains low spectral error and KL divergence of the invariant measure across ID and OOD settings (Table~\ref{tab:ablations}, panel e; Appendix~\ref{app:1d_chaos}).
\textbf{Policy interpretation.} On 2D ADR, reaction-primitive allocation increases monotonically with Damk\"ohler number while advection allocation tracks P\'eclet, consistent with dominant-process physics (Figure~\ref{fig:pipeline}c).



\section{Discussion}
\label{sec:discussion}

\textbf{When to use HyCOP.}
HyCOP is designed for structured PDEs where a natural process decomposition exists---advection--diffusion--reaction, shallow-water, Navier--Stokes, and many multi-physics systems.
For PDEs with no natural decomposition, monolithic surrogates remain the right tool.
When partial knowledge exists, hybrid dictionaries allow the framework to integrate whatever is available---numerical solvers for well-characterized processes, pretrained surrogates for expensive or poorly resolved ones, learned closures for unknown mechanisms.
\textbf{Limitations.}
The primitive dictionary requires domain knowledge to specify, though dictionary robustness (\S\ref{sec:ablations}) shows the framework tolerates redundant or missing entries.
Current 2D experiments use resolutions up to $128{\times}128$ (zero-shot from $32{\times}32$); scalability to production resolutions is an engineering question---the policy is resolution-agnostic and only primitive costs scale with grid size.
HyCOP-Hyb's OOD degradation (\S\ref{sec:exp:transfer}) shows that learned primitives are the framework's weakest link under shift; improving their robustness is an active research direction.
\textbf{Broader significance.}
Scientific computing has always been compositional---finite-element assembly, climate-model coupling, multi-physics splitting---and HyCOP provides a learning-theoretic foundation for this practice with formal guarantees on the resulting error.
Recent V\&V frameworks~\citep{jakeman2026vvsciml} recognize hybrid SciML models as a first-class category; HyCOP provides one concrete instantiation with built-in diagnostics via the error decomposition.
\textbf{Future work.}
Three directions: (1)~fully learned primitive dictionaries removing the need for numerical solvers, with FiLM-conditioned or DeepONet-factored architectures for time-queried primitives; (2)~application to climate/weather foundation models as compositional backbones, where each component (radiation, convection, dynamics) is a dictionary entry; (3)~connections to program synthesis---HyCOP's programs are discrete compositions amenable to symbolic search, potentially enabling automated discovery of splitting schemes.

\bibliography{related_works, proof}
\bibliographystyle{plainnat}

\newpage
\appendix
\onecolumn


\begin{table}[t]
\caption{\textbf{Positioning.} $^\dagger$Depends on whether the fixed schedule suits the target regime.}
\label{tab:positioning}
\centering
\footnotesize
\setlength{\tabcolsep}{3pt}
\begin{tabular}{lccc}
\toprule
 & Classical splitting & Neural operators & \textbf{HyCOP} \\
\midrule
Composition & Fixed schedule & None (monolithic) & \textbf{Learned policy} \\
Regime adaptivity & None & Learned weights & \textbf{Conditioned $\pi_\theta$} \\
Dictionary type & Numerical only & Learned only & \textbf{Hybrid (any)} \\
Process interpretability & Full & None & \textbf{Full (programs)} \\
Failure diagnosis & Manual & None & \textbf{Process-level} \\
Modular transfer & Manual & Retrain & \textbf{Dict.\ swap/enrich} \\
OOD robustness & Regime-dep.$^\dagger$ & Limited & \textbf{Strong} \\
\bottomrule
\end{tabular}
\end{table}

\section{Theory \& Proofs}
\label{sec:proof}

In this appendix we state assumptions and prove the guarantees used in the main text.
Under standard well-posedness and splitting regularity conditions, (i) finite compositions of (sub)flows are expressive enough to approximate the solution operator on compact query sets, and (ii) HyCOP error decomposes into composition (splitting) error plus primitive-implementation error.
We do not develop new splitting theory; we adapt classical tools to our hypothesis class.

\subsection{Function spaces}

We work on the $d$-dimensional torus $\Omega = \mathbb{T}^d = (\mathbb{R}/\mathbb{Z})^d$ with periodic boundary conditions. The main function spaces are:

\begin{itemize}[leftmargin=2em]
    \item $L^2(\Omega)$: square-integrable functions with norm:
    \[
        \|u\|_{L^2} = \left(\int_{\Omega} |u(x)|^2 \, dx\right)^{1/2}.
    \]
    \item $H^s(\Omega)$: Sobolev space of order $s \geq 0$ with Fourier norm:
    \[
        \|u\|_{H^s} = \left(\sum_{k \in \mathbb{Z}^d} (1 + |k|^2)^s |\hat{u}_k|^2\right)^{1/2},
    \]
    where $\hat{u}_k$ denotes the Fourier coefficient of $u$ at mode $k$.
    \item $C([0,T]; X)$: continuous functions from $[0,T]$ to a Banach space $X$.
\end{itemize}

For $s > s'$, we have the continuous embedding $H^s(\Omega) \hookrightarrow H^{s'}(\Omega)$~\citep{Adams2003}. Intuitively, higher $s$ corresponds to more spatial regularity of the solution.

\subsection{Assumptions}

We first formalize the class of PDEs and splits we study.

\begin{definition}[Stable Split System]\label{def:stable_split}
Consider the parameterized PDE:
\[
    \frac{\partial u}{\partial t} = \mathcal{F}(u,\mu),
\]
defined on $\mathbb{T}^d$ with initial condition $u(0,\cdot) = u_0(\cdot)$. Suppose that the differential operator $\mathcal{F}$ decomposes into atomic operators:
\[
    \mathcal{F} = \sum_{i=1}^n \mathcal{F}_i,
\]
and that for each $t \in [0,T]$ and parameter $\mu$, the solution $u(t,\mu) \in H^s(\Omega)$ varies continuously with respect to $t$ and $\mu$. We say this is a \emph{stable split system} if:

\begin{enumerate}[label=(A\arabic*)]
    \item (Well-posedness) For $u_0 \in H^s(\Omega)$ with $s \geq s_0$, the full equation admits a unique solution:
    \[
        u \in C([0,T]; H^s(\Omega)).
    \]
    \item (Sub-problem well-posedness) Each sub-problem:
    \[
        \frac{\partial v}{\partial t} = \mathcal{F}_i(v,\mu),
    \]
    is well-posed on $[0,T]$ for every $1 \leq i \leq n$.
    \item (Regularity preservation) If $v_0 \in H^s(\Omega)$, then $\Phi^{(i)}_\tau(v_0) \in H^s(\Omega)$ for any $\tau \in (0,T)$, where $\Phi^{(i)}_\tau$ is the flow map generated by $\mathcal{F}_i$.
\end{enumerate}
\end{definition}

Assumptions (A1)--(A3) say that the full PDE and each of its atomic components are classically well-posed and do not destroy Sobolev regularity over the time horizon of interest.

We next collect regularity assumptions needed to control splitting errors.

\begin{assumption}[Regularity for error analysis]\label{ass:regularity}
We assume:

\begin{enumerate}[label=(R\arabic*)]
    \item (Commutator bounds) For $s \geq s_1$,
    \[
        \|[\mathcal{F}_i, \mathcal{F}_j](u)\|_{L^2} \leq C_{ij} \|u\|_{H^s},
    \]
    for all $i,j$ and all $u \in H^s$, with constants $C_{ij}$ independent of $u$.
    \item (Stability of exact flows) There exists $\omega > 0$ such that for all $|t| \leq T$,
    \[
        \|\Phi^{(i)}_t(u)\|_{H^s} \leq e^{\omega|t|} \|u\|_{H^s},
    \]
    for every $i$.
    \item (Higher commutators) Nested commutators up to depth $p$ map $H^{s_p} \to L^2$ boundedly.
\end{enumerate}

Here the commutator is:
\[
    [\mathcal{F}_i, \mathcal{F}_j](u) = D\mathcal{F}_j(u)[\mathcal{F}_i(u)] - D\mathcal{F}_i(u)[\mathcal{F}_j(u)].
\]
These assumptions are standard in the analysis of splitting methods~\citep{HairerEtAl2006, Jahnke2000}.
\end{assumption}

\begin{remark}[Regularity versus order]
The Sobolev index $s_p$ required for order-$p$ splitting increases with $p$. Typical values are $s_2 \approx s_0 + 2$ for Strang splitting and $s_4 \approx s_0 + 4$ for fourth-order methods~\citep{Jahnke2000, Thalhammer2008}. Intuitively, higher-order schemes require more derivatives to be controlled in the commutator expansions.
\end{remark}

We also need assumptions on the approximate sub-flows used as primitives.

\begin{assumption}[Primitive-implementation requirements]\label{ass:approx}
Let $\widehat{\Phi}^{(i)}_t$ be an implemented primitive approximating the exact subflow $\Phi^{(i)}_t$.
Assume:
\begin{enumerate}[label=(S\arabic*)]
\item (Local accuracy) $\|\widehat{\Phi}^{(i)}_h(v)-\Phi^{(i)}_h(v)\| \le \delta_i h^{q+1}\|v\|_{H^s}$ for all $v\in H^s$.
\item (Lipschitz stability) $\|\widehat{\Phi}^{(i)}_t(u)-\widehat{\Phi}^{(i)}_t(v)\|_{H^s}\le e^{\tilde\omega |t|}\|u-v\|_{H^s}$ for $|t|\le T$.
\end{enumerate}
\end{assumption}

\begin{remark}
Assumption~\ref{ass:approx} (S2) is automatically satisfied by classical stable time-stepping schemes (e.g., A-stable Runge--Kutta, spectral methods) under appropriate CFL conditions. When neural operators serve as sub-solvers, stability must be enforced by architectural and training choices (e.g., spectral normalization, contractive updates).
\end{remark}

\subsection{Existence of a finite split sequence}
As a first step, we establish an existence result that justifies restricting attention to finite-step compositions and allows us to define a concrete search space and hypothesis class.

Suppose we are given primitive operators $\mathcal{O}_1,\dots,\mathcal{O}_n$ acting on a suitable Sobolev space $\mathbb{H}^s$ (see Appendix for details). Let $\mathcal{F}_i$ and $\Phi^{(i)}_{\tau}$ denote the corresponding differential operator and flow map of the $i$th primitive.

\begin{definition}[Finite composite flows]
For $k\ge 1$, a $k$-step composite flow is
\[
\Psi^{(k)}=\Phi^{(j_k)}_{\tau_k}\circ\cdots\circ\Phi^{(j_1)}_{\tau_1},
\qquad j_r\in\{1,\dots,n\},\ \tau_r\in\mathbb{R}.
\]
Let $\mathcal{C}_k$ be the set of all such maps and $\mathcal{C}=\bigcup_{k\ge 0}\mathcal{C}_k$.
Replacing $\Phi^{(i)}$ by implemented primitives $\widehat{\Phi}^{(i)}$ defines $\widehat{\mathcal{C}}$.
\end{definition}

Our first result shows that the hypothesis class
of HyCOP, which consists of finite composite flows built
from primitives $\mathcal{O}_i$, is dense (in $L^2$) in the space of solution
operators for stable split systems (Definition~\ref{def:stable_split}) on compact input sets.

\begin{theorem}[Expressivity of finite composite flows]
\label{thm:existence}
    Let $\Lambda$ denote a compact subset of the parameter--initial--boundary space, and let $t \in [0,T]$. Suppose that for each $x \in \Lambda$ the PDE:
    \[
        \frac{\partial u}{\partial t} = \mathcal{F}(u,\mu),
    \]
    with input $x = (\mu,u_0,b,\Omega)$ defines a stable split system in the sense of Definition~\ref{def:stable_split}, and that the differential operator admits a decomposition
   $ 
        \mathcal{F} = \mathcal{F}_1 + \cdots + \mathcal{F}_n.
   $
    Then, for any $\varepsilon > 0$ there exists a finite composite flow operator $\Psi \in \mathcal{C}$ such that:
    \[
        \sup_{x \in \Lambda} \big\| u(t;x) - \Psi(x,t) \big\|_{L^2} < \varepsilon,
    \]
    where $u(t;x)$ denotes the exact solution at time $t$ for input $x$.
\end{theorem}

\subsection{Definitions}

We recall the parameterization of splitting strategies and the loss functional.

\begin{definition}[Splitting strategy]\label{def:strategy}
A splitting strategy of maximum depth $K_{\max}$ is:
\[
    \theta = (k, (i_1,\dots,i_k), (\tau_1,\dots,\tau_k)) \in \Theta,
\]
specifying the composition length $k \leq K_{\max}$, operator indices $i_j \in \{1,\ldots,n\}$, and time durations $\tau_j > 0$. The space $\Theta$ embeds into $\mathbb{R}^{K_{\max}(n+1)+1}$ via the continuous relaxation:
\[
    \iota: \theta \mapsto \left(k, \{\sigma_j\}_{j=1}^{K_{\max}}, \{\tau_j\}_{j=1}^{K_{\max}}\right),
\]
where each $\sigma_j \in \Delta^{n-1} \subset \mathbb{R}^n$ is a probability distribution over operators (e.g., via a softmax $\sigma_{jl} = e^{z_{jl}}/\sum_{\ell} e^{z_{j\ell}}$). Discrete operator selection can be recovered by $i_j = \arg\max_l \sigma_{jl}$.
\end{definition}
We allow $\tau\in\mathbb{R}$ in $\mathcal{C}$ for algebraic convenience; HyCOP strategies restrict to $\tau_j>0$.

\begin{remark}[Operator algebra viewpoint]
The space $\mathcal{C}$ of finite composite flows generated by the primitives forms a monoid under composition. A neural network outputting elements of $\Theta$ therefore acts as a \emph{policy} that selects an element of this operator algebra, synthesizing algorithms from primitive sub-flows $\{\tilde{\Phi}^{(i)}_\tau\}$.
\end{remark}

\begin{definition}[Loss]\label{def:loss}
For a query $(x,t)$ with $x=(\mu,u_0,b,\Omega)$ and a program $\theta$,
\[
\mathcal{L}(x,t;\theta)=\|\mathcal{S}(t;x)-\widehat{\Psi}_\theta(x,t)\|_{L^2},
\qquad
\widehat{\Psi}_\theta(x,t)=\widehat{\Phi}^{(i_k)}_{\tau_k}\circ\cdots\circ\widehat{\Phi}^{(i_1)}_{\tau_1}(u_0).
\]
\end{definition}

\subsection{Lie derivative formalism}

We briefly recall the Lie-derivative formalism that linearizes the flow at the level of observables. This is classical in geometric numerical integration~\citep{HairerEtAl2006, BlanesEtAl2024}.

\begin{definition}[Lie derivative]
Let $\mathcal{F}$ be a (possibly nonlinear) vector field on a Banach space $X$, with flow $\Phi^{(\mathcal{F})}_t$. For a smooth observable $g: X \to \mathbb{R}$, the Lie derivative $\mathcal{L}_{\mathcal{F}}$ acts by:
\[
    (\mathcal{L}_{\mathcal{F}} g)(u)
    = Dg(u)[\mathcal{F}(u)]
    = \left.\frac{d}{dt}\right|_{t=0} g(\Phi^{(\mathcal{F})}_t(u)).
\]
\end{definition}

\begin{proposition}[Lie transformation]\label{prop:lie}
For smooth observables $g$ and regular initial data $u$,
\[
    g(\Phi^{(\mathcal{F})}_t(u)) = \left(e^{t \mathcal{L}_{\mathcal{F}}} g\right)(u),
\]
where the exponential is the strongly convergent series,
\[
    e^{t\mathcal{L}_{\mathcal{F}}}
    = \sum_{k=0}^\infty \frac{t^k}{k!} \mathcal{L}_{\mathcal{F}}^k.
\]
\end{proposition}

\begin{proof}
Fix $u$ and define $\psi(t) = g(\Phi^{(\mathcal{F})}_t(u))$. By the chain rule,
\[
    \frac{d\psi}{dt}
    = Dg(\Phi^{(\mathcal{F})}_t(u))[\mathcal{F}(\Phi^{(\mathcal{F})}_t(u))]
    = (\mathcal{L}_{\mathcal{F}} g)(\Phi^{(\mathcal{F})}_t(u)).
\]
Viewed as a function of $t$, this is the linear ODE:
\[
    \frac{d}{dt} \psi(t) = (\mathcal{L}_{\mathcal{F}} \psi)(t), \qquad \psi(0) = g(u),
\]
whose solution is $\psi(t) = e^{t\mathcal{L}_{\mathcal{F}}} g(u)$ by the usual exponential representation of linear flows.
\end{proof}

\begin{remark}[Linearization principle]
Even when $\mathcal{F}$ is nonlinear on state space $X$, the Lie derivative $\mathcal{L}_{\mathcal{F}}$ is a \emph{linear} operator on observables. This is the key reason why Baker--Campbell--Hausdorff (BCH) analysis of splitting methods for nonlinear PDEs mirrors the linear case~\citep{BlanesEtAl2024, McLachlanQuispel2002}.
\end{remark}

\begin{definition}[Lie bracket]
The Lie bracket $[\mathcal{F}, \mathcal{G}]$ of vector fields satisfies:
\[
    \mathcal{L}_{[\mathcal{F},\mathcal{G}]} = [\mathcal{L}_{\mathcal{F}}, \mathcal{L}_{\mathcal{G}}]
    = \mathcal{L}_{\mathcal{F}}\mathcal{L}_{\mathcal{G}} - \mathcal{L}_{\mathcal{G}}\mathcal{L}_{\mathcal{F}}.
\]
Explicitly~\citep{HairerEtAl2006},
\[
    [\mathcal{F}, \mathcal{G}](u)
    = D\mathcal{G}(u)[\mathcal{F}(u)] - D\mathcal{F}(u)[\mathcal{G}(u)].
\]
\end{definition}

\subsection{Proof of Theorem~\ref{thm:existence} (Expressivity of finite composite flows)}

We sketch the proof of the existence theorem using the Lie--Trotter product formula. Recall that $\mathcal{F} = \sum_{i=1}^n \mathcal{F}_i$ and denote by $\Phi^{(\mathcal{F})}_t$ the exact flow of the full PDE.

For $N \in \mathbb{N}$, define the $N$-fold Lie--Trotter composition:
\[
    \Psi_N
    = \left(\Phi^{(n)}_{T/N} \circ \Phi^{(n-1)}_{T/N} \circ \cdots \circ \Phi^{(1)}_{T/N}\right)^N
    \in \mathcal{C}_{Nn}.
\]
The theorem asserts that $\Psi_N$ converges to the exact flow as $N \to \infty$, uniformly over inputs in a compact set.

\paragraph{Step 1: Local error.}
For a small step $h = T/N$, the BCH formula~\citep{HairerEtAl2006} applied at the level of Lie derivatives gives:
\[
    e^{h\mathcal{L}_{\mathcal{F}_n}} \cdots e^{h\mathcal{L}_{\mathcal{F}_1}}
    = \exp\left(
        h\sum_{i=1}^n \mathcal{L}_{\mathcal{F}_i}
        + \frac{h^2}{2}\sum_{i < j} [\mathcal{L}_{\mathcal{F}_i}, \mathcal{L}_{\mathcal{F}_j}]
        + \mathcal{O}(h^3)
    \right).
\]
By Assumption~\ref{ass:regularity} (R1), the commutator term is bounded, and standard splitting analysis~\citep{BlanesEtAl2024} yields the one-step local error bound:
\[
    \|\Phi^{(\mathcal{F})}_h(v) - \Phi^{(n)}_h \circ \cdots \circ \Phi^{(1)}_h(v)\|_{L^2}
    \leq C h^2 \|v\|_{H^{s_1}},
\]
for a constant $C$ independent of $v$.

\paragraph{Step 2: Global error.}
Using the telescoping argument known as \emph{Lady Windermere's fan}~\citep{HairerEtAl2006} together with stability (R2), we obtain:
\[
    \|\Phi^{(\mathcal{F})}_T(u_0) - \Psi_N(u_0)\|_{L^2}
    \leq N \cdot e^{\omega T} \cdot C h^2 \cdot \sup_{t \in [0,T]} \|u(t)\|_{H^{s_1}}
    = \mathcal{O}(h) = \mathcal{O}\big(1/N\big).
\]
In other words, the Lie--Trotter scheme is first-order accurate globally.

\paragraph{Step 3: Uniformity and conclusion.}
Given $\varepsilon > 0$, choose $N$ large enough that the global error is less than $\varepsilon$ for all admissible initial data $u_0$ and parameters $\mu$ in a compact set $\Lambda$; this is possible because the constants above can be chosen uniformly on compact parameter sets for a stable split system. Then $\Psi_N \in \mathcal{C}$ satisfies the desired approximation property, which completes the proof. \qed

\begin{corollary}[Error rate for higher-order splittings]
Order-$p$ splitting schemes (e.g., Strang, Yoshida) achieve global error $\mathcal{O}(h^p)$ with $h = T/N$~\citep{HairerEtAl2006, McLachlanQuispel2002}. Classical examples include Strang splitting ($p=2$) and Yoshida's fourth-order composition ($p=4$).
\end{corollary}

\subsection{Error decomposition for approximate sub-flows}

We now incorporate the fact that HyCOP uses approximate sub-flows $\widehat{\Phi}^{(i)}$. We suppress the fixed time $t$ to make the notations more compact.

\begin{theorem}[Error decomposition]\label{thm:error_decomp_appendix}
Let Assumptions~\ref{ass:regularity}--\ref{ass:approx} hold. For a strategy $\theta$ with $k$ compositions and step sizes $\{\tau_j\}$, the total error decomposes as:
\[
    \|\mathcal{S}(x) - \widehat{\Psi}_\theta(x)\|_{L^2}
    \leq \underbrace{\|\mathcal{S}(x) - \Psi_\theta(x)\|_{L^2}}_{\text{splitting error}}
       + \underbrace{\|\Psi_\theta(x) - \widehat{\Psi}_\theta(x)\|_{L^2}}_{\text{sub-solver error}},
\]
where $\Psi_\theta = \Phi^{(i_k)}_{\tau_k} \circ \cdots \circ \Phi^{(i_1)}_{\tau_1}$ is the exact composite flow. Moreover, the sub-solver error satisfies:
\[
    \|\Psi_\theta(x) - \widehat{\Psi}_\theta(x)\|_{L^2}
    \leq C_{\emph{sol}} \, e^{\bar{\omega} T}
        \left(\sum_{j=1}^k \tau_j^{q+1}\right) \|u_0\|_{H^s},
\]
where $\bar{\omega} = \max(\omega, \tilde{\omega})$ and $C_{\emph{sol}} = \max_i \delta_i$.
\end{theorem}

\begin{proof}
The decomposition itself is just the triangle inequality. To bound the sub-solver term, we compare the exact and approximate compositions step by step.

Define intermediate compositions in which we gradually replace approximate sub-flows with exact ones:
\[
    w^{(j)}
    = \Phi^{(i_k)}_{\tau_k} \circ \cdots \circ \Phi^{(i_{j+1})}_{\tau_{j+1}}
      \circ \Phi^{(i_j)}_{\tau_j}
      \circ \widehat{\Phi}^{(i_{j-1})}_{\tau_{j-1}} \circ \cdots \circ \widehat{\Phi}^{(i_1)}_{\tau_1}(u_0),
\]
so that $w^{(0)} = \widehat{\Psi}_\theta(u_0)$ (all approximate) and $w^{(k)} = \Psi_\theta(u_0)$ (all exact). Then,
\[
    \Psi_\theta(u_0) - \widehat{\Psi}_\theta(u_0)
    = \sum_{j=1}^k \big( w^{(j)} - w^{(j-1)} \big).
\]

Let $t_j = \sum_{l=1}^{j} \tau_l$ be the cumulative time after $j$ steps, and define the intermediate state:
\[
    \tilde{v}_{j-1}
    = \widehat{\Phi}^{(i_{j-1})}_{\tau_{j-1}} \circ \cdots \circ \widehat{\Phi}^{(i_1)}_{\tau_1}(u_0).
\]
Each difference term can be written as:
\[
    w^{(j)} - w^{(j-1)}
    = \Phi^{(i_k)}_{\tau_k} \circ \cdots \circ \Phi^{(i_{j+1})}_{\tau_{j+1}}
      \Big( \Phi^{(i_j)}_{\tau_j}(\tilde{v}_{j-1})
           - \widehat{\Phi}^{(i_j)}_{\tau_j}(\tilde{v}_{j-1}) \Big).
\]

We now apply three estimates:

\emph{(i) Stability of exact flows (R2).} The outer composition of exact flows from step $j+1$ to $k$ satisfies:
\[
    \|\Phi^{(i_k)}_{\tau_k} \circ \cdots \circ \Phi^{(i_{j+1})}_{\tau_{j+1}}(a)
     - \Phi^{(i_k)}_{\tau_k} \circ \cdots \circ \Phi^{(i_{j+1})}_{\tau_{j+1}}(b)\|_{L^2}
    \leq e^{\omega(T-t_j)} \|a - b\|_{L^2}.
\]

\emph{(ii) Accuracy of approximate sub-flows (S1).} The local solver error at step $j$ satisfies:
\[
    \|\Phi^{(i_j)}_{\tau_j}(\tilde{v}_{j-1})
      - \widehat{\Phi}^{(i_j)}_{\tau_j}(\tilde{v}_{j-1})\|_{L^2}
    \leq \delta_{i_j} \tau_j^{q+1} \|\tilde{v}_{j-1}\|_{H^s}.
\]

\emph{(iii) Stability of approximate sub-flows (S2).} Iterating Assumption~\ref{ass:approx}~(S2) yields:
\[
    \|\tilde{v}_{j-1}\|_{H^s}
    \leq e^{\tilde{\omega} t_{j-1}} \|u_0\|_{H^s}.
\]

Combining (i)--(iii), each term in the telescoping sum is bounded by:
\[
    \|w^{(j)} - w^{(j-1)}\|_{L^2}
    \leq \delta_{i_j} \tau_j^{q+1}
         e^{\omega(T-t_j) + \tilde{\omega} t_{j-1}} \|u_0\|_{H^s}.
\]
Since $\omega(T - t_j) + \tilde{\omega} t_{j-1} \leq \bar{\omega} T$ with $\bar{\omega} = \max(\omega, \tilde{\omega})$ and $\delta_{i_j} \leq C_{\text{sol}} = \max_i \delta_i$, we obtain:
\[
    \|\Psi_\theta(u_0) - \widehat{\Psi}_\theta(u_0)\|_{L^2}
    \leq \sum_{j=1}^k \|w^{(j)} - w^{(j-1)}\|_{L^2}
    \leq C_{\text{sol}} e^{\bar{\omega} T}
        \left(\sum_{j=1}^k \tau_j^{q+1}\right) \|u_0\|_{H^s},
\]
which is the desired bound.
\end{proof}

\begin{corollary}[Total error rate]\label{cor:total_error}
For an order-$p$ splitting scheme with uniform step size $h = T/N$ and order-$q$ sub-solvers,
\[
    \|\mathcal{S}(x) - \widehat{\Psi}_N(x)\|_{L^2}
    \leq C_{\emph{split}} h^p + C_{\emph{sol}} h^q.
\]
When $q \geq p$, the splitting error dominates and the total error is $\mathcal{O}(h^p)$.
\end{corollary}

\begin{remark}[Interpretation]
Theorem~\ref{thm:existence} guarantees that the \emph{ideal} splitting error (first term) can be made arbitrarily small by choosing appropriate compositions of \emph{exact} sub-flows. Theorem~\ref{thm:error_decomp_appendix} shows how this ideal error is perturbed when sub-flows are replaced by approximate primitives. Together, they justify using finite composite flows of primitives as a universal function class for surrogate solution operators.
\end{remark}

\subsection{Proof of Theorem~\ref{thm:optimal_policy} (Existence of $\varepsilon$-optimal policy)}

We now justify the existence of an $\varepsilon$-optimal splitting policy and its approximation by a neural network.

\paragraph{Part 1: Continuity of $\mathcal{L}$.}
Under the assumptions above, the mapping
\[
    (x,\theta) \mapsto \widehat{\Psi}_\theta(x)
\]
is continuous in both arguments: continuity in $\{\tau_j\}$ and operator weights follows from continuous dependence of flows on time parameters and the softmax relaxation in Definition~\ref{def:strategy}; continuity in $x \in \Lambda$ follows from uniform well-posedness and stability of the flows. Since the $L^2$ norm is continuous, $\mathcal{L}(x,\theta)$ is continuous on $\Lambda \times \Theta$.

\paragraph{Part 2: Existence of an $\varepsilon$-optimal selection.}
We restrict attention to a compact subset $\bar{\Theta} \subset \Theta$ where $|\tau_j| \in [\tau_{\min}, \tau_{\max}]$ for fixed $0 < \tau_{\min} < \tau_{\max} < \infty$ and $k \leq K_{\max}$. Intuitively, extremely small or large time steps are either redundant or unstable and can be excluded without loss of optimality.

For each $x \in \Lambda$, consider the $\varepsilon$-sublevel set:
\[
    \Theta^*_\varepsilon(x)
    = \Big\{\theta \in \bar{\Theta}
        : \mathcal{L}(x, \theta)
          \leq \inf_{\theta' \in \bar{\Theta}} \mathcal{L}(x, \theta')
             + \varepsilon
      \Big\}.
\]
By compactness of $\bar{\Theta}$ and continuity of $\mathcal{L}$, the infimum is attained and each $\Theta^*_\varepsilon(x)$ is non-empty and closed. Berge's maximum theorem~\citep{Berge1963} implies that the argmin correspondence $x \mapsto \arg\min_{\theta} \mathcal{L}(x,\theta)$ is upper hemicontinuous with compact values, and standard perturbation arguments show that the $\varepsilon$-argmin correspondence $x \mapsto \Theta^*_\varepsilon(x)$ inherits these properties. In particular, for fixed $\varepsilon > 0$, there exists a selection $\theta^*_\varepsilon(x)$ that is Borel-measurable and $\varepsilon$-optimal for each $x$.

When we additionally view $\bar{\Theta}$ through the continuous relaxation $\iota(\bar{\Theta}) \subset \mathbb{R}^{K_{\max}(n+1)+1}$, we can apply approximate selection results (see, e.g.,~\citep{Michael1956}) to obtain a \emph{continuous} selection that is $\varepsilon$-optimal up to an arbitrarily small slack. For the purposes of this paper, we assume the existence of such a continuous $\varepsilon$-optimal policy $\theta^*_\varepsilon: \Lambda \to \bar{\Theta}$.
\smallskip

\paragraph{Part 3: Neural representation.}
Since $\Lambda$ is compact and $\theta^*_\varepsilon$ is continuous, it is uniformly continuous. By the Universal Approximation Theorem~\citep{Cybenko1989, Hornik1989}, there exists a feedforward neural network $\pi_\phi: \Lambda \to \bar{\Theta}$ such that:
\[
    \sup_{x \in \Lambda} \|\pi_\phi(x) - \theta^*_\varepsilon(x)\| < \delta
\]
for any prescribed $\delta > 0$. By continuity of $\mathcal{L}(x,\theta)$ in $\theta$, choosing $\delta$ sufficiently small ensures that $\pi_\phi$ is also $\varepsilon$-optimal (up to an arbitrarily small slack). This yields the desired learnable policy family.
\qed

\begin{remark}[Why $\varepsilon$-optimality?]
We work with near-optimal rather than exactly optimal strategies because exact minimizers may jump discontinuously as $x$ varies (e.g., when multiple compositions tie). Allowing an $\varepsilon$ margin ensures that we can select policies that vary smoothly with $x$, which is crucial for approximation by neural networks.
\end{remark}

\begin{remark}[Training via Evolution Strategies]
In practice, we do not solve the selection problem analytically. Instead, we parameterize $\pi_\phi$ as a neural network and minimize $\mathbb{E}_{x \sim \rho}[\mathcal{L}(x, \pi_\phi(x))]$ via Evolution Strategies~\citep{Salimans2017}. Since $\widehat{\Psi}_\theta$ may involve non-differentiable black-box solvers, ES provides gradient estimates of the form:
\[
    \nabla_\phi \mathbb{E}[\mathcal{L}]
    \approx \frac{1}{\sigma} \mathbb{E}_{\epsilon \sim \mathcal{N}(0,I)}
        \big[\epsilon \cdot \mathcal{L}(x, \pi_{\phi+\sigma\epsilon}(x))\big],
\]
which are compatible with our compositional setting.
\end{remark}

\subsection{Proof of Theorem~\ref{thm:fitting} (Universal approximation for fitting)}

Finally, we combine the previous results to obtain the universal approximation theorem for the surrogate fitting regime.

By Theorem~\ref{thm:existence}, for any $\varepsilon_1 > 0$, there exists an exact composite flow $\Psi \in \mathcal{C}$ such that the splitting error satisfies:
\[
    \sup_{x \in \Lambda} \|\mathcal{S}(x) - \Psi(x)\|_{L^2} < \varepsilon_1.
\]
By Theorem~\ref{thm:error_decomp_appendix}, if approximate sub-flows satisfy Assumption~\ref{ass:approx} with order $q$ and step sizes $\{\tau_j\}$ small enough, the sub-solver error satisfies:
\[
    \sup_{x \in \Lambda} \|\Psi(x) - \widehat{\Psi}(x)\|_{L^2} < \varepsilon_2.
\]

By the triangle inequality, the total error satisfies:
\[
    \sup_{x \in \Lambda} \|\mathcal{S}(x) - \widehat{\Psi}(x)\|_{L^2}
    \leq \varepsilon_1 + \varepsilon_2.
\]
Given any target $\varepsilon > 0$, choose $\varepsilon_1 = \varepsilon_2 = \varepsilon/2$ and appropriate compositions and step sizes to achieve these bounds.

By Theorem~\ref{thm:optimal_policy}, there exists a learnable policy $\theta^*$ corresponding to such a composition, and this policy can be approximated to arbitrary precision by a feedforward neural network. This yields the universal approximation statement for the surrogate fitting regime. \qed

\begin{remark}[Practical implications]
For Evolution Strategies with classical numerical sub-solvers (e.g., RK4, spectral methods), Assumption~\ref{ass:approx} is typically satisfied with high order $q$, so the dominant limitation is the splitting error. When neural operators serve as sub-solvers, the stability condition (S2) becomes the main design constraint.
\end{remark}

\begin{remark}[Structure preservation]
If a structural property $\mathcal{P}$ (e.g., positivity, conservation of mass) is preserved by each sub-flow $\Phi^{(i)}_t$ and under composition, then any $\Psi \in \mathcal{C}$ exactly preserves $\mathcal{P}$. With numerical sub-solvers, $\mathcal{P}$ is preserved up to the solver accuracy. HyCOP therefore inherits any invariants preserved at the level of primitives.
\end{remark}

\subsection{Generalization under distributional shift}
\label{sec:gen_shift}

The preceding results establish that finite composite flows of primitives can approximate the solution operator uniformly on compact sets, and that the practical error decomposes into a splitting error plus a sub-solver error. We now show how these properties can be leveraged to obtain simple generalization guarantees under \emph{covariate shift}, i.e., when the distribution of inputs (initial conditions, parameters, boundary conditions) changes between training and test.

\paragraph{Setup.}
Recall that inputs take the form:
\[
    x = (\mu, u_0, b, \Omega)  \in \Lambda,\, t\in \mathcal{T}
\]
and that we assume $\Lambda$ is compact (e.g., parameters, initial conditions, and boundary data lie in bounded subsets of suitable Banach spaces). We equip $\Lambda$ with a metric $d_\Lambda$ that reflects the natural distances between inputs (e.g., an $L^2$ or $H^s$ metric on $u_0$ and $b$, and Euclidean distance on $\mu$). Let $\rho$ denote the training distribution on $\Lambda$, and $\rho'$ a test distribution that may differ from $\rho$ (e.g., due to shifted initial conditions or parameter ranges).

We make explicit the Lipschitz continuity of the solution operator and the HyCOP surrogate with respect to inputs.
To make notations compact we suppress the time $t$ in the following.

\begin{assumption}[Lipschitz dependence on inputs]\label{ass:lipschitz_inputs}
Under the stable split system assumptions, there exist constants $L_{\mathcal{S}}, L_{\text{HyCOP}} > 0$ such that, for all $x,x' \in \Lambda$,
\begin{align*}
    \|\mathcal{S}(x) - \mathcal{S}(x')\|_{L^2}
    &\leq L_{\mathcal{S}} \, d_\Lambda(x,x'), \\
    \|\widehat{\Psi}_\theta(x) - \widehat{\Psi}_\theta(x')\|_{L^2}
    &\leq L_{\text{HyCOP}} \, d_\Lambda(x,x'),
\end{align*}
for any fixed strategy $\theta$.
\end{assumption}

\begin{remark}
Assumption~\ref{ass:lipschitz_inputs} is a standard consequence of well-posedness and stability for the full PDE and for the approximate flows. In particular, the stability bounds (R2) and (S2) together imply that perturbing $(\mu,u_0,b,\Omega)$ at $t=0$ produces changes in $u(t)$ and $\widehat{\Psi}_\theta(x)$ that are at most exponentially amplified in time, uniformly over $t \in [0,T]$. On a compact parameter set $\Lambda$, this yields finite Lipschitz constants $L_{\mathcal{S}}, L_{\text{HyCOP}}$.
\end{remark}

Under Assumption~\ref{ass:lipschitz_inputs}, the loss $\mathcal{L}(x,\theta)$ is Lipschitz in $x$ as well.

\begin{lemma}[Lipschitz loss]\label{lem:lipschitz_loss}
For any fixed $\theta$, the loss:
\[
    \mathcal{L}(x,\theta) = \|\mathcal{S}(x) - \widehat{\Psi}_\theta(x)\|_{L^2},
\]
satisfies:
\[
    |\mathcal{L}(x,\theta) - \mathcal{L}(x',\theta)|
    \leq L_{\text{tot}} \, d_\Lambda(x,x'),
\]
for all $x,x' \in \Lambda$, where $L_{\text{tot}} = L_{\mathcal{S}} + L_{\text{HyCOP}}$.
\end{lemma}

\begin{proof}
By the reverse triangle inequality,
\begin{align*}
    |\mathcal{L}(x,\theta) - \mathcal{L}(x',\theta)|
    &= \big|\|\mathcal{S}(x) - \widehat{\Psi}_\theta(x)\|_{L^2}
          - \|\mathcal{S}(x') - \widehat{\Psi}_\theta(x')\|_{L^2}\big| \\
    &\leq \|\big(\mathcal{S}(x) - \widehat{\Psi}_\theta(x)\big)
           - \big(\mathcal{S}(x') - \widehat{\Psi}_\theta(x')\big)\|_{L^2} \\
    &\leq \|\mathcal{S}(x) - \mathcal{S}(x')\|_{L^2}
        + \|\widehat{\Psi}_\theta(x) - \widehat{\Psi}_\theta(x')\|_{L^2} \\
    &\leq \big(L_{\mathcal{S}} + L_{\text{HyCOP}}\big)\, d_\Lambda(x,x'),
\end{align*}
which gives the desired bound.
\end{proof}

We can now relate the expected loss under two different input distributions via a Wasserstein distance.

\begin{definition}[Wasserstein-1 distance]
Let $\mathcal{P}(\Lambda)$ denote the set of probability measures on $(\Lambda, d_\Lambda)$. The Wasserstein-1 distance between $\rho, \rho' \in \mathcal{P}(\Lambda)$ is:
\[
    W_1(\rho,\rho')
    = \inf_{\gamma \in \Pi(\rho,\rho')}
        \int_{\Lambda \times \Lambda} d_\Lambda(x,x') \, d\gamma(x,x'),
\]
where $\Pi(\rho,\rho')$ is the set of couplings of $\rho$ and $\rho'$.
\end{definition}

\begin{theorem}[Generalization under covariate shift]\label{thm:shift_generalization}
Let Assumption~\ref{ass:lipschitz_inputs} hold, and fix a strategy $\theta$. Then, for any two distributions $\rho, \rho'$ on $\Lambda$,
\[
    \left|\mathbb{E}_{x \sim \rho'}[\mathcal{L}(x,\theta)]
          - \mathbb{E}_{x \sim \rho}[\mathcal{L}(x,\theta)]\right|
    \leq L_{\text{tot}} \, W_1(\rho,\rho'),
\]
where $L_{\text{tot}} = L_{\mathcal{S}} + L_{\text{HyCOP}}$.
\end{theorem}

\begin{proof}
By Lemma~\ref{lem:lipschitz_loss}, the loss $\mathcal{L}(\cdot,\theta)$ is $L_{\text{tot}}$-Lipschitz on $(\Lambda, d_\Lambda)$. By the Kantorovich--Rubinstein duality for $W_1$, for any $L$-Lipschitz function $f$ we have:
\[
    \left|\mathbb{E}_{\rho'}[f] - \mathbb{E}_{\rho}[f]\right|
    \leq L \, W_1(\rho,\rho').
\]
Applying this with $f(x) = \mathcal{L}(x,\theta)$ and $L = L_{\text{tot}}$ gives the result.
\end{proof}

\begin{remark}[Interpretation for PDE inputs]
If $\rho$ represents a distribution over initial conditions, parameters, and boundary conditions used during training, and $\rho'$ is a test distribution that shifts these (e.g., new initial condition statistics, modified parameter ranges), then Theorem~\ref{thm:shift_generalization} quantifies how much the \emph{expected} surrogate error can deteriorate under this shift. If the shift is small in $W_1$ (for example, small perturbations in initial conditions in $H^s$), the increase in expected error is at most linear in the size of the shift.
\end{remark}

\begin{remark}[Combining with uniform approximation]
In the ideal surrogate fitting regime of Theorem~\ref{thm:fitting}, HyCOP achieves a uniform bound:
\[
    \sup_{x \in \Lambda} \mathcal{L}(x,\theta^*) \leq \varepsilon.
\]
In that case, \emph{any} distributional shift supported in $\Lambda$ satisfies:
\[
    \mathbb{E}_{x \sim \rho'}[\mathcal{L}(x,\theta^*)]
    \leq \varepsilon,
\]
so generalization is automatic and independent of $W_1(\rho,\rho')$. The more interesting regime is when we only know that the expected training error $\mathbb{E}_{x \sim \rho}[\mathcal{L}(x,\theta)]$ is small; Theorem~\ref{thm:shift_generalization} then quantifies how far we can move in input space before the expected test error deteriorates.
\end{remark}

\begin{remark}[From population to finite-sample generalization]
Theorem~\ref{thm:shift_generalization} is a \emph{distributional} statement: it bounds the gap between population risks under $\rho$ and $\rho'$. Standard statistical learning tools (e.g., Rademacher complexity or covering number bounds for the family $\{\mathcal{L}(\cdot,\pi_\phi(\cdot))\}$) can be combined with this result to control three gaps:
\[
    \underbrace{\mathbb{E}_{\rho'}[\mathcal{L}] - \mathbb{E}_{\rho}[\mathcal{L}]}_{\text{distribution shift}}
    \;+\;
    \underbrace{\mathbb{E}_{\rho}[\mathcal{L}] - \mathbb{E}_{\hat{\rho}_N}[\mathcal{L}]}_{\text{finite-sample generalization}}
    \;+\;
    \underbrace{\mathbb{E}_{\hat{\rho}_N}[\mathcal{L}] - \mathbb{E}_{\hat{\rho}_N}[\widehat{\mathcal{L}}]}_{\text{optimization}},
\]
where $\hat{\rho}_N$ is the empirical training distribution and $\widehat{\mathcal{L}}$ is the loss of the learned policy. We leave a full statistical learning theory for HyCOP to future work, but Theorem~\ref{thm:shift_generalization} shows that covariate shift enters the picture through a simple and interpretable Wasserstein term.
\end{remark}

\section{HyCOP Training Details}
\label{app:training}

\subsection{Policy inputs and outputs}
\label{app:training:policy}

The policy conditions on the PDE query $x=(\mu,u_0,b,\Omega)$, the query time $t$, and a compact feature
vector $f(x)\in\mathbb{R}^m$ capturing regime information and coarse state statistics:
\[
\tilde{x}=(x,t,f(x)),\qquad 
f(x)=\big[f_{\mathrm{phys}}(x),\, f_{\mathrm{stats}}(u_0)\big].
\]
These summaries combine dimensionless regime indicators (e.g., P\'eclet, Damk\"ohler) with scale-free state statistics (e.g., coefficient of variation of $u_0$ and its gradients), helping the policy allocate time across primitives. Specific feature definitions for each PDE system are given in Appendix~\ref{app:exp_configs}.

Given $\tilde{x}=(x,t,f(x))$, the policy predicts a $k$-step program: operator choices and positive durations $\tau_r > 0$ (parameterized via softplus to ensure positivity). Operator selection is categorical at test time and trained via a softmax relaxation.

\subsection{Evolution Strategies training}
\label{app:training:es}

We optimize policy parameters with Evolution Strategies (ES), treating program execution as a black box.
While ES is a zeroth-order method, it is not naive random search: the update rule provides an unbiased 
estimate of $\nabla_\theta J(\theta)$ in expectation.
Each generation samples $\epsilon_i\sim\mathcal{N}(0,I)$ and evaluates antithetic losses
$\mathcal{L}(\theta\pm\sigma\epsilon_i)$; this symmetric evaluation halves gradient variance compared 
to one-sided sampling.
We then apply rank-based fitness shaping, which replaces raw losses with their ranks before computing 
updates, making optimization robust to outliers and invariant to loss scaling.
The parameter update is
\[
g=\frac{1}{2M\sigma}\sum_{i=1}^M\big(w_i^+ - w_i^-\big)\epsilon_i,
\qquad
\theta\leftarrow \theta-\eta g.
\]
Together, these variance-reduction techniques enable stable training despite non-differentiable 
primitive evaluations.
ES is compatible with hybrid dictionaries (numerical and learned primitives) and supports transfer 
by freezing selected primitives while updating only the policy and/or added residual modules. (Fig.~\ref{fig:hycop_training_es})


\begin{figure}[t]
\centering
\resizebox{\columnwidth}{!}{%
\begin{tikzpicture}[
  font=\small,
  >=Latex,
  panel/.style={draw, rounded corners=3pt, thick, fill=white, inner sep=6pt},
  panelthin/.style={draw, rounded corners=3pt, semithick, fill=white, inner sep=5pt},
  title/.style={font=\bfseries},
  tag/.style={fill=white, draw=none, inner sep=1.2pt, rounded corners=2pt, font=\scriptsize},
  box/.style={draw, rounded corners=3pt, semithick, fill=white, inner sep=5pt, align=center},
  arrow/.style={->, thick},
  dashedpanel/.style={draw, dashed, rounded corners=6pt, semithick}
]

\node[panelthin, minimum width=3.2cm, minimum height=1.05cm, align=left] (query) at (0,0) {};
\node[title, anchor=north west] at ($(query.north west)+(0.18,0.45)$) {Query};
\node[anchor=north west] at ($(query.north west)+(0.18,-0.15)$)
{$x=(u_0,\mu,b,\Omega,\mathcal{T})$};

\node[panelthin, minimum width=3.2cm, minimum height=1.05cm, right=14mm of query] (featbox) {};
\node[title] at ($(featbox.north)+(0,-0.28)$) {Features $f(x)$};
\node[font=\scriptsize] at ($(featbox.center)+(0,-0.10)$)
{Pe/Da/Fr, variances, $T$, \ldots};

\node[panelthin, minimum width=4.0cm, minimum height=1.05cm, right=14mm of featbox] (policy) {};
\node[title] at ($(policy.north)+(0,-0.28)$) {Policy $\pi_\theta$};
\node[font=\scriptsize] at ($(policy.center)+(0,-0.10)$)
{program $(i_1,\tau_1),\dots,(i_k,\tau_k)$};

\node[panelthin, minimum width=3.6cm, minimum height=1.05cm, right=14mm of policy] (exec) {};
\node[title] at ($(exec.north)+(0,-0.28)$) {Compose \& execute};
\node[font=\scriptsize] at ($(exec.center)+(0,-0.10)$)
{$\widehat{\Phi}(x;\pi_\theta)$};

\node[panelthin, minimum width=3.1cm, minimum height=1.05cm, right=14mm of exec] (out) {};
\node[title] at ($(out.north)+(0,-0.28)$) {Prediction};
\node[font=\scriptsize] at ($(out.center)+(0,-0.10)$)
{$\{u(t)\}_{t\in\mathcal{T}}$};

\draw[arrow] (query.east) -- (featbox.west) node[pos=0.5, above, tag] {extract};
\draw[arrow] (featbox.east) -- (policy.west) node[pos=0.5, above, tag] {cond.};
\draw[arrow] (policy.east) -- (exec.west);
\draw[arrow] (exec.east) -- (out.west);

\node[box] (shape) at ([yshift=-22mm]$(query.south)!0.5!(out.south)$)
{fitness shaping\\[-2pt]\scriptsize rank transform};

\node[box, left=12mm of shape] (fitness)
{evaluate fitness\\[-2pt]\scriptsize run composite\\[-2pt]\scriptsize loss $\mathcal{L}$};

\node[box, left=12mm of fitness] (anti)
{antithetic\\ sampling\\[-2pt]\scriptsize $\theta\pm\sigma\epsilon$};

\node[box, right=12mm of shape] (grad)
{ES gradient\\[-2pt]\scriptsize estimate};

\node[box, right=12mm of grad] (upd)
{parameter\\ update};

\node[dashedpanel, fit=(anti)(fitness)(shape)(grad)(upd), inner sep=8pt] (loop) {};

\draw[arrow] (anti) -- (fitness);
\draw[arrow] (fitness) -- (shape);
\draw[arrow] (shape) -- (grad);
\draw[arrow] (grad) -- (upd);

\draw[arrow] (upd.north) .. controls +(0,12mm) and +(0,-12mm) .. (policy.south);
\node[tag] at ($(upd.north)+(0.6,9mm)$) {iterate};

\end{tikzpicture}%
}
\caption{\textbf{HyCOP training with Evolution Strategies (ES).}
HyCOP conditions the policy on low-dimensional physics-based features $f(x)$ (e.g., P\'eclet/Damk\"ohler/Froude numbers and state statistics) to predict a split program (primitive choices and durations).
We optimize policy parameters via ES using black-box evaluations of the composed operator loss.}
\label{fig:hycop_training_es}
\end{figure}
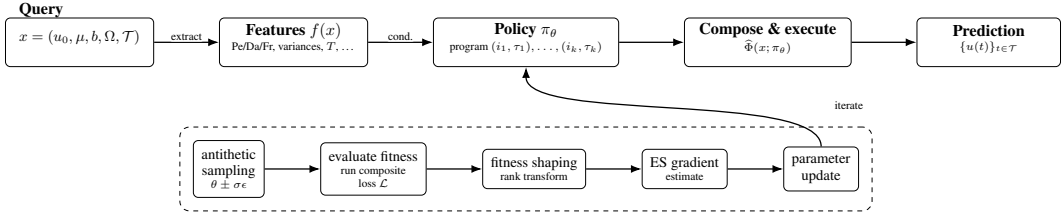


\begin{table}[p]
\caption{\textbf{Ablations and analysis.}
All methods train on the same data; HyCOP uses identical ES hyperparameters ($M{=}500$, $\sigma{=}0.02$, 200 generations) for all main benchmarks without per-problem tuning. Benchmark-specific details in Appendix text. Adaptation experiments (\S\ref{sec:exp:regime_b}), which use
only 120 ADR samples, train with a smaller ES budget; details in Appendix~\ref{app:ad_adr}. }
\label{tab:ablations}
\centering

\vspace{0.3em}
{\footnotesize\textbf{(a) Conditioning features: dimensionless vs.\ raw IC inputs} (1D SWE)}\\[3pt]
\footnotesize
\setlength{\tabcolsep}{4pt}
\begin{tabular}{lc ccc ccc}
\toprule
& & \multicolumn{3}{c}{ID} & \multicolumn{3}{c}{OOD} \\
\cmidrule(lr){3-5}\cmidrule(lr){6-8}
Features & Dims & Rel.\ $L^2$ & RMSE & Max Err & Rel.\ $L^2$ & RMSE & Max Err \\
\midrule
Raw IC         & 129 & \textbf{2.79\sn{-2}} & \textbf{4.27\sn{-2}} & 1.21\sn{0}  & 4.54\sn{-2} & 8.09\sn{-2} & 1.04\sn{0}  \\
\textbf{Dimensionless} & \textbf{4} & 2.97\sn{-2} & 5.27\sn{-2} & \textbf{4.66\sn{-1}} & \textbf{3.07\sn{-2}} & \textbf{5.27\sn{-2}} & \textbf{4.66\sn{-1}} \\
\midrule
\multicolumn{2}{l}{\textit{OOD improvement}} & & & & $32\%$ & $35\%$ & $55\%$ \\
\bottomrule
\end{tabular}

\vspace{1.0em}

{\footnotesize\textbf{(b) ES hyperparameter sensitivity} (2D SWE, Rel.\ $L^2$ $\downarrow$, 200 generations)}\\[3pt]
\footnotesize
\setlength{\tabcolsep}{5pt}
\begin{tabular}{lccccc}
\toprule
$M$ & $\sigma{=}0.005$ & $\sigma{=}0.01$ & $\sigma{=}0.02$ & $\sigma{=}0.05$ & $\sigma{=}0.1$ \\
\midrule
100 & 0.274 & 0.065 & 0.035 & 0.308 & 0.488 \\
250 & 0.125 & 0.102 & 0.031 & 0.155 & 0.206 \\
500 & 0.052 & 0.026 & \textbf{0.028}$^\star$ & 0.392 & 0.612 \\
\bottomrule
\multicolumn{6}{l}{\scriptsize $^\star$Paper default. Even $M{=}100$ outperforms all monolithic baselines OOD ($>3.5\sn{-1}$).}
\end{tabular}

\vspace{1.0em}

{\footnotesize\textbf{(c) Zero-shot resolution transfer} (2D SWE, $T{=}0.3$; policy trained at $32{\times}32$)}\\[3pt]
\footnotesize
\setlength{\tabcolsep}{5pt}
\begin{tabular}{lcccc}
\toprule
Resolution & Rel.\ $L^2$ & RMSE & Max Err & Time (s/sample) \\
\midrule
$32{\times}32$ (trained) & 2.40\sn{-2} & 1.03\sn{-2} & 1.90\sn{-1} & 0.126 \\
$128{\times}128$ (zero-shot) & 3.37\sn{-2} & 2.06\sn{-2} & 6.89\sn{-1} & 0.645 \\
\bottomrule
\end{tabular}

\vspace{1.0em}

{\footnotesize\textbf{(d) Dictionary robustness} (2D SWE / AD$\to$ADR)}\\[3pt]
\footnotesize
\setlength{\tabcolsep}{4pt}
\begin{tabular}{l l c c}
\toprule
Experiment & Dictionary & Rel.\ $L^2$ & Notes \\
\midrule
\multicolumn{4}{l}{\emph{Redundant primitive (2D SWE):}} \\
Correct dict.           & \{Adv, Grav\}                          & 2.40\sn{-2} & --- \\
+ dummy reaction        & \{Adv, Grav, React$_{\text{dummy}}$\}  & 2.30\sn{-2} & React$_{\text{dummy}}$: 3.5\% time \\
\midrule
\multicolumn{4}{l}{\emph{Missing primitive (AD$\to$ADR):}} \\
AD zero-shot on ADR     & \{Adv, Diff\}                          & 1.81\sn{-1} & errors localize to reaction \\
+ UNO residual          & \{Adv, Diff, Resid.\}                  & 4.42\sn{-2} & policy relearned only \\
\bottomrule
\end{tabular}

\vspace{1.0em}

{\footnotesize\textbf{(e) Chaotic validation: 1D Kuramoto--Sivashinsky}
(trained: $W{\in}[24,40]$, $T{\in}[5,8]$)}\\[3pt]
\footnotesize
\setlength{\tabcolsep}{4pt}
\begin{tabular}{l cc cc}
\toprule
& \multicolumn{2}{c}{ID} & \multicolumn{2}{c}{OOD ($W{\in}[40,50]$, $T{\in}[8,20]$)} \\
\cmidrule(lr){2-3}\cmidrule(lr){4-5}
Method & SE ($\downarrow$) & KL ($\downarrow$) & SE ($\downarrow$) & KL ($\downarrow$) \\
\midrule
HyCOP & 9.74\sn{-2} & 4.72\sn{-2} & 5.75\sn{-2} & 4.95\sn{-2} \\
\bottomrule
\end{tabular}

\vspace{1.0em}

{\footnotesize\textbf{(f) Computational cost} (2D SWE)}\\[3pt]
\footnotesize
\setlength{\tabcolsep}{3.5pt}
\begin{tabular}{l c cccc}
\toprule
& Training & \multicolumn{4}{c}{Inference time per sample (s)} \\
\cmidrule(lr){3-6}
Model & Fwd.\ passes & 1-step & 5-step & 10-step & 20-step \\
\midrule
U-Net                        & 1.33\sn{7} & 0.042 & 0.026 & 0.049 & 0.100 \\
AR-Loc.\ Int.\ Diff.\ FNO   & 1.33\sn{7} & 0.173 & 0.258 & 0.289 & 0.327 \\
Poseidon (fine-tune)         & 4.86\sn{5} & 0.134 & 0.714 & 1.331 & 2.667 \\
HyCOP                        & 5.00\sn{5} & 0.680 & 1.005 & 1.205 & 1.341 \\
\bottomrule
\end{tabular}

\end{table}

\section{Experimental configurations and results}
\label{app:exp_configs}

We summarize PDE definitions, data distributions, train/test splits (ID\ and OOD), conditioning, and hyperparameters.

\paragraph{Common setup.}
\label{app:common_setup}
Unless stated otherwise, we use periodic boundaries. HyCOP predicts (i) primitive choices, (ii) normalized time allocations across sub-steps, and (iii) an adaptive program length ($k\in[3,18]$).
We train HyCOP with Evolution Strategies (ES)~\citep{Salimans2017} (population 500, $\sigma=0.02$, lr $5\!\times\!10^{-3}$, wd $10^{-3}$, antithetic sampling, rank-based shaping) for 200 generations.
FNO/Loc. Int. Diff. FNO follow PDEBench~\citep{PDEBench2022} and are trained with Adam (lr $10^{-3}$) for 500 epochs~\citep{li2020fourier,herde2024poseidon}. DeepONet~\citep{lu2021deeponet} and PINO~\citep{li2021pino} share the same Adam optimizer (lr $10^{-3}$, weight decay $10^{-4}$) and 500-epoch budget as FNO. DeepONet uses an unstacked branch--trunk architecture with $p{=}128$ basis coefficients per output channel and hidden width 256: a CNN branch encodes the initial state, and an MLP trunk encodes query coordinates $(x,y)$. PINO uses the same FNO backbone as our FNO baseline; we follow the two-phase schedule of~\citet{li2021pino}, where the first half of training uses data loss only and the second half adds a PDE-residual loss evaluated on the training grid.
All 1D systems use $N=64$; our 2D SWE and ADR benchmarks use $32\times 32$ with 10,000 training trajectories. The PDEBench 2D compressible Navier--Stokes dataset is used at its native $128\times 128$ resolution~\citep{PDEBench2022}.
When HyCOP is trained with variable query time, fixed-time baselines are evaluated at the same target time.

\paragraph{Time conditioning.}
HyCOP is trained with variable query time when the task involves time querying; fixed-time baselines (FNO/Loc. Int. Diff. FNO) are trained/evaluated at the reported target time.

\paragraph{Metrics.}
Given prediction $\hat{u}$ and reference $u$ on a grid with $N$ points (or $N_xN_y$ in 2D), we report
$\mathrm{RMSE}(\hat{u},u)=\sqrt{\frac{1}{N}\sum_{p}\|\hat{u}(p)-u(p)\|_2^2}$
and relative $L^2$ error
$\mathrm{Rel}L^2(\hat{u},u)=\frac{\|\hat{u}-u\|_2}{\|u\|_2}$,
where $\|\cdot\|_2$ denotes the discrete $\ell_2$ norm over grid points and state channels.
We also report $\mathrm{MaxErr}=\max_{p}\|\hat{u}(p)-u(p)\|_2$.

\paragraph{Spectral error by bands (fRMSE).}
Let $\hat{u}(k)$ denote the 2D discrete Fourier transform of $u$, and let 
$|k| = \sqrt{k_x^2 + k_y^2}$ be the radial wavenumber (normalized so that Nyquist corresponds to $|k|=0.5$).
We partition wavenumbers into disjoint bands:
$\mathcal{K}_{\text{low}} = \{k : |k| < 0.1\}$,
$\mathcal{K}_{\text{mid}} = \{k : 0.1 \le |k| < 0.3\}$,
$\mathcal{K}_{\text{high}} = \{k : |k| \ge 0.3\}$.
Define
\[
\mathrm{fRMSE}_{\mathcal{K}}(\hat{u},u)
=\sqrt{\frac{1}{|\mathcal{K}|}\sum_{k\in\mathcal{K}}\big|\hat{u}_{\mathrm{pred}}(k)-\hat{u}_{\mathrm{ref}}(k)\big|^2},
\]
and report $\mathrm{fRMSE}_{\text{low/mid/high}}$ for each band.

\paragraph{Boundary RMSE (bRMSE).}
For problems with non-periodic boundaries, let $\mathcal{B}\subset\Omega_h$ denote the boundary band
(the outermost $5\%$ of grid cells on each side). Define
\[
\mathrm{bRMSE}(\hat{u},u)=\sqrt{\frac{1}{|\mathcal{B}|}\sum_{p\in\mathcal{B}}\|\hat{u}(p)-u(p)\|_2^2}.
\]
For periodic-boundary benchmarks, bRMSE is reported only when a boundary module is present (e.g., wall-transfer).

\paragraph{Constraint RMSE (cRMSE).}
Let $C(u)\in\mathbb{R}^r$ denote diagnostic constraints/invariants (e.g., total mass $\int h\,dxdy$ for SWE).
We report
\[
\mathrm{cRMSE}(\hat{u},u)=\sqrt{\frac{1}{r}\sum_{j=1}^r\big(C_j(\hat{u})-C_j(u)\big)^2}.
\]
When the PDE does not admit the corresponding invariant (e.g., Fisher--KPP reactions), cRMSE is not applicable.

We compute fRMSE on Fourier magnitudes (not phases) and average over wavenumbers in each band,
so values can be small when spectra agree even if pointwise phases differ.

\subsection{Benchmark coverage}
\label{app:benchmark_coverage}
HyCOP is evaluated on PDE families spanning linear/nonlinear dynamics, stiffness, multi-physics coupling, discontinuities, and long-horizon behavior.
We use ``stiff'' for regimes where diffusive/source terms impose shorter time scales than advection, and ``multiscale'' for solutions with energy across multiple spatial frequencies and/or sharp localized gradients over the evaluation horizon.

\noindent\textbf{Regime notes.}
ADR varies P\'eclet and Damk\"ohler numbers through $(c_x,c_y,D_x,D_y,r)$, producing sharp fronts and reaction-dominated transients in parts of OOD.
SWE OOD includes extreme-Froude/transcritical configurations that induce steep gradients, wave interactions, and shock-like features; dam-break transfer further stresses discontinuities and boundary reflections.

\begin{table}[t]
\caption{\textbf{Benchmark coverage and validation tests.}}
\label{tab:benchmark_coverage}
\centering
\footnotesize
\setlength{\tabcolsep}{5pt}
\begin{tabular}{lccccc}
\toprule
System & Nonlinear & Multiphysics & Multiscale & Chaotic & Long-time \\
\midrule
\emph{(Benchmarked)} \\
1D Advection--Diffusion (AD)    & --         & \checkmark & \checkmark & --         & \checkmark \\
1D Viscous Burgers              & \checkmark & --         & \checkmark & --         & \checkmark \\
1D Shallow Water (smooth ICs)   & \checkmark & \checkmark & \checkmark & --         & \checkmark \\
2D ADR (Fisher--KPP)            & \checkmark & \checkmark & \checkmark & --         & \checkmark \\
2D SWE (smooth / vortices)      & \checkmark & \checkmark & \checkmark & --         & \checkmark \\
2D Compressible Navier--Stokes  & \checkmark & \checkmark & \checkmark & --         & \checkmark \\
Dam-break SWE (transfer)        & \checkmark & \checkmark & \checkmark & --         & \checkmark \\
\midrule
\emph{(Validated)} \\
1D Kuramoto--Sivashinsky (KS)   & \checkmark & --         & \checkmark & \checkmark & \checkmark \\
\bottomrule
\end{tabular}
\end{table}

\begin{table}[t]
\caption{\textbf{Dictionary and policy size.} Primitives are single-process routines;
For AD$\to$ADR transfer, we report two complementary adaptation tracks (Path~(a): numerical AD primitives $+$ learned residual; Path~(b): pretrained AD $+$ numerical reaction), and a fully-learned ablation where every primitive is a per-process learned surrogate.}
\label{tab:dictionary}
\centering
\footnotesize
\setlength{\tabcolsep}{3pt}
\begin{tabular}{llcc}
\toprule
Benchmark & Dictionary $\mathbb{D}$ & $n$ & Policy params \\
\midrule
1D Advection--Diffusion           & \{Adv., Diff.\}                                & 2 & ${\sim}50$ \\
1D Viscous Burgers                & \{Nonlin.\ adv., Visc.\ diff.\}                & 2 & ${\sim}50$ \\
1D Shallow Water                  & \{Wave/adv., Gravity\}                         & 2 & ${\sim}50$ \\
1D Kuramoto--Sivashinsky          & \{Linear, Nonlinear\}                          & 2 & ${\sim}50$ \\
2D ADR (Fisher--KPP)              & \{Adv., Diff., Reaction\}                      & 3 & ${\sim}73$ \\
2D Shallow Water                  & \{Wave/adv., Gravity\}                         & 2 & ${\sim}50$ \\
2D Compressible Navier--Stokes    & \{Euler adv., Viscous diff.\}                  & 2 & ${\sim}50$ \\
\midrule
AD$\to$ADR transfer, Path~(a)     & \{Adv., Diff., Learned resid.\}                & 3 & ${\sim}73$ \\
AD$\to$ADR transfer, Path~(b)     & \{Learned AD, Reaction\}            & 2 & ${\sim}50$ \\
\midrule
\emph{Ablation: fully learned ADR} & \{Learned Adv, Learned Diff, Learned Reaction\}           & 3 & ${\sim}73$ \\
\bottomrule
\end{tabular}
\end{table}

\subsection{1D Systems}
\label{app:1d_systems}

\subsubsection{Advection--Diffusion}
\label{app:1d_ad}

We consider the 1D advection--diffusion equation on $x\in[0,10]$,
\begin{equation}
    \frac{\partial u}{\partial t} + c\,\frac{\partial u}{\partial x} = D\,\frac{\partial^2 u}{\partial x^2},
\end{equation}
with the canonical split into advection and diffusion primitives.

\paragraph{Data generation.}
Training trajectories are generated with advection speed $c\in[0.5,3.0]$ and diffusion coefficient $D\in[0.01,0.5]$.
Initial conditions are sampled from five families: Gaussian pulses, step functions, sinusoidal waves, multi-Gaussian superpositions, and random low-frequency Fourier series.
We generate 10{,}000 training samples and evaluate on held-out ID\ test data at a fixed target time $T=0.5$.

\paragraph{OOD evaluation.}
OOD test cases extrapolate both parameters
($c\in[0.1,0.5]\cup[3.0,5.0]$, $D\in[0.001,0.01]\cup[0.5,1.0]$)
and initial conditions (very narrow/wide Gaussians, higher-frequency oscillations, and multi-step profiles).

\paragraph{Model inputs and conditioning.}
HyCOP is trained with a variable target time $T\in[0.1,1.0]$ and evaluated by querying $T=0.5$.
FNO and Loc. Int. Diff. FNO are trained and evaluated at the same fixed target time $T=0.5$.
The HyCOP policy network uses four physics-inspired features: a local P\'eclet number, concentration variance, gradient variance, and $T$.
FNO and Loc. Int. Diff. FNO use input channels for $u_0$, the spatial grid, and normalized parameters $(c,D)$.

\begin{figure}[H]
    \centering
    \includegraphics[width=\linewidth]{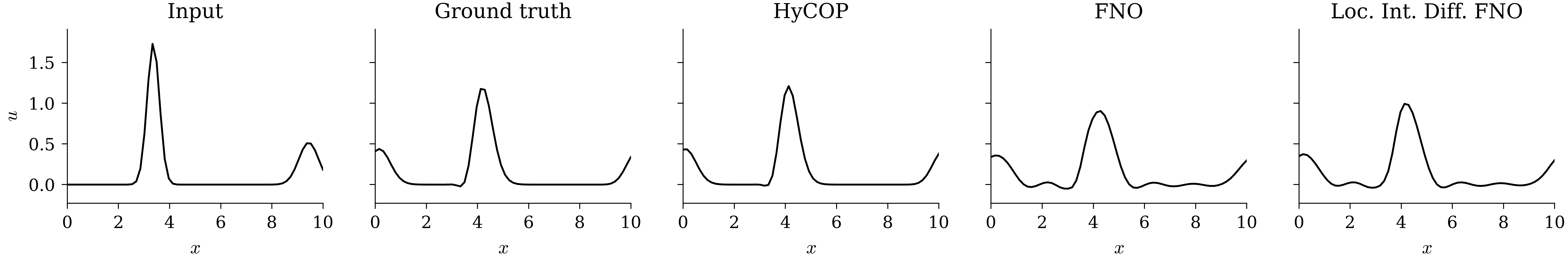}
    \caption{1D advection--diffusion (ID): qualitative comparison at $T=0.5$.}
    \label{fig:1d_ad_id}
\end{figure}

\begin{figure}[H]
    \centering
    \includegraphics[width=\linewidth]{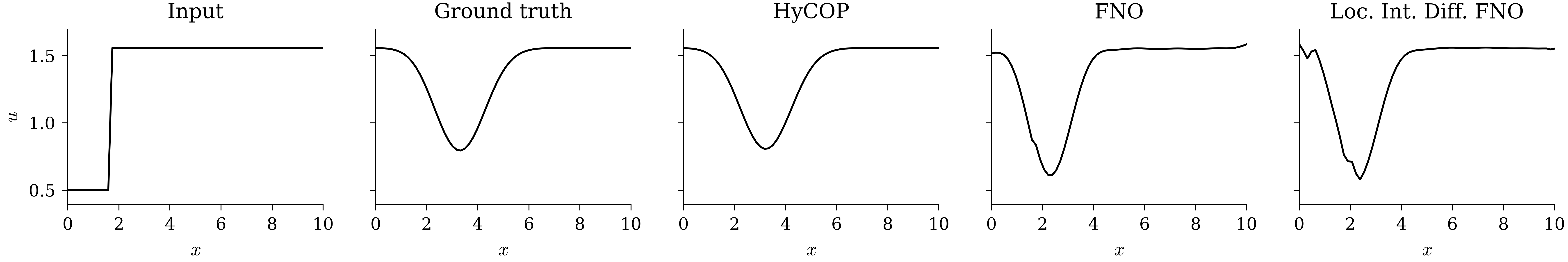}
    \caption{1D advection--diffusion (OOD): extrapolation in $(c,D)$ and IC family at $T=0.5$.}
    \label{fig:1d_ad_ood}
\end{figure}

\subsubsection{Shallow Water Equations}
\label{app:1d_swe}

We consider the 1D shallow water equations:
\begin{equation}
    \frac{\partial h}{\partial t} + \frac{\partial (hu)}{\partial x} = 0, \quad
    \frac{\partial (hu)}{\partial t} + \frac{\partial}{\partial x}\left(hu^2 + \frac{1}{2}gh^2\right) = 0,
\end{equation}
on $[0,10]$, where $h$ is the water height, $u$ is the velocity, and $g$ is gravitational acceleration.
We use an operator split into (i) advection and (ii) gravity-wave primitives.

Training samples use $g \in [9.0, 11.0]$ and five initial-condition classes: Gaussian wave perturbations, smoothed dam-break profiles (tanh transitions), smooth Fourier superpositions, smoothed step transitions, and rarefaction waves (10{,}000 samples total).
HyCOP is trained with variable target time $T \in [0.15, 0.4]$, while FNO and Loc. Int. Diff. FNO use fixed $T=0.3$.
All methods are evaluated on the same fixed-time test set at $T=0.3$.
OOD tests extrapolate in gravity ($g \in [7.0, 9.0] \cup [11.0, 13.0]$) and initial conditions (extreme Froude numbers, transcritical flows, hydraulic jumps, and high-frequency standing waves).

The HyCOP policy network uses four physics-based features: maximum Froude number, height variance, momentum variance, and $T$.
FNO and Loc. Int. Diff. FNO input channels encode $(h_0,hu_0)$, the spatial grid, and normalized gravity $g$.

\begin{figure}[H]
    \centering
    \includegraphics[width=0.95\textwidth]{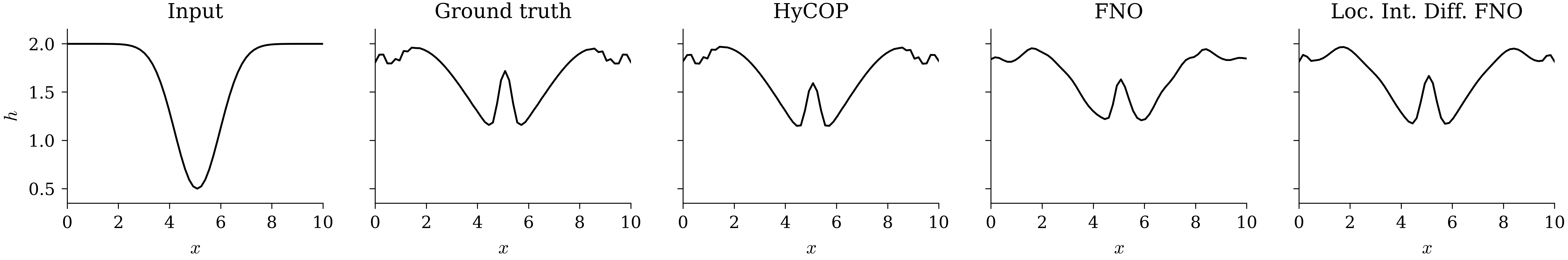}
    \caption{1D SWE qualitative example (ID).}
    \label{fig:1d_swe_id}
\end{figure}

\begin{figure}[H]
    \centering
    \includegraphics[width=0.95\textwidth]{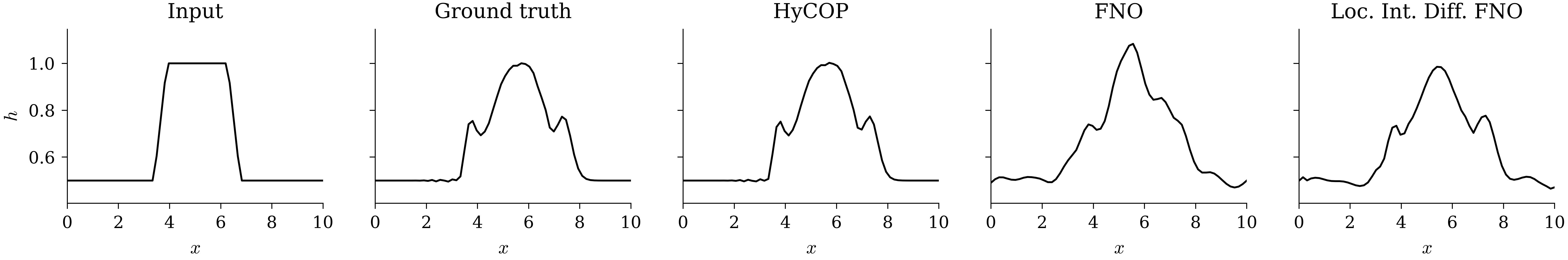}
    \caption{1D SWE qualitative example (OOD).}
    \label{fig:1d_swe_ood}
\end{figure}

\subsubsection{Viscous Burgers}
\label{app:1d_burgers}

We consider the 1D viscous Burgers equation on $x\in[0,2]$,
\begin{equation}
    \frac{\partial u}{\partial t} + u\,\frac{\partial u}{\partial x} = \nu\,\frac{\partial^2 u}{\partial x^2},
\end{equation}
with the canonical split into nonlinear advection and viscous diffusion primitives.

\paragraph{Data generation.}
Training trajectories are generated with viscosity $\nu\in[0.005,0.1]$.
Initial conditions are sampled from six families: step functions, sinusoidal waves, Gaussian pulses, sawtooth waves, smooth $\tanh$ transitions, and random low-frequency Fourier series.
We generate 10{,}000 training samples and evaluate on held-out ID\ test data at a fixed target time $T=0.5$.

\paragraph{OOD evaluation.}
OOD test cases extrapolate viscosity ($\nu\in[0.002,0.005]\cup[0.1,0.2]$) and initial conditions (sharper/smoother transitions, larger amplitudes, multiple steps, and higher-frequency oscillations).

\paragraph{Model inputs and conditioning.}
HyCOP is trained with variable target time $T$ and evaluated by querying $T=0.5$.
FNO and Loc. Int. Diff. FNO are trained and evaluated at the same fixed target time $T=0.5$.
The HyCOP policy network uses four physics-inspired features: a Reynolds-number proxy, gradient strength, velocity amplitude, and $T$.
FNO and Loc. Int. Diff. FNO use input channels for $u_0$, the spatial grid, and normalized viscosity $\nu$.

\begin{figure}[H]
    \centering
    \includegraphics[width=\linewidth]{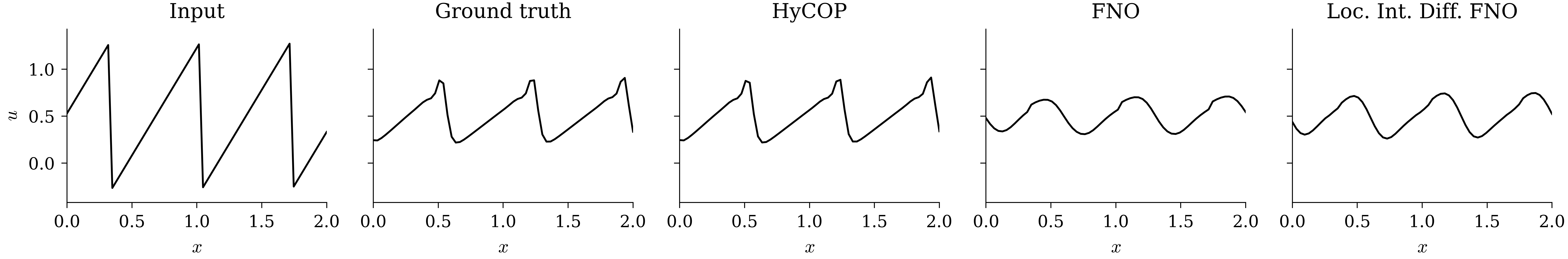}
    \caption{1D viscous Burgers (ID): qualitative comparison at $T=0.5$.}
    \label{fig:1d_burgers_id}
\end{figure}

\begin{figure}[H]
    \centering
    \includegraphics[width=\linewidth]{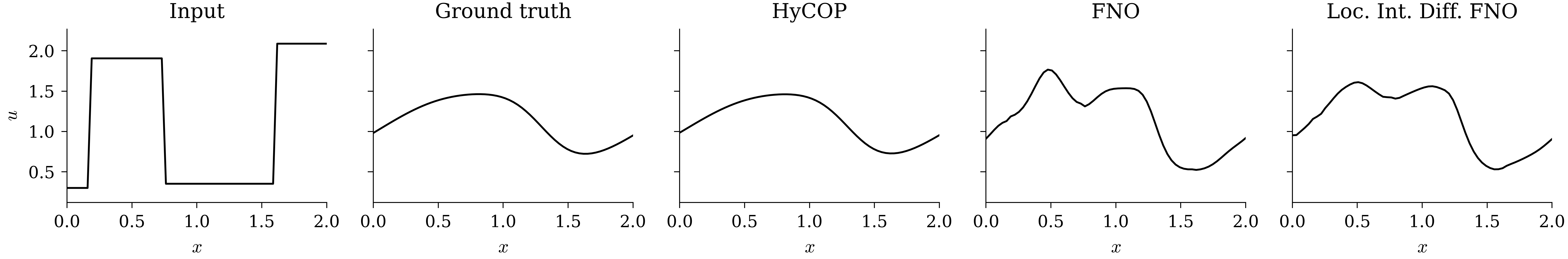}
    \caption{1D viscous Burgers (OOD): extrapolation in $\nu$ and IC family at $T=0.5$.}
    \label{fig:1d_burgers_ood}
\end{figure}


\begin{table}[t]
\caption{\textbf{1D benchmark results.}
All entries are reported on the respective 1D test sets.
``Loc.\ Int.\ Diff.\ FNO'' is abbreviated as ``LIDFNO'' for space.
HyCOP uses numerical primitives throughout.}
\label{tab:1d_benchmarks}
\centering

\vspace{0.3em}
{\footnotesize\textbf{(a) 1D Advection--Diffusion}}\\[3pt]
\scriptsize
\setlength{\tabcolsep}{2.5pt}
\begin{tabular}{l cccccccc}
\toprule
& Rel.\ $L^2$ & fRMSE$_{\text{low}}$ & fRMSE$_{\text{mid}}$ & fRMSE$_{\text{high}}$ & RMSE & Max Err & bRMSE & cRMSE \\
\midrule
\multicolumn{9}{l}{\emph{In-distribution}} \\
FNO       & 1.59\sn{-2} & 1.73\sn{-1} & 1.01\sn{-1} & 1.70\sn{-2} & 2.14\sn{-2} & 5.02\sn{-1} & 2.99\sn{-2} & 3.55\sn{-3} \\
LIDFNO    & \textbf{1.34\sn{-2}} & \textbf{1.57\sn{-1}} & 8.20\sn{-2} & 1.28\sn{-2} & \textbf{1.80\sn{-2}} & 5.30\sn{-1} & \textbf{2.42\sn{-2}} & 2.82\sn{-3} \\
\textbf{HyCOP} & 2.97\sn{-2} & 4.56\sn{-1} & \textbf{3.02\sn{-2}} & \textbf{1.58\sn{-3}} & 3.32\sn{-2} & \textbf{4.16\sn{-1}} & 3.03\sn{-2} & \textbf{4.56\sn{-8}} \\
\midrule
\multicolumn{9}{l}{\emph{Out-of-distribution}} \\
FNO       & 3.43\sn{-1} & 3.83\sn{0}  & 2.14\sn{0}  & 8.60\sn{-1} & 4.92\sn{-1} & 7.97\sn{0}  & 4.63\sn{-1} & 1.91\sn{-1} \\
LIDFNO    & 3.33\sn{-1} & 3.41\sn{0}  & 2.15\sn{0}  & 8.48\sn{-1} & 4.59\sn{-1} & 8.01\sn{0}  & 4.30\sn{-1} & 1.07\sn{-1} \\
\textbf{HyCOP} & \textbf{9.06\sn{-2}} & \textbf{4.73\sn{-1}} & \textbf{9.82\sn{-1}} & \textbf{2.84\sn{-1}} & \textbf{1.99\sn{-1}} & \textbf{1.79\sn{0}} & \textbf{1.87\sn{-1}} & \textbf{1.04\sn{-7}} \\
\bottomrule
\end{tabular}

\vspace{0.8em}

{\footnotesize\textbf{(b) 1D Viscous Burgers}}\\[3pt]
\scriptsize
\setlength{\tabcolsep}{2.5pt}
\begin{tabular}{l cccccccc}
\toprule
& Rel.\ $L^2$ & fRMSE$_{\text{low}}$ & fRMSE$_{\text{mid}}$ & fRMSE$_{\text{high}}$ & RMSE & Max Err & bRMSE & cRMSE \\
\midrule
\multicolumn{9}{l}{\emph{In-distribution}} \\
FNO       & 2.52\sn{-2} & 2.59\sn{-1} & 7.87\sn{-2} & 1.53\sn{-2} & 2.59\sn{-2} & 1.29\sn{0}  & 2.59\sn{-2} & 6.41\sn{-3} \\
LIDFNO    & 1.83\sn{-2} & 1.82\sn{-1} & 7.51\sn{-2} & 1.29\sn{-2} & 2.10\sn{-2} & 1.12\sn{0}  & 2.10\sn{-2} & 3.72\sn{-3} \\
\textbf{HyCOP} & \textbf{2.05\sn{-3}} & \textbf{2.44\sn{-2}} & \textbf{4.16\sn{-3}} & \textbf{1.10\sn{-3}} & \textbf{2.13\sn{-3}} & \textbf{1.19\sn{-1}} & \textbf{2.21\sn{-3}} & \textbf{2.90\sn{-8}} \\
\midrule
\multicolumn{9}{l}{\emph{Out-of-distribution}} \\
FNO       & 1.87\sn{-1} & 2.32\sn{0}  & 1.25\sn{0}  & 4.38\sn{-1} & 2.67\sn{-1} & 7.55\sn{0}  & 3.06\sn{-1} & 2.88\sn{-2} \\
LIDFNO    & 1.77\sn{-1} & 2.18\sn{0}  & 1.24\sn{0}  & 4.37\sn{-1} & 2.60\sn{-1} & 8.09\sn{0}  & 3.06\sn{-1} & 3.27\sn{-2} \\
\textbf{HyCOP} & \textbf{9.61\sn{-3}} & \textbf{7.30\sn{-2}} & \textbf{1.08\sn{-1}} & \textbf{7.66\sn{-2}} & \textbf{2.79\sn{-2}} & \textbf{1.18\sn{0}} & \textbf{3.11\sn{-2}} & \textbf{5.35\sn{-8}} \\
\bottomrule
\end{tabular}

\vspace{0.8em}

{\footnotesize\textbf{(c) 1D Shallow-Water Equations}}\\[3pt]
\scriptsize
\setlength{\tabcolsep}{2.5pt}
\begin{tabular}{l cccccccc}
\toprule
& Rel.\ $L^2$ & fRMSE$_{\text{low}}$ & fRMSE$_{\text{mid}}$ & fRMSE$_{\text{high}}$ & RMSE & Max Err & bRMSE & cRMSE \\
\midrule
\multicolumn{9}{l}{\emph{In-distribution}} \\
FNO       & 4.32\sn{-2} & 2.01\sn{-1} & 2.12\sn{-1} & 1.35\sn{-1} & 4.64\sn{-2} & \textbf{7.51\sn{-1}} & 7.29\sn{-2} & 5.56\sn{-3} \\
LIDFNO    & 3.05\sn{-2} & \textbf{1.16\sn{-1}} & 1.19\sn{-1} & 1.14\sn{-1} & \textbf{3.66\sn{-2}} & 7.72\sn{-1} & \textbf{3.09\sn{-2}} & 3.16\sn{-3} \\
\textbf{HyCOP} & \textbf{2.81\sn{-2}} & 2.33\sn{-1} & \textbf{9.02\sn{-2}} & \textbf{3.02\sn{-2}} & 4.49\sn{-2} & 1.04\sn{0} & 5.51\sn{-2} & \textbf{4.29\sn{-8}} \\
\midrule
\multicolumn{9}{l}{\emph{Out-of-distribution}} \\
FNO       & 1.86\sn{-1} & 1.26\sn{0}  & 5.96\sn{-1} & 1.53\sn{-1} & 3.58\sn{-1} & 5.29\sn{0}  & 3.16\sn{-1} & 1.06\sn{-1} \\
LIDFNO    & 1.54\sn{-1} & 7.05\sn{-1} & 4.73\sn{-1} & 1.38\sn{-1} & 3.09\sn{-1} & 3.88\sn{0}  & 2.39\sn{-1} & 7.83\sn{-2} \\
\textbf{HyCOP} & \textbf{3.07\sn{-2}} & \textbf{1.73\sn{-1}} & \textbf{1.18\sn{-1}} & \textbf{2.43\sn{-2}} & \textbf{5.27\sn{-2}} & \textbf{4.66\sn{-1}} & \textbf{4.68\sn{-2}} & \textbf{3.62\sn{-8}} \\
\bottomrule
\end{tabular}

\end{table}

\subsection{2D Systems (Fixed-Time Inference)}
\label{app:2d_pair}

\subsubsection{Advection--Diffusion--Reaction (Fisher--KPP)}
\label{app:2d_adr}

We consider the 2D advection--diffusion--reaction (ADR) equation with Fisher--KPP kinetics,
\begin{equation}
    \frac{\partial u}{\partial t}
    + c_x \frac{\partial u}{\partial x}
    + c_y \frac{\partial u}{\partial y}
    =
    D_x \frac{\partial^2 u}{\partial x^2}
    + D_y \frac{\partial^2 u}{\partial y^2}
    + r u(1-u),
\end{equation}
on $\Omega=[0,1]^2$. We use an operator split into advection, diffusion, and reaction primitives.

\paragraph{Data generation \& Training.}
Training samples use $c_x,c_y \in [0.2,1.5]$, $D_x,D_y \in [0.05,0.2]$, and $r \in [0.1,1.0]$.
Initial conditions are drawn from six classes: Gaussian pulses, smoothed step functions, ring patterns, diagonal stripes, multi-Gaussian superpositions, and smooth sigmoidal transitions.
HyCOP is trained with variable target time $T \in [0.1,0.35]$, baselines use fixed $T=0.2$; all methods are evaluated on the same fixed-time test set at $T=0.2$.
The HyCOP policy network uses seven physics-based features: Péclet numbers in $x$ and $y$, Damköhler number, concentration variance, gradient variance in $x$ and $y$, and $T$.
FNO, Loc.\ Int.\ Diff.\ FNO, and PINO share input channels: $u_0$, spatial grids, and normalized parameters $(c_x, c_y, D_x, D_y, r)$. PINO evaluates the ADR residual $\partial_t u + c_x \partial_x u + c_y \partial_y u - D_x \partial_{xx} u - D_y \partial_{yy} u - r\,u(1-u)$ via central differences on the training grid, with $\partial_t u$ approximated by $(u(T){-}u_0)/T$. DeepONet encodes $u_0$ through its CNN branch and the normalized parameters through the branch MLP head.

\paragraph{OOD evaluation.}
OOD tests extrapolate in parameters
($c_x,c_y \in [0.1,0.2] \cup [1.5,2.5]$,
$D_x,D_y \in [0.01,0.05] \cup [0.2,0.4]$,
$r \in [0.05,0.1] \cup [1.0,2.5]$)
and initial conditions (very narrow/wide Gaussians, near-saturation/extinction states, high-frequency patterns, and sharp fronts).

Across these ranges, ADR spans regimes where advection competes with diffusion (via Péclet variation) and reaction competes with transport (via Damköhler variation), yielding sharp traveling fronts and reaction-dominated transients, particularly in OOD.

\begin{figure}[H]
    \centering
    \includegraphics[width=0.95\textwidth]{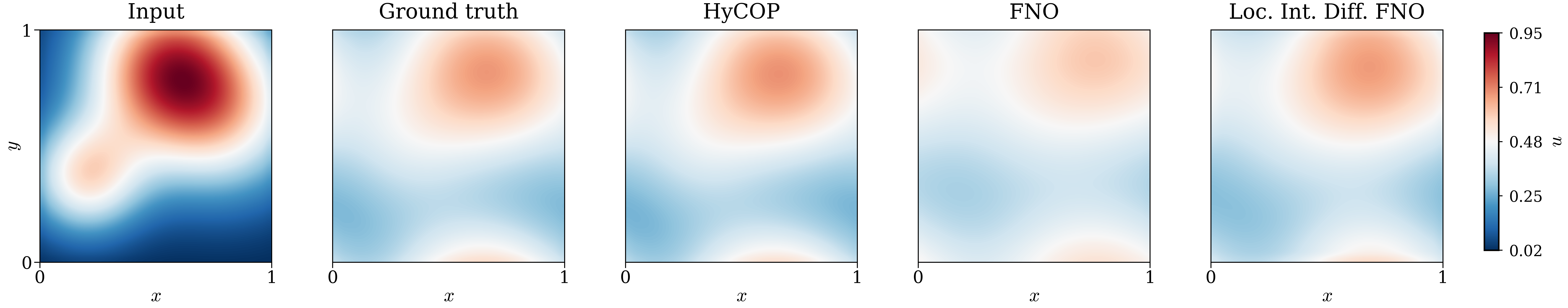}
    \caption{2D ADR qualitative example (ID).}
    \label{fig:2d_adr_id}
\end{figure}

\begin{figure}[H]
    \centering
    \includegraphics[width=0.95\textwidth]{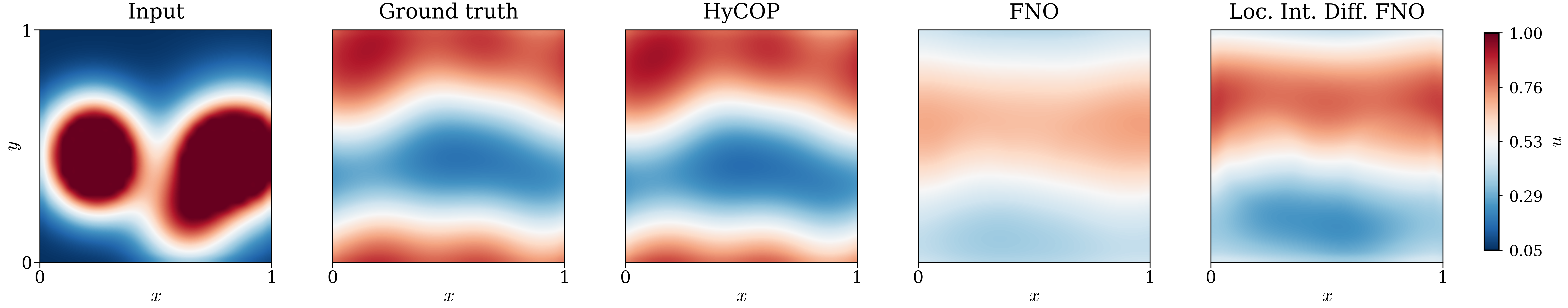}
    \caption{2D ADR qualitative example (OOD).}
    \label{fig:2d_adr_ood}
\end{figure}

\noindent\emph{Note:} cRMSE is not applicable for ADR (no conserved mass constraint under Fisher--KPP kinetics).

\begin{table}[H]
\caption{2D ADR test-set metrics (ID).}
\label{table:adr_id}
\vskip 0.15in
\begin{center}
\begin{small}
\begin{sc}
\resizebox{\textwidth}{!}{%
\begin{tabular}{lcccccccc}
\toprule
Model & Rel.\ $L^2$ & fRMSE low & fRMSE mid & fRMSE high & RMSE & Max error & bRMSE & cRMSE \\
\midrule
DeepONet  & $1.58 \cdot 10^{-1}$ & $8.49 \cdot 10^{0}$ & $3.35 \cdot 10^{-1}$ & $1.19 \cdot 10^{-1}$ & $5.95 \cdot 10^{-2}$ & $4.55 \cdot 10^{-1}$ & $6.35 \cdot 10^{-2}$ & --- \\
FNO       & $8.42 \cdot 10^{-2}$ & $4.46 \cdot 10^{0}$ & $6.41 \cdot 10^{-2}$ & $2.69 \cdot 10^{-2}$ & $3.28 \cdot 10^{-2}$ & $2.94 \cdot 10^{-1}$ & $3.21 \cdot 10^{-2}$ & --- \\
PINO      & $8.42 \cdot 10^{-2}$ & $4.47 \cdot 10^{0}$ & $5.33 \cdot 10^{-2}$ & $2.39 \cdot 10^{-2}$ & $3.29 \cdot 10^{-2}$ & $2.95 \cdot 10^{-1}$ & $3.20 \cdot 10^{-2}$ & --- \\
Loc. Int. Diff. FNO  & $3.15 \cdot 10^{-2}$ & $1.63 \cdot 10^{0}$ & $1.03 \cdot 10^{-1}$ & $2.03 \cdot 10^{-2}$ & $1.22 \cdot 10^{-2}$ & $1.37 \cdot 10^{-1}$ & $1.25 \cdot 10^{-2}$ & --- \\
\textbf{HyCOP (ours)}  & $\bm{2.10 \cdot 10^{-2}}$ & $\bm{1.12 \cdot 10^{0}}$ & $\bm{1.30 \cdot 10^{-4}}$ & $\bm{3.80 \cdot 10^{-7}}$ & $\bm{8.34 \cdot 10^{-3}}$ & $\bm{5.91 \cdot 10^{-2}}$ & $\bm{8.17 \cdot 10^{-3}}$ & --- \\
\bottomrule
\end{tabular}
}
\end{sc}
\end{small}
\end{center}
\vskip -0.1in
\end{table}

\begin{table}[H]
\caption{2D ADR test-set metrics (OOD).}
\label{table:adr_ood}
\vskip 0.15in
\begin{center}
\begin{small}
\begin{sc}
\resizebox{\textwidth}{!}{%
\begin{tabular}{lcccccccc}
\toprule
Model & Rel.\ $L^2$ & fRMSE low & fRMSE mid & fRMSE high & RMSE & Max error & bRMSE & cRMSE \\
\midrule
DeepONet  & $2.82 \cdot 10^{-1}$ & $1.34 \cdot 10^{1}$ & $3.90 \cdot 10^{-1}$ & $1.21 \cdot 10^{-1}$ & $1.12 \cdot 10^{-1}$ & $7.79 \cdot 10^{-1}$ & $1.02 \cdot 10^{-1}$ & --- \\
FNO       & $2.60 \cdot 10^{-1}$ & $1.20 \cdot 10^{1}$ & $1.46 \cdot 10^{-1}$ & $7.00 \cdot 10^{-2}$ & $9.58 \cdot 10^{-2}$ & $6.89 \cdot 10^{-1}$ & $9.05 \cdot 10^{-2}$ & --- \\
PINO      & $2.35 \cdot 10^{-1}$ & $1.18 \cdot 10^{1}$ & $1.31 \cdot 10^{-1}$ & $5.99 \cdot 10^{-2}$ & $9.59 \cdot 10^{-2}$ & $6.95 \cdot 10^{-1}$ & $9.08 \cdot 10^{-2}$ & --- \\
Loc. Int. Diff. FNO  & $1.89 \cdot 10^{-1}$ & $9.66 \cdot 10^{0}$ & $2.45 \cdot 10^{-1}$ & $7.23 \cdot 10^{-2}$ & $7.86 \cdot 10^{-2}$ & $6.56 \cdot 10^{-1}$ & $7.81 \cdot 10^{-2}$ & --- \\
\textbf{HyCOP (ours)}  & $\bm{2.87 \cdot 10^{-2}}$ & $\bm{1.56 \cdot 10^{0}}$ & $\bm{2.48 \cdot 10^{-2}}$ & $\bm{3.34 \cdot 10^{-4}}$ & $\bm{1.87 \cdot 10^{-2}}$ & $\bm{3.48 \cdot 10^{-1}}$ & $\bm{1.88 \cdot 10^{-2}}$ & --- \\
\bottomrule
\end{tabular}
}
\end{sc}
\end{small}
\end{center}
\vskip -0.1in
\end{table}

\subsubsection{Shallow Water Equations}
\label{app:2d_swe}

We consider the 2D shallow water equations on $\Omega=[0,10]^2$,
\begin{align}
    \frac{\partial h}{\partial t}
    + \frac{\partial (hu)}{\partial x}
    + \frac{\partial (hv)}{\partial y} &= 0, \\
    \frac{\partial (hu)}{\partial t}
    + \frac{\partial}{\partial x}\!\left(hu^2 + \frac{1}{2}gh^2\right)
    + \frac{\partial (huv)}{\partial y} &= 0, \\
    \frac{\partial (hv)}{\partial t}
    + \frac{\partial (huv)}{\partial x}
    + \frac{\partial}{\partial y}\!\left(hv^2 + \frac{1}{2}gh^2\right) &= 0,
\end{align}
where $h$ is water depth, $(u,v)$ is velocity, and $g$ is gravitational acceleration.
We use a canonical split into (i) advective transport and (ii) gravity-wave (pressure) forcing primitives.

\paragraph{Data generation \& training.}
Training samples use $g \in [9.0,11.0]$ with six initial-condition families:
Gaussian wave perturbations, smoothed dam-break profiles (vertical and horizontal), smooth Fourier superpositions,
smoothed oblique transitions, vortex patterns, and radial dam-break configurations.
HyCOP is trained with variable target time $T \in [0.15,0.4]$, baselines are trained at fixed $T=0.3$.
For fair comparison, all methods are evaluated on the same fixed-time test set at $T=0.3$.

The HyCOP policy network uses six physics-based features: maximum Froude numbers in $x$ and $y$,
height variance, $x$-momentum variance, $y$-momentum variance, and $T$.
FNO, Loc.\ Int.\ Diff.\ FNO, and PINO share input channels: $(h_0, hu_0, hv_0)$, spatial grids, and normalized gravity $g$. PINO evaluates the SWE residuals (mass and $x$-/$y$-momentum conservation) via central differences with periodic boundaries, with $\partial_t u$ approximated by $(u(T){-}u_0)/T$. DeepONet encodes the three initial-state channels through its CNN branch and normalized $g$ through the branch MLP head.

\paragraph{OOD evaluation.}
OOD tests extrapolate in gravity ($g \in [7.0,9.0] \cup [11.0,13.0]$) and initial conditions,
including extreme Froude regimes, transcritical flows, high-frequency standing waves, strong vortices,
and extreme radial dam-break configurations.

These OOD regimes include near-critical/transcritical configurations (large or spatially varying Froude) that amplify nonlinear wave interactions, steep gradients, and boundary reflections, stressing long-horizon stability.

\begin{figure}[H]
    \centering
    \includegraphics[width=0.95\textwidth]{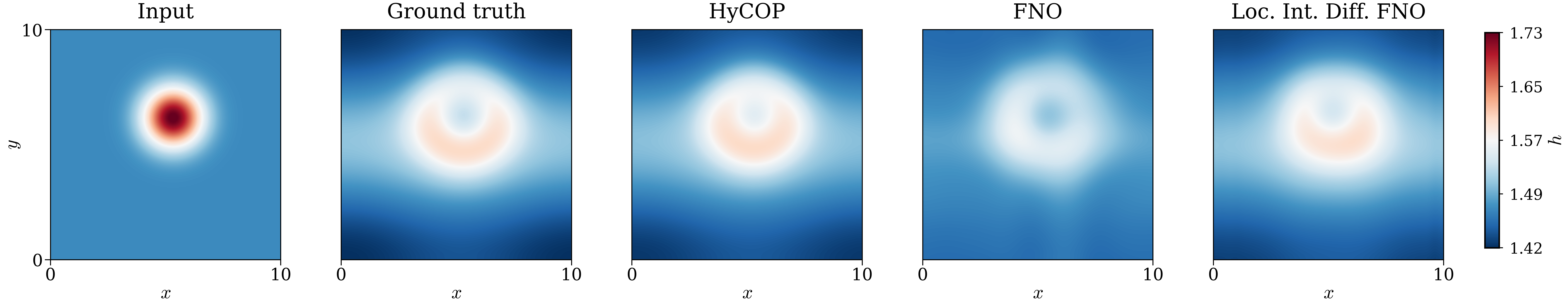}
    \caption{2D SWE qualitative example (ID).}
    \label{fig:2d_swe_id}
\end{figure}

\begin{figure}[H]
    \centering
    \includegraphics[width=0.95\textwidth]{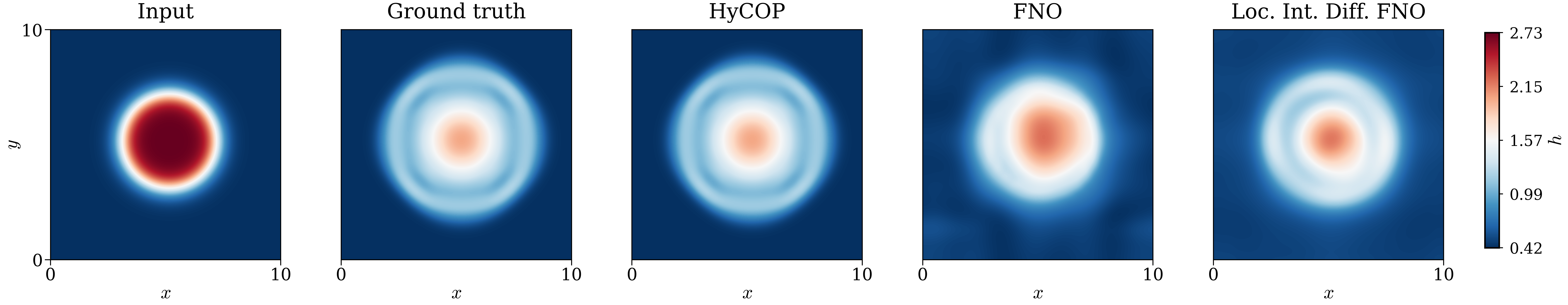}
    \caption{2D SWE qualitative example (OOD).}
    \label{fig:2d_swe_ood}
\end{figure}

\begin{table}[H]
\caption{2D SWE test-set metrics (ID).}
\label{table:app_sw_id}
\vskip 0.15in
\begin{center}
\begin{small}
\begin{sc}
\resizebox{\textwidth}{!}{%
\begin{tabular}{lcccccccc}
\toprule
Model & Rel.\ $L^2$ & fRMSE low & fRMSE mid & fRMSE high & RMSE & Max error & bRMSE & cRMSE \\
\midrule
DeepONet  & $3.89 \cdot 10^{-1}$ & $3.95 \cdot 10^{0}$ & $1.10 \cdot 10^{0}$ & $3.82 \cdot 10^{-1}$ & $6.47 \cdot 10^{-2}$ & $1.38 \cdot 10^{0}$ & $7.23 \cdot 10^{-2}$ & $3.62 \cdot 10^{-3}$ \\
FNO       & $1.19 \cdot 10^{-1}$ & $1.33 \cdot 10^{0}$ & $6.57 \cdot 10^{-1}$ & $3.69 \cdot 10^{-1}$ & $2.86 \cdot 10^{-2}$ & $5.80 \cdot 10^{-1}$ & $3.86 \cdot 10^{-2}$ & $3.50 \cdot 10^{-3}$ \\
PINO      & $1.17 \cdot 10^{-1}$ & $1.33 \cdot 10^{0}$ & $6.49 \cdot 10^{-1}$ & $3.70 \cdot 10^{-1}$ & $2.86 \cdot 10^{-2}$ & $5.58 \cdot 10^{-1}$ & $3.82 \cdot 10^{-2}$ & $3.68 \cdot 10^{-3}$ \\
Loc. Int. Diff. FNO  & $7.49 \cdot 10^{-2}$ & $\bm{6.57 \cdot 10^{-1}}$ & $1.77 \cdot 10^{-1}$ & $8.61 \cdot 10^{-2}$ & $1.25 \cdot 10^{-2}$ & $4.26 \cdot 10^{-1}$ & $\bm{1.48 \cdot 10^{-2}}$ & $2.51 \cdot 10^{-3}$ \\
\textbf{HyCOP (ours)}  & $\bm{2.40 \cdot 10^{-2}}$ & $7.23 \cdot 10^{-1}$ & $\bm{1.02 \cdot 10^{-1}}$ & $\bm{2.66 \cdot 10^{-2}}$ & $\bm{1.03 \cdot 10^{-2}}$ & $\bm{1.90 \cdot 10^{-1}}$ & $1.69 \cdot 10^{-2}$ & $\bm{2.16 \cdot 10^{-8}}$ \\
\bottomrule
\end{tabular}
}
\end{sc}
\end{small}
\end{center}
\vskip -0.1in
\end{table}

\begin{table}[H]
\caption{2D SWE test-set metrics (OOD).}
\label{table:app_sw_ood}
\vskip 0.15in
\begin{center}
\begin{small}
\begin{sc}
\resizebox{\textwidth}{!}{%
\begin{tabular}{lcccccccc}
\toprule
Model & Rel.\ $L^2$ & fRMSE low & fRMSE mid & fRMSE high & RMSE & Max error & bRMSE & cRMSE \\
\midrule
DeepONet  & $5.61 \cdot 10^{-1}$ & $3.10 \cdot 10^{1}$ & $3.35 \cdot 10^{0}$ & $5.67 \cdot 10^{-1}$ & $5.52 \cdot 10^{-1}$ & $4.32 \cdot 10^{0}$ & $4.55 \cdot 10^{-1}$ & $2.27 \cdot 10^{-1}$ \\
FNO       & $3.80 \cdot 10^{-1}$ & $1.69 \cdot 10^{1}$ & $3.03 \cdot 10^{0}$ & $5.38 \cdot 10^{-1}$ & $4.15 \cdot 10^{-1}$ & $3.99 \cdot 10^{0}$ & $3.60 \cdot 10^{-1}$ & $2.27 \cdot 10^{-1}$ \\
PINO      & $3.83 \cdot 10^{-1}$ & $1.77 \cdot 10^{1}$ & $3.01 \cdot 10^{0}$ & $5.48 \cdot 10^{-1}$ & $4.18 \cdot 10^{-1}$ & $3.97 \cdot 10^{0}$ & $3.65 \cdot 10^{-1}$ & $2.37 \cdot 10^{-1}$ \\
Loc. Int. Diff. FNO  & $3.54 \cdot 10^{-1}$ & $2.22 \cdot 10^{1}$ & $1.78 \cdot 10^{0}$ & $3.76 \cdot 10^{-1}$ & $3.53 \cdot 10^{-1}$ & $3.46 \cdot 10^{0}$ & $3.16 \cdot 10^{-1}$ & $1.79 \cdot 10^{-1}$ \\
\textbf{HyCOP (ours)}  & $\bm{5.00 \cdot 10^{-2}}$ & $\bm{1.25 \cdot 10^{0}}$ & $\bm{3.17 \cdot 10^{-1}}$ & $\bm{5.54 \cdot 10^{-2}}$ & $\bm{5.00 \cdot 10^{-2}}$ & $\bm{3.46 \cdot 10^{-1}}$ & $\bm{4.35 \cdot 10^{-2}}$ & $\bm{3.03 \cdot 10^{-8}}$ \\
\bottomrule
\end{tabular}
}
\end{sc}
\end{small}
\end{center}
\vskip -0.1in
\end{table}

\subsubsection{Compressible Navier--Stokes (PDEBench)}
\label{app:ns}

We evaluate single-step prediction on the 2D compressible Navier--Stokes benchmark from PDEBench~\citep{PDEBench2022} at Mach $M{=}0.1$ and dynamic/bulk viscosity $\eta{=}\zeta{=}0.1$. The task is to predict density, velocity $(V_x, V_y)$, and pressure at $T{=}0.05$ from the initial state. Data is used as released by PDEBench; we do not regenerate trajectories.

\paragraph{Data and splits.}
We use the PDEBench 2D compressible Navier--Stokes dataset at resolution $128\times 128$ with the train/test split provided by~\citet{PDEBench2022}. All methods train on the same data and are evaluated on the same held-out set. No OOD extrapolation is performed on NS: PDEBench does not provide controlled OOD splits (parameter extrapolation, long-horizon rollout, boundary shift), which is why we designed the SWE and ADR benchmarks.

\paragraph{HyCOP dictionaries.}
HyCOP uses a two-primitive dictionary corresponding to the advective and viscous sub-flows of the compressible NS system:
\[
\mathbb{D}_{\mathrm{NS}}=\{\mathcal{O}_{\mathrm{adv}},\mathcal{O}_{\mathrm{diff}}\},
\]
where $\mathcal{O}_{\mathrm{adv}}$ advances the inviscid Euler flux via a fourth-order Runge--Kutta (RK4) step and $\mathcal{O}_{\mathrm{diff}}$ advances the viscous diffusion term in Fourier space. Both primitives are textbook routines with zero learnable parameters. We report two configurations:
\begin{itemize}
    \item \textbf{HyCOP.} Both primitives are numerical (RK4 advection + spectral diffusion). The only learnable component is the $\sim$50-parameter policy.
    \item \textbf{HyCOP-Hyb.} The advective primitive is replaced by a time-queried FNO-FiLM surrogate (see~\ref{app:fnofilm_details}) pretrained on single-process advection data and frozen; the diffusion primitive remains numerical.
\end{itemize}
HyCOP's policy uses the same physics-based features as our other 2D benchmarks (regime indicators and state statistics).

\paragraph{Strang baseline.}
Strang splitting uses the same two primitives as HyCOP with a fixed Strang schedule: $\mathcal{O}_{\mathrm{diff}}^{\tau/2} \circ \mathcal{O}_{\mathrm{adv}}^{\tau} \circ \mathcal{O}_{\mathrm{diff}}^{\tau/2}$. This isolates the effect of the learned schedule: any performance gap between HyCOP and Strang is attributable entirely to the composition policy, since primitives and dictionary are identical.

\paragraph{Monolithic baselines.}
DeepONet~\citep{lu2021deeponet} and PINO~\citep{li2021pino} follow the common setup (Section~\ref{app:common_setup}): shared Adam optimizer, 500-epoch budget, and the DeepONet and PINO architectural choices described there. DeepONet's branch encodes the four-channel initial state $(\rho_0, V_{x,0}, V_{y,0}, p_0)$; the trunk encodes query coordinates $(x,y)$. PINO reuses the FNO backbone and augments the data loss with the compressible NS residual (continuity, momentum, and energy) via central differences with periodic boundaries, with $\partial_t$ approximated by $(u(T){-}u_0)/T$. FNO and U-Net numbers in Table~\ref{tab:regime_a}(a) are PDEBench-reported~\citep{PDEBench2022}; we do not retrain these models.

\paragraph{Metrics.}
Following PDEBench conventions, we report nRMSE and cRMSE. nRMSE is computed per-channel and averaged across channels; cRMSE measures conservation error via channel-wise diagnostic integrals ($\int \rho\, dx\, dy$ for mass, and the corresponding momentum integrals). Both HyCOP variants and Strang inherit cRMSE at machine precision from the numerical primitives; monolithic baselines have no such guarantee.

\subsection{2D Systems (Trajectory)}
\label{app:2d_trajectory}

Trajectory experiments evaluate long-horizon prediction on the same in-distribution (ID) and out-of-distribution (OOD) splits as the fixed-time setting; see Sections~\ref{app:2d_adr} and~\ref{app:2d_swe} for parameter ranges and OOD construction. We report errors at rollout horizons of 1, 5, 10, and 20 steps.

\subsubsection{Advection--Diffusion--Reaction}
\label{app:2d_adr_trajectory}

We use the same 2D ADR system as in Section~\ref{app:2d_adr}. Trajectories are generated by integrating the reference solver from $t=0$ to $T_{\mathrm{final}}=0.4$ with step size $\Delta t=0.02$, yielding 21 snapshots per trajectory. We generate 10{,}000 trajectories total.

\paragraph{HyCOP (non-autoregressive multi-time querying).}
HyCOP uses the same policy architecture and feature configuration as in the fixed-time experiments (Section~\ref{app:2d_adr}). Given $(u_0,\mu)$ and a query time (or set of times) $T$, the policy outputs a split program and time ratios, and the resulting composite operator is evaluated \emph{once} to produce $u(T)$. Thus HyCOP does not require step-by-step rollout.

\paragraph{Autoregressive baselines.}
Loc. Int. Diff. FNO~\citep{liuschiaffini2024localfno} is trained for single-step prediction at fixed $\Delta t=0.02$ and rolled out autoregressively at test time. Inputs include $u(t)$, spatial grids, and normalized parameters (8 channels total), and the network predicts $u(t+\Delta t)$; the architecture matches Section~\ref{app:2d_adr}.

For AR-Loc. Int. Diff. FNO we follow the autoregressive training protocol of~\citet{tran2023ffno}: teacher forcing with ground-truth inputs, small Gaussian noise injection ($\sigma=0.001$) for robustness, cosine learning-rate decay, batch size 32, and 100 epochs.

U-Net~\citep{ronneberger2015unet} serves as a convolutional sequence baseline. We use a standard encoder--decoder with skip connections (4 pooling levels, initial width 16), the same 8-channel inputs as Loc. Int. Diff. FNO, and the same autoregressive training protocol (teacher forcing, noise $\sigma=0.001$, cosine decay, 100 epochs).

\begin{figure}[t]
    \centering
    \includegraphics[width=0.95\textwidth]{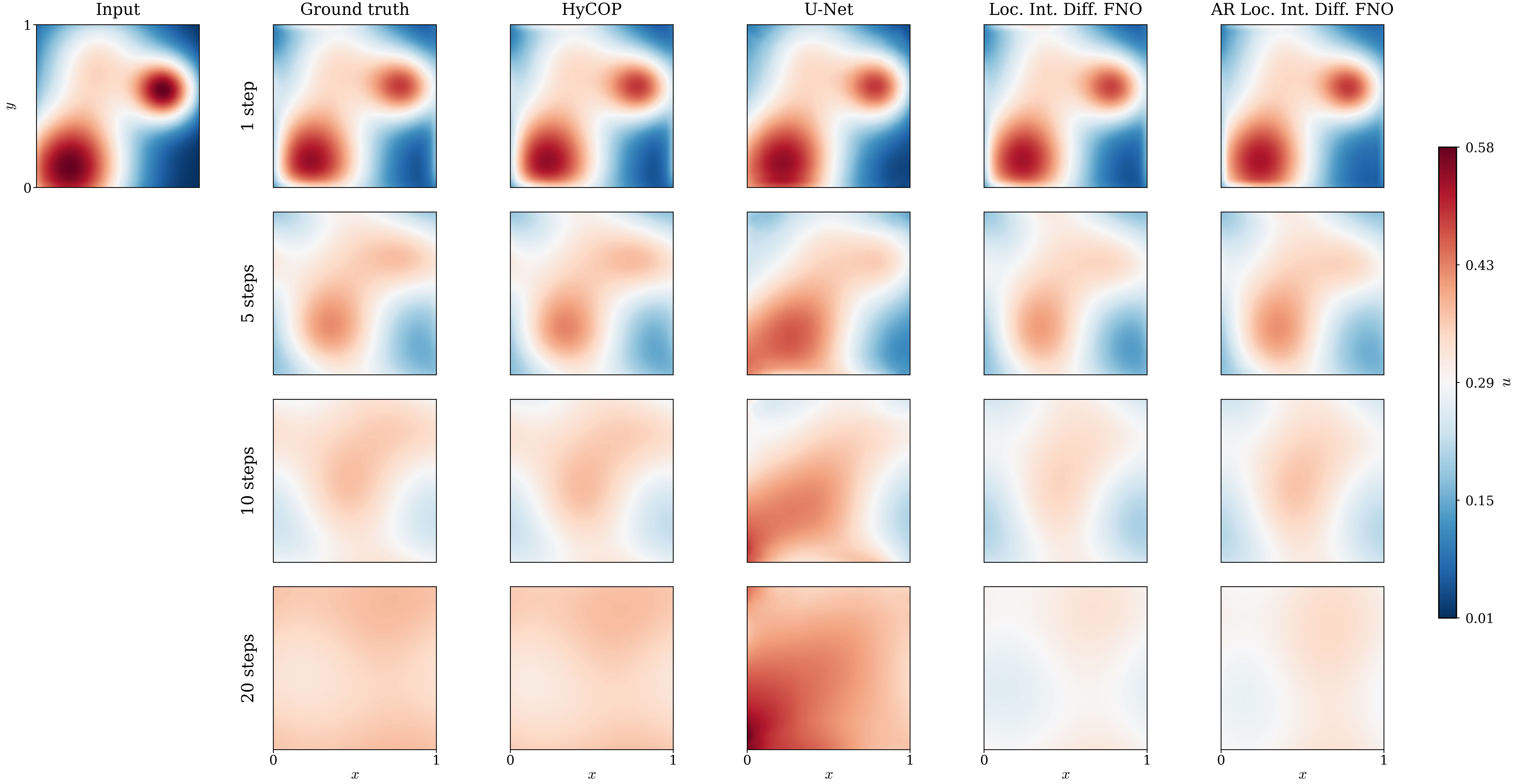}
    \caption{\textbf{2D ADR trajectory (ID).} Example long-horizon prediction at horizons of 1/5/10/20 steps for the 2D ADR system (Section~\ref{app:2d_adr_trajectory}).}
    \label{fig:2d-adr-traj-id}
\end{figure}

\begin{figure}[t]
    \centering
    \includegraphics[width=0.95\textwidth]{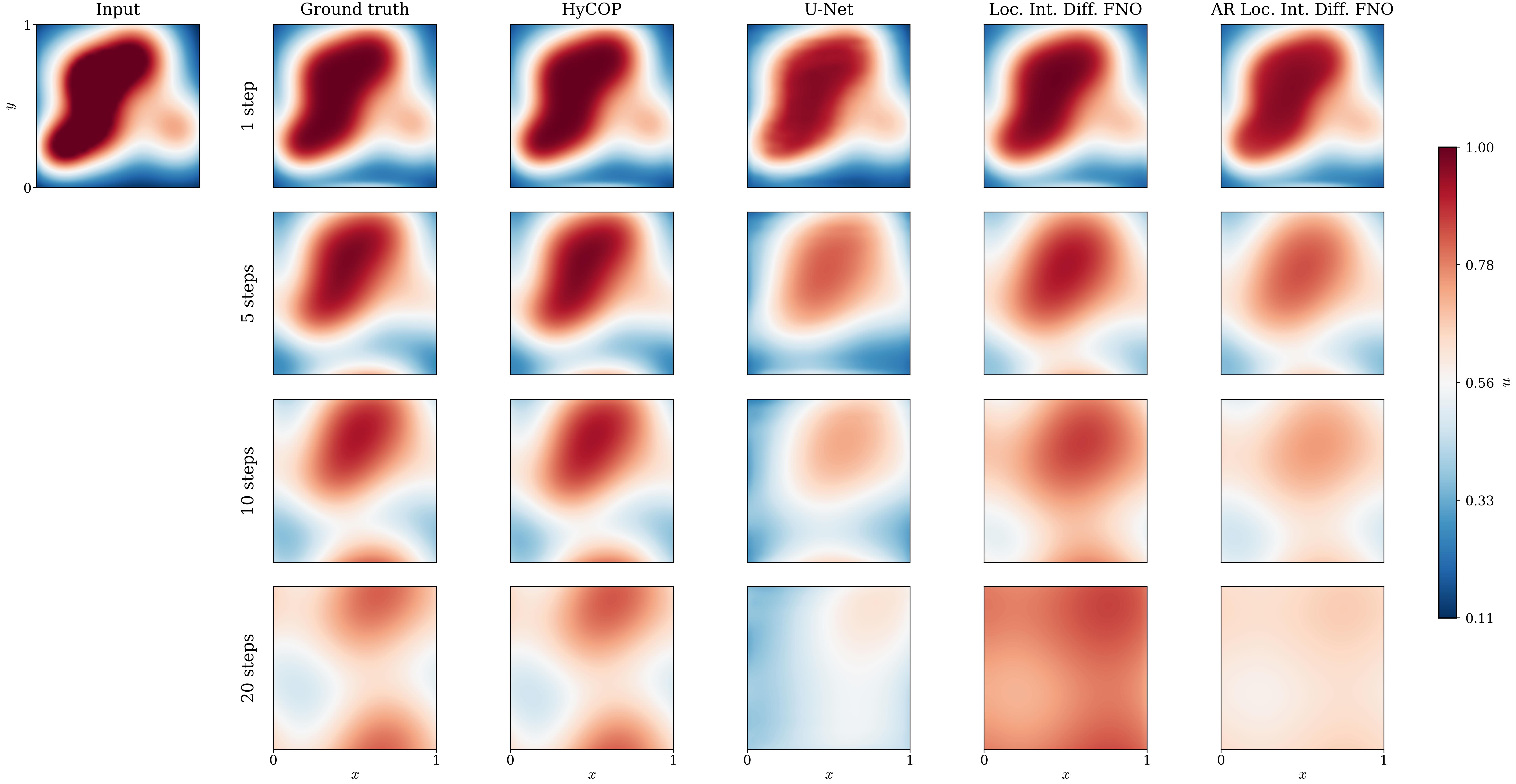}
    \caption{\textbf{2D ADR trajectory (OOD).} Same visualization as Figure~\ref{fig:2d-adr-traj-id} on OOD initial conditions/parameters (see Section~\ref{app:2d_adr}).}
    \label{fig:2d-adr-traj-ood}
\end{figure}

\begin{table}[H]
\caption{\textbf{2D ADR multi-step trajectory error (ID).} Relative $L^2$ and RMSE at horizons of 1/5/10/20 steps. Autoregressive models are rolled out; HyCOP answers each horizon via direct time querying (non-autoregressive).}
\label{tab:app_adr_multistep_id}
\vskip 0.15in
\begin{center}
\begin{small}
\begin{sc}
\resizebox{\textwidth}{!}{%
\begin{tabular}{lcccccccc}
\toprule
\multirow{2}{*}{Model} &
\multicolumn{2}{c}{1 step} & 
\multicolumn{2}{c}{5 steps} &
\multicolumn{2}{c}{10 steps} &
\multicolumn{2}{c}{20 steps} \\
\cmidrule(lr){2-3} \cmidrule(lr){4-5} \cmidrule(lr){6-7} \cmidrule(lr){8-9}
 & Rel.\ $L^2$ & RMSE & Rel.\ $L^2$ & RMSE & Rel.\ $L^2$ & RMSE & Rel.\ $L^2$ & RMSE \\
\midrule
U-Net        & $9.42 \cdot 10^{-2}$ & $3.59 \cdot 10^{-2}$ & $1.70 \cdot 10^{-1}$ & $5.85 \cdot 10^{-2}$ & $2.09 \cdot 10^{-1}$ & $6.82 \cdot 10^{-2}$ & $2.30 \cdot 10^{-1}$ & $7.38 \cdot 10^{-2}$ \\
Loc. Int. Diff. FNO     & $2.87 \cdot 10^{-2}$ & $1.03 \cdot 10^{-2}$ & $7.78 \cdot 10^{-2}$ & $2.48 \cdot 10^{-2}$ & $1.11 \cdot 10^{-1}$ & $3.44 \cdot 10^{-2}$ & $1.51 \cdot 10^{-1}$ & $4.85 \cdot 10^{-2}$ \\
AR-Loc. Int. Diff. FNO  & $3.86 \cdot 10^{-2}$ & $1.38 \cdot 10^{-2}$ & $8.24 \cdot 10^{-2}$ & $2.57 \cdot 10^{-2}$ & $1.02 \cdot 10^{-1}$ & $3.15 \cdot 10^{-2}$ & $1.02 \cdot 10^{-1}$ & $3.30 \cdot 10^{-2}$ \\
\textbf{HyCOP (ours)}
             & $\bm{6.02 \cdot 10^{-3}}$ & $\bm{2.15 \cdot 10^{-3}}$
             & $\bm{1.68 \cdot 10^{-2}}$ & $\bm{5.42 \cdot 10^{-3}}$
             & $\bm{2.10 \cdot 10^{-2}}$ & $\bm{6.65 \cdot 10^{-3}}$
             & $\bm{1.96 \cdot 10^{-2}}$ & $\bm{6.42 \cdot 10^{-3}}$ \\
\bottomrule
\end{tabular}
}
\end{sc}
\end{small}
\end{center}
\vskip -0.1in
\end{table}

\begin{table}[H]
\caption{\textbf{2D ADR multi-step trajectory error (OOD).} Same protocol as Table~\ref{tab:app_adr_multistep_id} evaluated on OOD parameter/IC shifts.}
\label{tab:app_adr_multistep_ood}
\vskip 0.15in
\begin{center}
\begin{small}
\begin{sc}
\resizebox{\textwidth}{!}{%
\begin{tabular}{lcccccccc}
\toprule
\multirow{2}{*}{Model} &
\multicolumn{2}{c}{1 step} & 
\multicolumn{2}{c}{5 steps} &
\multicolumn{2}{c}{10 steps} &
\multicolumn{2}{c}{20 steps} \\
\cmidrule(lr){2-3} \cmidrule(lr){4-5} \cmidrule(lr){6-7} \cmidrule(lr){8-9}
 & Rel.\ $L^2$ & RMSE & Rel.\ $L^2$ & RMSE & Rel.\ $L^2$ & RMSE & Rel.\ $L^2$ & RMSE \\
\midrule
U-Net        & $1.83 \cdot 10^{-1}$ & $5.05 \cdot 10^{-2}$ & $5.23 \cdot 10^{-1}$ & $9.38 \cdot 10^{-2}$ & $7.91 \cdot 10^{-1}$ & $1.20 \cdot 10^{-1}$ & $1.07 \cdot 10^{0}$ & $1.45 \cdot 10^{-1}$ \\
Loc. Int. Diff. FNO     & $1.56 \cdot 10^{-1}$ & $3.09 \cdot 10^{-2}$ & $5.65 \cdot 10^{-1}$ & $7.88 \cdot 10^{-2}$ & $9.17 \cdot 10^{-1}$ & $1.26 \cdot 10^{-1}$ & $1.30 \cdot 10^{0}$ & $2.08 \cdot 10^{-1}$ \\
AR-Loc. Int. Diff. FNO  & $1.49 \cdot 10^{-1}$ & $3.48 \cdot 10^{-2}$ & $3.97 \cdot 10^{-1}$ & $6.77 \cdot 10^{-2}$ & $5.23 \cdot 10^{-1}$ & $8.24 \cdot 10^{-2}$ & $5.72 \cdot 10^{-1}$ & $8.58 \cdot 10^{-2}$ \\
\textbf{HyCOP (ours)}
             & $\bm{8.29 \cdot 10^{-3}}$ & $\bm{2.79 \cdot 10^{-3}}$
             & $\bm{2.12 \cdot 10^{-2}}$ & $\bm{6.66 \cdot 10^{-3}}$
             & $\bm{2.87 \cdot 10^{-2}}$ & $\bm{9.29 \cdot 10^{-3}}$
             & $\bm{3.78 \cdot 10^{-2}}$ & $\bm{1.28 \cdot 10^{-2}}$ \\
\bottomrule
\end{tabular}
}
\end{sc}
\end{small}
\end{center}
\vskip -0.1in
\end{table}

\subsubsection{Shallow Water Equations}
\label{app:2d_swe_trajectory}

\paragraph{Data generation \& HyCOP.}
We use the same 2D SWE system as in Section~\ref{app:2d_swe} with identical parameter ranges and initial-condition classes.
Trajectories are generated by solving the reference PDE from $t=0$ to $T_{\text{final}}=1.0$ with time step $\Delta t=0.05$, yielding 21 snapshots per trajectory; we generate 10{,}000 trajectories total.
HyCOP uses the same policy-network configuration and feature set as in the fixed-time experiments (Section~\ref{app:2d_swe}).
Importantly, HyCOP does not rely on autoregressive rollout: given $(u_0,\mu)$ and a target time (or set of times), the learned policy outputs a composite program that maps directly to $u(T)$ in a single evaluation.

\paragraph{Baselines.}
Loc. Int. Diff. FNO~\citep{liuschiaffini2024localfno} is trained for single-step prediction with fixed $\Delta t=0.05$.
The model takes 6 input channels ($(h_0,(hu)_0,(hv)_0)$, two spatial-grid channels, and normalized $g$) and predicts the next-step state; architecture follows Section~\ref{app:2d_swe}.
Trajectory inference is obtained via autoregressive rollout.
For AR-Loc. Int. Diff. FNO, we follow~\citet{tran2023ffno}: teacher forcing under a Markov (single-step) assumption, Gaussian noise injection ($\sigma=0.001$), cosine learning-rate decay, batch size 32, and 100 training epochs.
U-Net~\citep{ronneberger2015unet} is trained with the same autoregressive protocol, using a standard 4-level encoder--decoder with skip connections and initial channel width 16.
Poseidon-B~\citep{herde2024poseidon} is included only for SWE trajectories: we fine-tune the pretrained Poseidon-B backbone (157M parameters) on 128 in-distribution trajectories using the protocol of~\citet{herde2024poseidon}
(backbone LR $5\times 10^{-5}$, embedding/time-embedding LR $5\times 10^{-4}$, weight decay $10^{-6}$, cosine scheduler, batch size 40, gradient clipping at 5.0, early stopping with patience 200 epochs).

\begin{figure}[t]
    \centering
    \includegraphics[width=0.95\textwidth]{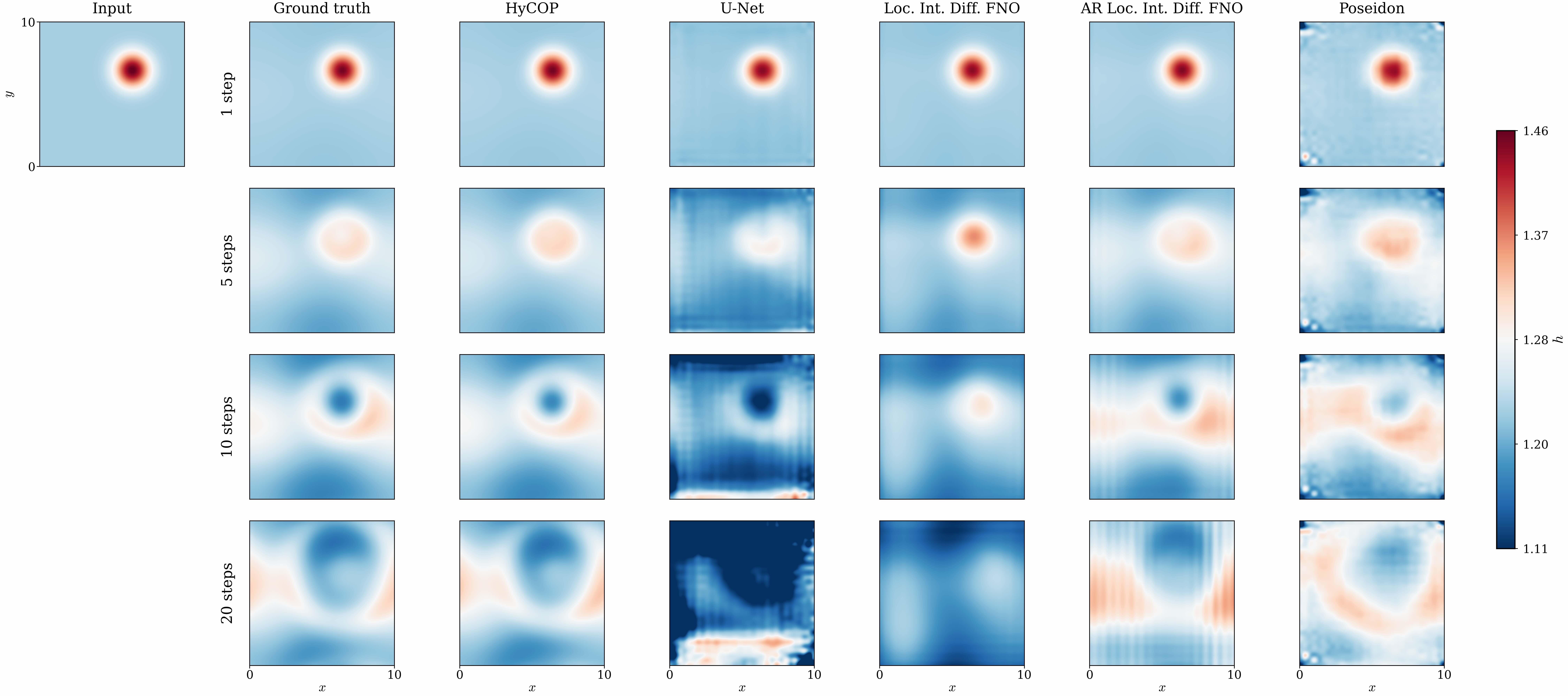}
    \caption{\textbf{2D SWE trajectory (ID).} Example long-horizon prediction at horizons of 1/5/10/20 steps for the 2D SWE system (Section~\ref{app:2d_swe_trajectory}).}
    \label{fig:2d-swe-traj-id}
\end{figure}

\begin{figure}[t]
    \centering
    \includegraphics[width=0.95\textwidth]{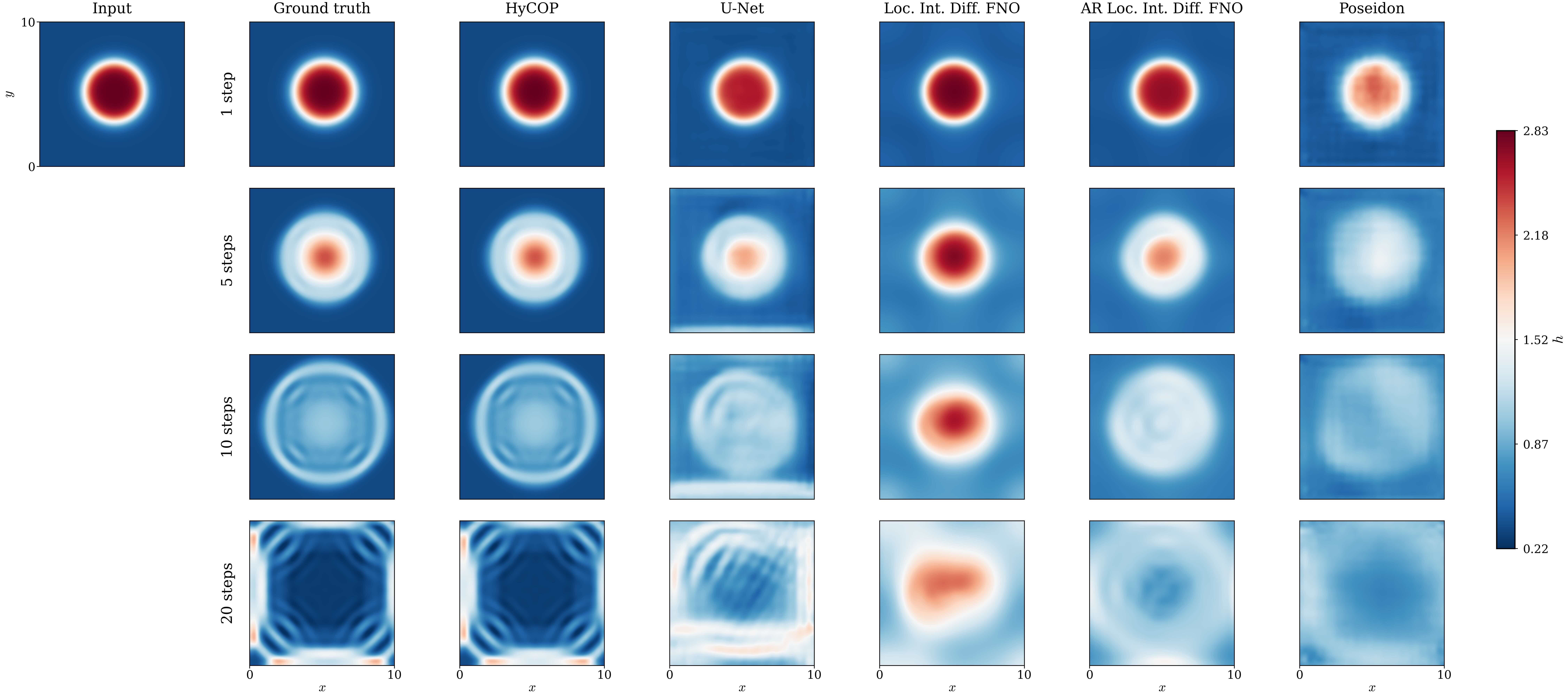}
    \caption{\textbf{2D SWE trajectory (OOD).} Same visualization as Figure~\ref{fig:2d-swe-traj-id} on OOD parameter/IC shifts (see Section~\ref{app:2d_swe}).}
    \label{fig:2d-swe-traj-ood}
\end{figure}

\begin{table}[H]
\caption{\textbf{2D SWE multi-step trajectory error (ID).} Relative $L^2$ and RMSE at horizons of 1/5/10/20 steps. Autoregressive models are rolled out; HyCOP answers each horizon via direct time querying (non-autoregressive).}
\label{tab:swe_multistep_id_app}
\vskip 0.15in
\begin{center}
\begin{small}
\begin{sc}
\resizebox{\textwidth}{!}{%
\begin{tabular}{lcccccccc}
\toprule
\multirow{2}{*}{Model} &
\multicolumn{2}{c}{1 step} & 
\multicolumn{2}{c}{5 steps} &
\multicolumn{2}{c}{10 steps} &
\multicolumn{2}{c}{20 steps} \\
\cmidrule(lr){2-3} \cmidrule(lr){4-5} \cmidrule(lr){6-7} \cmidrule(lr){8-9}
 & Rel.\ $L^2$ & RMSE & Rel.\ $L^2$ & RMSE & Rel.\ $L^2$ & RMSE & Rel.\ $L^2$ & RMSE \\
\midrule
U-Net        & $2.75 \cdot 10^{-1}$ & $2.25 \cdot 10^{-2}$ & $1.04 \cdot 10^{0}$ & $7.87 \cdot 10^{-2}$ & $2.28 \cdot 10^{0}$ & $1.46 \cdot 10^{-1}$ & $4.20 \cdot 10^{0}$ & $2.78 \cdot 10^{-1}$ \\
Loc. Int. Diff. FNO     & $9.01 \cdot 10^{-2}$ & $6.64 \cdot 10^{-3}$ & $4.95 \cdot 10^{-1}$ & $6.26 \cdot 10^{-2}$ & $1.37 \cdot 10^{0}$ & $1.22 \cdot 10^{-1}$ & $5.52 \cdot 10^{0}$ & $2.26 \cdot 10^{-1}$ \\
Poseidon     & $9.51 \cdot 10^{-2}$ & $1.13 \cdot 10^{-2}$ & $3.47 \cdot 10^{-1}$ & $2.92 \cdot 10^{-2}$ & $6.07 \cdot 10^{-1}$ & $4.91 \cdot 10^{-2}$ & $8.42 \cdot 10^{-1}$ & $7.37 \cdot 10^{-2}$ \\
AR-Loc. Int. Diff. FNO  & $6.66 \cdot 10^{-2}$ & $6.51 \cdot 10^{-3}$ & $1.72 \cdot 10^{-1}$ & $2.66 \cdot 10^{-2}$ & $2.84 \cdot 10^{-1}$ & $4.80 \cdot 10^{-2}$ & $4.86 \cdot 10^{-1}$ & $8.39 \cdot 10^{-2}$ \\
\textbf{HyCOP (ours)}
             & $\bm{7.12 \cdot 10^{-3}}$ & $\bm{8.65 \cdot 10^{-4}}$
             & $\bm{1.91 \cdot 10^{-2}}$ & $\bm{6.37 \cdot 10^{-3}}$
             & $\bm{4.06 \cdot 10^{-2}}$ & $\bm{1.03 \cdot 10^{-2}}$
             & $\bm{6.94 \cdot 10^{-2}}$ & $\bm{1.26 \cdot 10^{-2}}$ \\
\bottomrule
\end{tabular}
}
\end{sc}
\end{small}
\end{center}
\vskip -0.1in
\end{table}

\begin{table}[H]
\caption{\textbf{2D SWE multi-step trajectory error (OOD).} Same protocol as Table~\ref{tab:swe_multistep_id_app} evaluated on OOD parameter/IC shifts.}
\label{tab:swe_multistep_ood_app}
\vskip 0.15in
\begin{center}
\begin{small}
\begin{sc}
\resizebox{\textwidth}{!}{%
\begin{tabular}{lcccccccc}
\toprule
\multirow{2}{*}{Model} &
\multicolumn{2}{c}{1 step} & 
\multicolumn{2}{c}{5 steps} &
\multicolumn{2}{c}{10 steps} &
\multicolumn{2}{c}{20 steps} \\
\cmidrule(lr){2-3} \cmidrule(lr){4-5} \cmidrule(lr){6-7} \cmidrule(lr){8-9}
 & Rel.\ $L^2$ & RMSE & Rel.\ $L^2$ & RMSE & Rel.\ $L^2$ & RMSE & Rel.\ $L^2$ & RMSE \\
\midrule
U-Net        & $2.59 \cdot 10^{-1}$ & $2.16 \cdot 10^{-1}$ & $5.92 \cdot 10^{-1}$ & $5.35 \cdot 10^{-1}$ & $9.21 \cdot 10^{-1}$ & $6.73 \cdot 10^{-1}$ & $1.20 \cdot 10^{0}$ & $8.51 \cdot 10^{-1}$ \\
Loc. Int. Diff. FNO     & $3.10 \cdot 10^{-1}$ & $1.55 \cdot 10^{-1}$ & $6.93 \cdot 10^{-1}$ & $5.03 \cdot 10^{-1}$ & $1.32 \cdot 10^{0}$ & $6.93 \cdot 10^{-1}$ & $1.44 \cdot 10^{0}$ & $9.40 \cdot 10^{-1}$ \\
Poseidon     & $4.01 \cdot 10^{-1}$ & $3.37 \cdot 10^{-1}$ & $6.54 \cdot 10^{-1}$ & $6.14 \cdot 10^{-1}$ & $7.78 \cdot 10^{-1}$ & $6.96 \cdot 10^{-1}$ & $9.35 \cdot 10^{-1}$ & $8.01 \cdot 10^{-1}$ \\
AR-Loc. Int. Diff. FNO  & $2.72 \cdot 10^{-1}$ & $1.49 \cdot 10^{-1}$ & $5.33 \cdot 10^{-1}$ & $4.34 \cdot 10^{-1}$ & $7.89 \cdot 10^{-1}$ & $5.23 \cdot 10^{-1}$ & $7.55 \cdot 10^{-1}$ & $6.23 \cdot 10^{-1}$ \\
\textbf{HyCOP (ours)}
             & $\bm{1.01 \cdot 10^{-2}}$ & $\bm{4.84 \cdot 10^{-3}}$
             & $\bm{4.54 \cdot 10^{-2}}$ & $\bm{2.83 \cdot 10^{-2}}$
             & $\bm{8.30 \cdot 10^{-2}}$ & $\bm{4.20 \cdot 10^{-2}}$
             & $\bm{1.21 \cdot 10^{-1}}$ & $\bm{7.80 \cdot 10^{-2}}$ \\
\bottomrule
\end{tabular}
}
\end{sc}
\end{small}
\end{center}
\vskip -0.1in
\end{table}

\paragraph{Dam Break Transfer} \label{app:dam_swemnics}

For the SWE test cases, the trusted solver for data generation is SWEMniCS~\cite{dawson2024swemnics}. SWEMniCS solves the 2D SWE using finite element methods in space, and implicit finite differences in time. The solver is based on the FEniCS  framework~\cite{baratta2023dolfinx} and includes a test suite of predefined physically relevant cases.

\subsection{AD$\to$ADR adaptation: three configurations}
\label{app:ad_adr}

This experiment tests whether HyCOP can adapt when the target data contain physics absent during pretraining.
We consider transfer from advection--diffusion (AD) to advection--diffusion--reaction (ADR), where the reaction mechanism is either unknown (Path (a)) or known but lives alongside a pretrained AD surrogate (Path (b)).
We additionally include a fully-learned ablation that replaces every numerical primitive with a per-process FNO-FiLM surrogate.

\paragraph{Source and target systems.}
The source AD system is:
\begin{equation}
    \frac{\partial u}{\partial t} + c_x \frac{\partial u}{\partial x} + c_y \frac{\partial u}{\partial y}
    = D_x \frac{\partial^2 u}{\partial x^2} + D_y \frac{\partial^2 u}{\partial y^2},
\end{equation}
and the target ADR system $\mathcal{P}^{\text{target}}$ augments AD with an additional term:
\begin{equation}
    \frac{\partial u}{\partial t} + c_x \frac{\partial u}{\partial x} + c_y \frac{\partial u}{\partial y}
    = D_x \frac{\partial^2 u}{\partial x^2} + D_y \frac{\partial^2 u}{\partial y^2} + \mathcal{G}(u),
\end{equation}
where in our benchmark $\mathcal{G}(u)=ru(1-u)$. In Path (a), $\mathcal{G}$ is treated as unknown and learned only through black-box simulator queries; in Path (b), $\mathcal{G}$ is known and instantiated as a textbook numerical reaction solver. Parameter ranges and OOD splits follow Section~\ref{app:2d_adr}.

\subsubsection{Path (a): Compose--diagnose--enrich with a learned residual}
\label{app:ad_adr:path_a}

\paragraph{Pretraining on AD.}
We pretrain HyCOP, Loc.\ Int.\ Diff.\ FNO, and U-Net on 10{,}000 AD pair-data samples with
$c_x,c_y\in[0.2,1.5]$, $D_x,D_y\in[0.05,0.2]$, and target time $T=0.2$.
Training configurations follow Sections~\ref{app:2d_adr} and~\ref{app:2d_adr_trajectory}.
HyCOP learns a split program over two primitives corresponding to advection ($\mathcal{A}$) and diffusion ($\mathcal{D}$).

\begin{table}[H]
\caption{Zero-Shot Pretrained AD Inference Performance on ADR}
\label{table:adr-zeroshot}
\vskip 0.15in
\begin{center}
\begin{small}
\begin{sc}
\resizebox{0.5\textwidth}{!}{%
\begin{tabular}{lccc}
\toprule
Model & Rel. $L^2$ & RMSE & Max Error \\
\midrule
U-Net            & $1.81 \cdot 10^{-1}$      & $\bm{7.41 \cdot 10^{-2}}$ & $1.76 \cdot 10^{-1}$ \\
Loc. Int. Diff. FNO         & $1.86 \cdot 10^{-1}$      & $7.61 \cdot 10^{-2}$ & $1.60 \cdot 10^{-1}$ \\
HyCOP (pretrain) & $\bm{1.81 \cdot 10^{-1}}$ & $7.42 \cdot 10^{-2}$ & $\bm{1.05 \cdot 10^{-1}}$ \\
\bottomrule
\end{tabular}
}
\end{sc}
\end{small}
\end{center}
\vskip -0.1in
\end{table}

\paragraph{Dictionary enrichment via a black-box residual operator.}
To enable transfer, we augment the AD dictionary with a third primitive intended to capture the \emph{missing} physics.
We assume access to queries of both a simplified simulator (AD) and a target simulator (ADR) from the \emph{same} input state.
Using a small step size $\Delta t=0.01$, we estimate the instantaneous residual rate via a first-order finite difference:
\begin{equation}
    \mathcal{R}(u)\;\approx\;\frac{u_{\mathrm{ADR}}(\Delta t;u)-u_{\mathrm{AD}}(\Delta t;u)}{\Delta t}.
\end{equation}
For sufficiently small $\Delta t$, this approximates the additive drift contribution of the unknown term $\mathcal{G}(u)$.

We train a U-NO~\citep{rahman2023uno} to learn the map $u \mapsto \mathcal{R}(u)$ from state $u$ alone (single input channel),
using states sampled from 120 ADR trajectories.
The U-NO uses hidden width 32, 4 Fourier modes, and 4 layers.
After training, the U-NO parameters are frozen and treated as a residual primitive $\mathcal{R}_{\mathrm{UNO}}$.

\paragraph{HyCOP adaptation with a frozen residual primitive.}
We extend the HyCOP policy from two operators $(\mathcal{A},\mathcal{D})$ to three operators $(\mathcal{A},\mathcal{D},\mathcal{R}_{\mathrm{UNO}})$
by replacing sigmoid selection with a 3-way softmax.
We initialize the feature extractor, time-ratio head, and length head from AD pretraining, and adapt only the policy parameters.
Fine-tuning uses 120 ADR samples with ES (population size 50, noise std.\ 0.03, learning rate 0.005, weight decay 0.001) for 20 generations.
We evaluate (i) zero-shot AD-pretrained models and (ii) enriched-and-adapted HyCOP on 120 held-out ADR test samples.

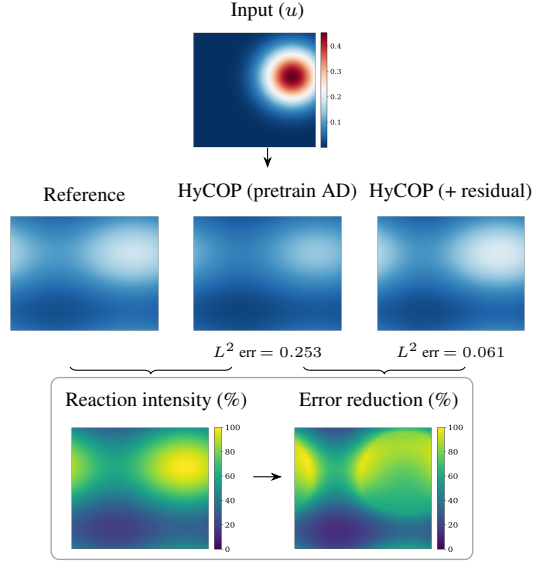
\begin{figure}[t]
  \centering
\begin{tikzpicture}[
  scale=0.95,
  every node/.style={transform shape} 
]

\def\wTop{2.05cm}
\def\hTop{1.60cm}
\def\wRow{2.05cm}
\def\hRow{1.60cm}
\def\wBot{2.35cm}
\def\hBot{1.80cm}

\node[ptitle] (sec_a) at (0,0) {Input ($u$)};
\node[panel, below=1pt of sec_a] (input) {
  \includegraphics[width=\wTop,height=\hTop]{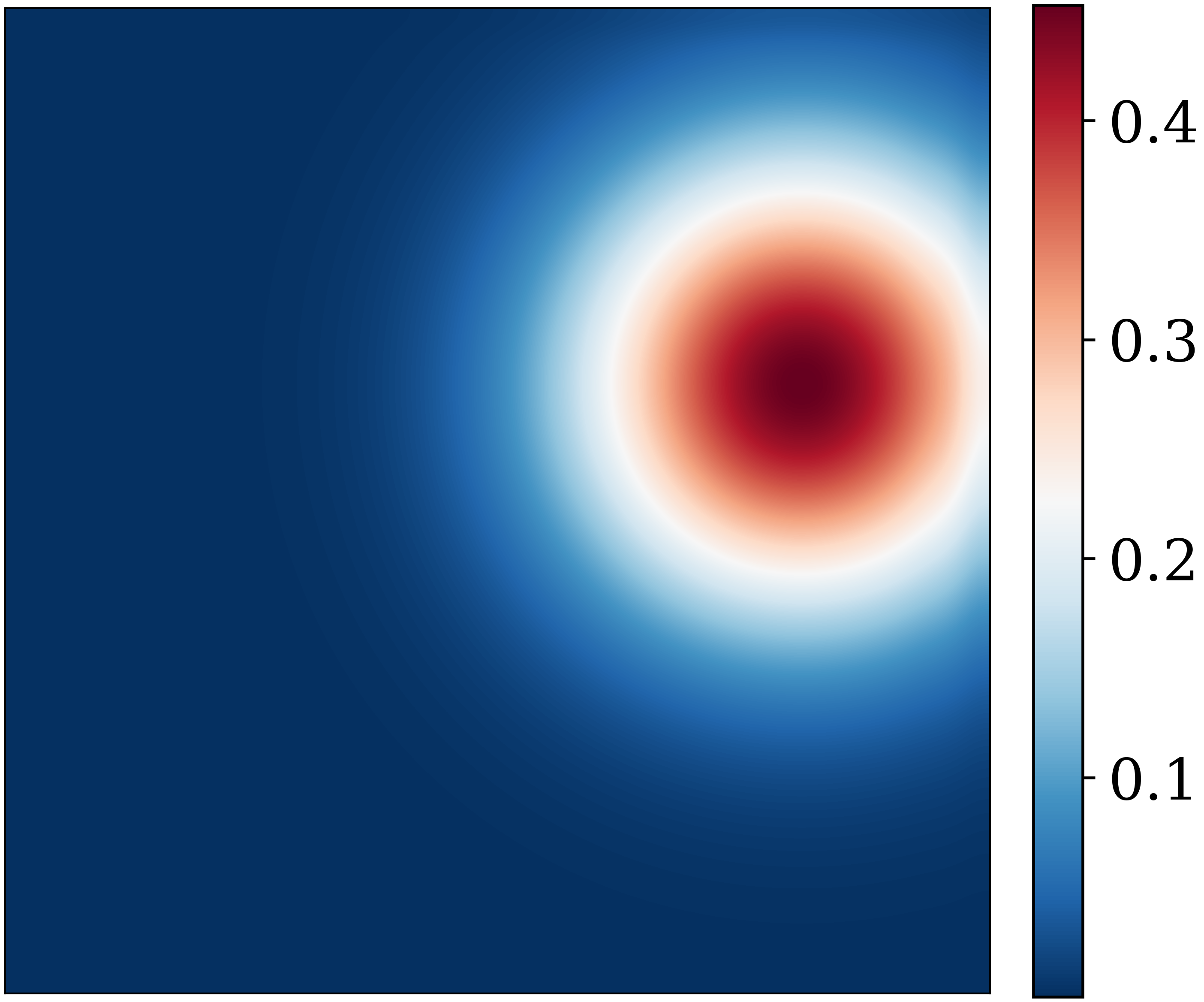}
};

\node[panel] (gt) at (-2.55,-3.65) {
  \includegraphics[width=\wRow,height=\hRow]{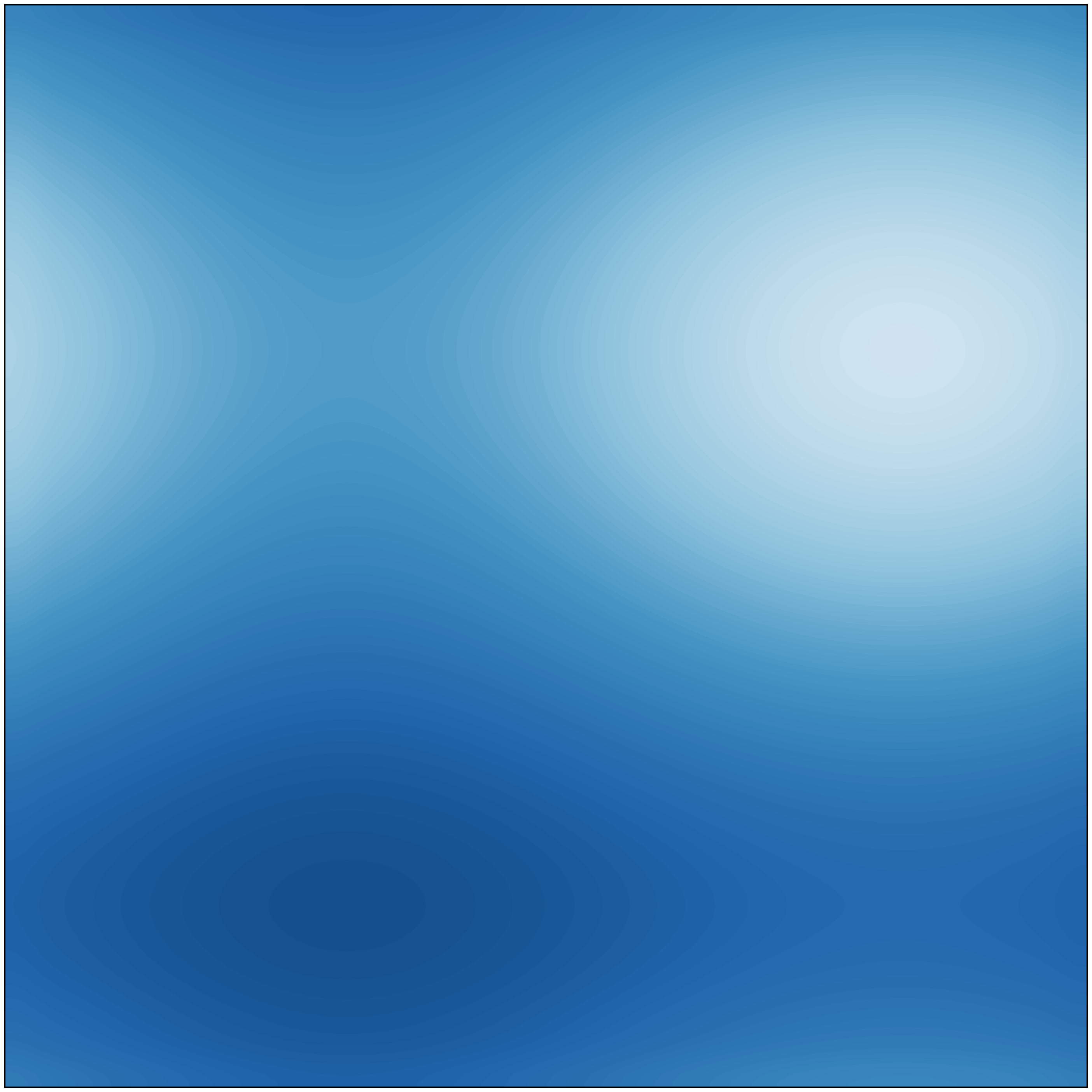}
};
\node[ptitle, above=2pt of gt] {Reference};

\node[panel] (pt) at (0,-3.65) {
  \includegraphics[width=\wRow,height=\hRow]{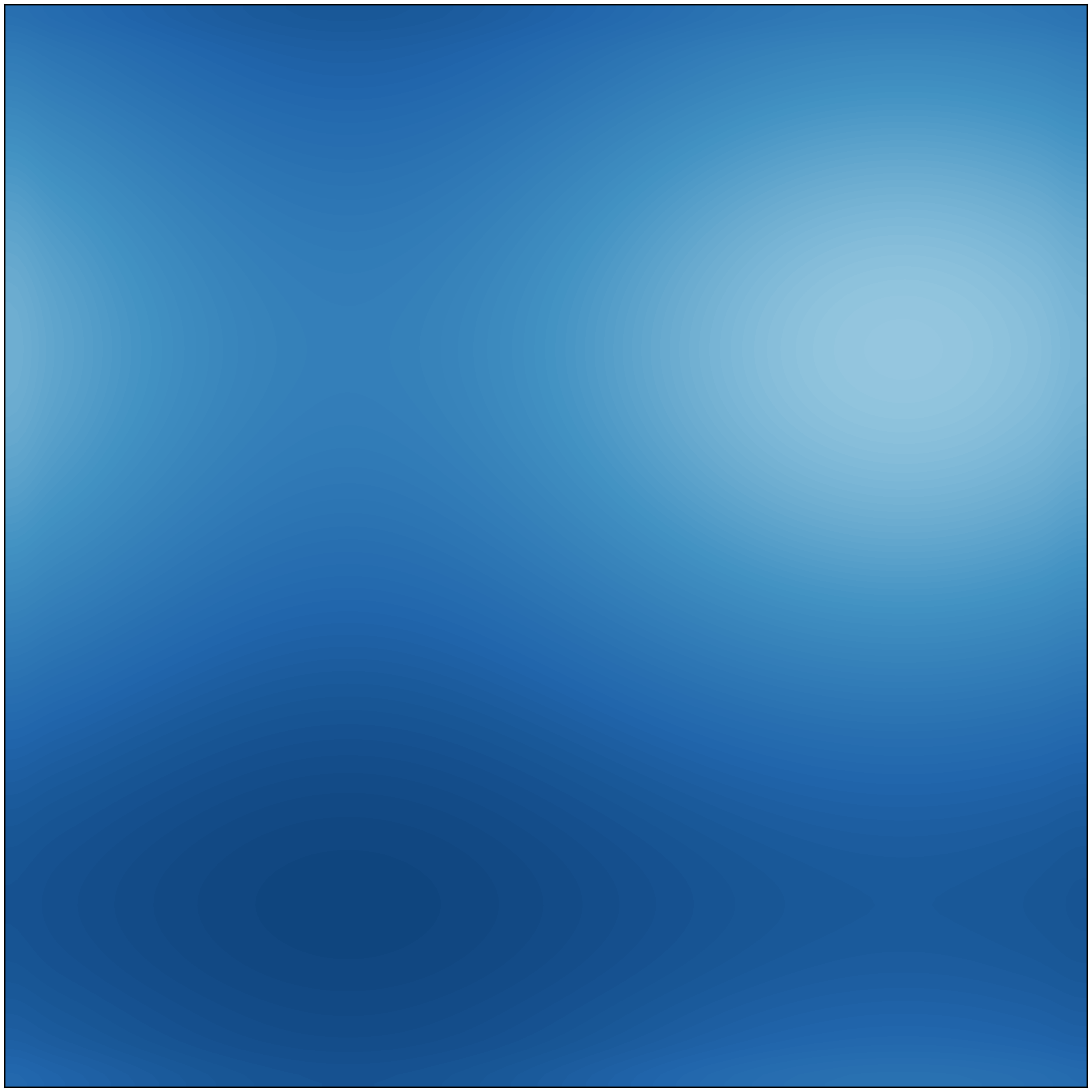}
};
\node[ptitle, above=2pt of pt] {HyCOP (pretrain AD)};
\node[psub, below=1pt of pt] {$L^2$ err $= 0.253$};

\node[panel] (ft) at (2.55,-3.65) {
  \includegraphics[width=\wRow,height=\hRow]{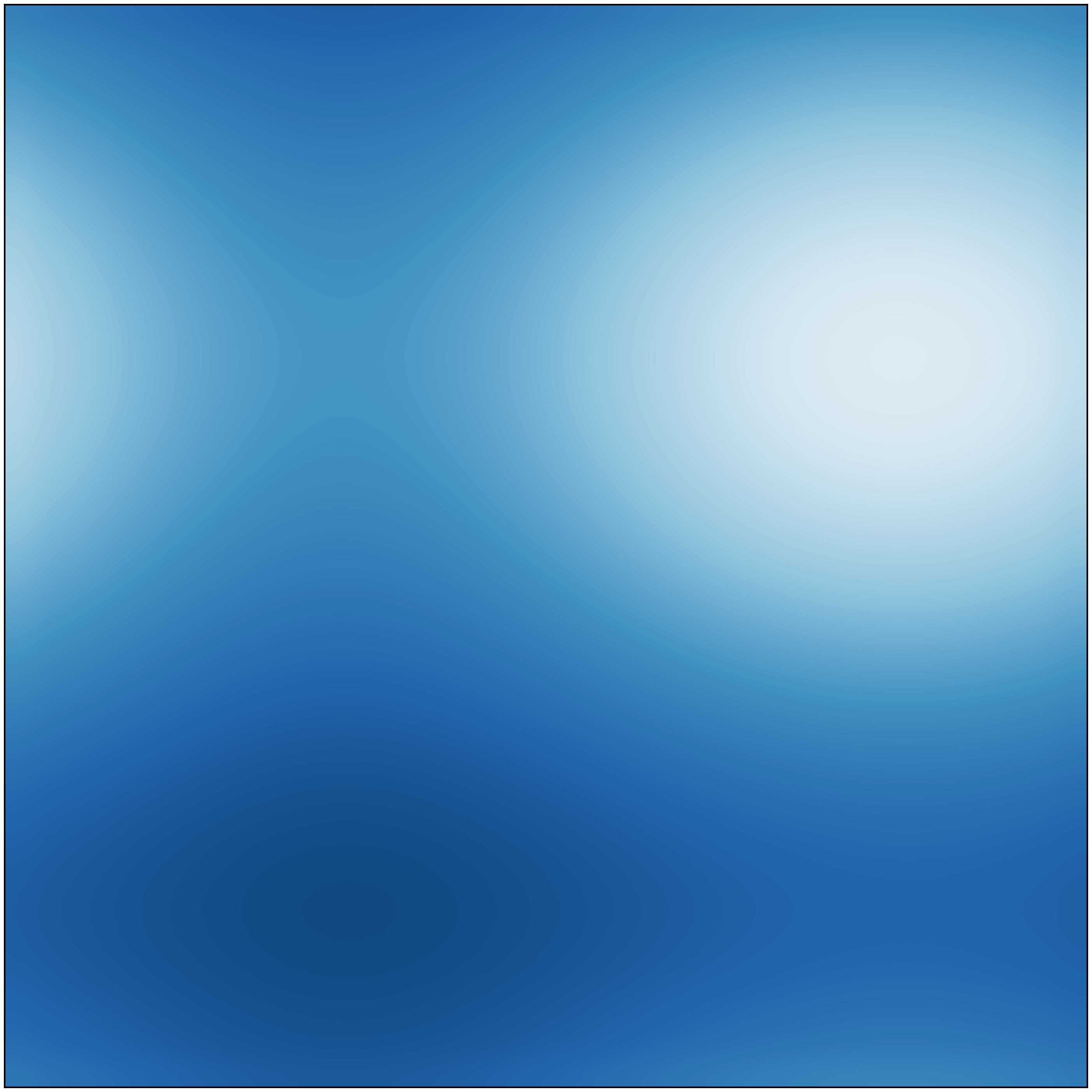}
};
\node[ptitle, above=2pt of ft] {HyCOP (+ residual)};
\node[psub, below=1pt of ft] {$L^2$ err $= 0.061$};

\node[panel] (reaction) at (-1.55,-6.65) {
  \includegraphics[width=\wBot,height=\hBot]{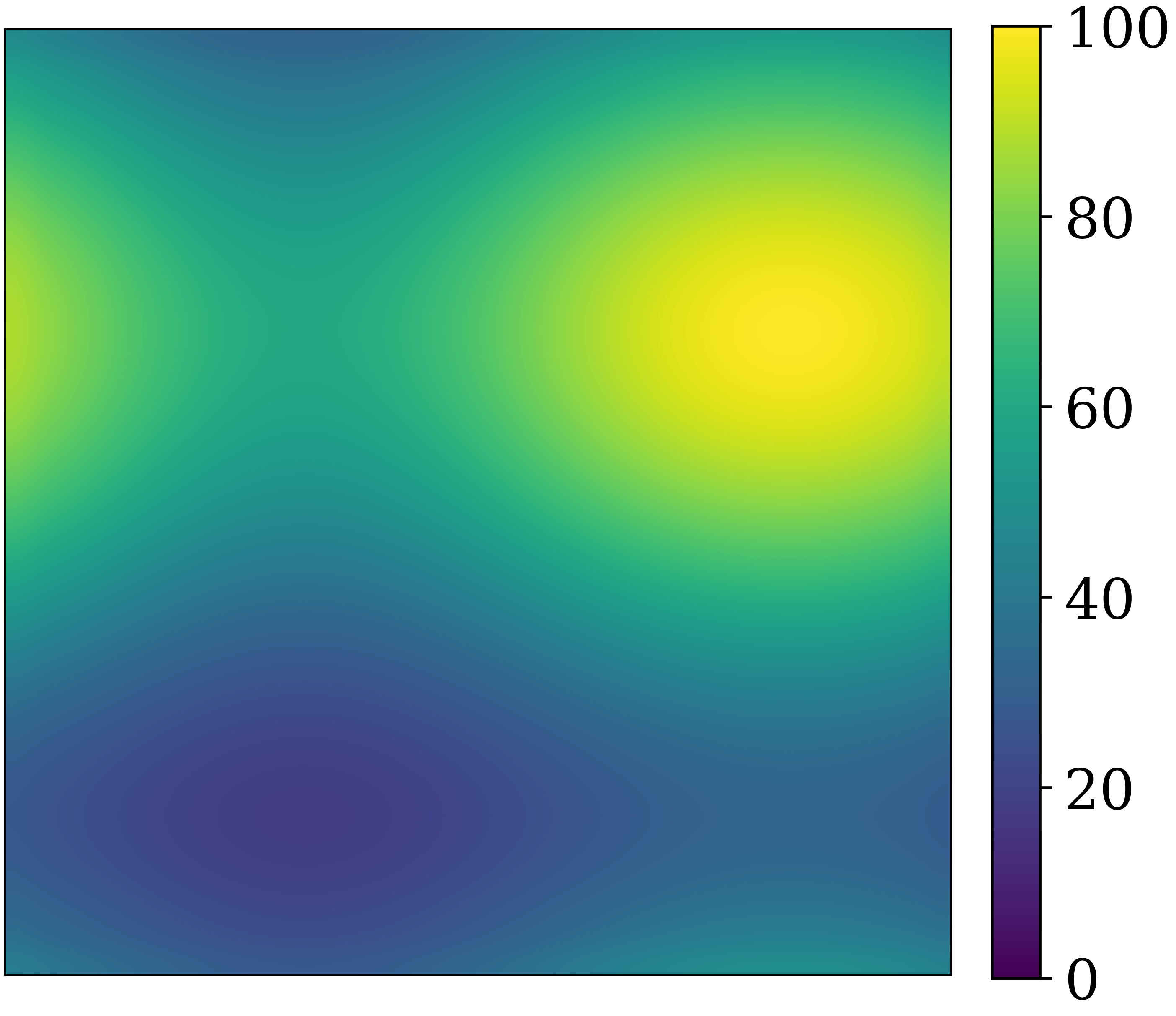}
};
\node[ptitle, above=2pt of reaction] (reaction_t) {Reaction intensity (\%)};

\node[panel] (delta) at (1.55,-6.65) {
  \includegraphics[width=\wBot,height=\hBot]{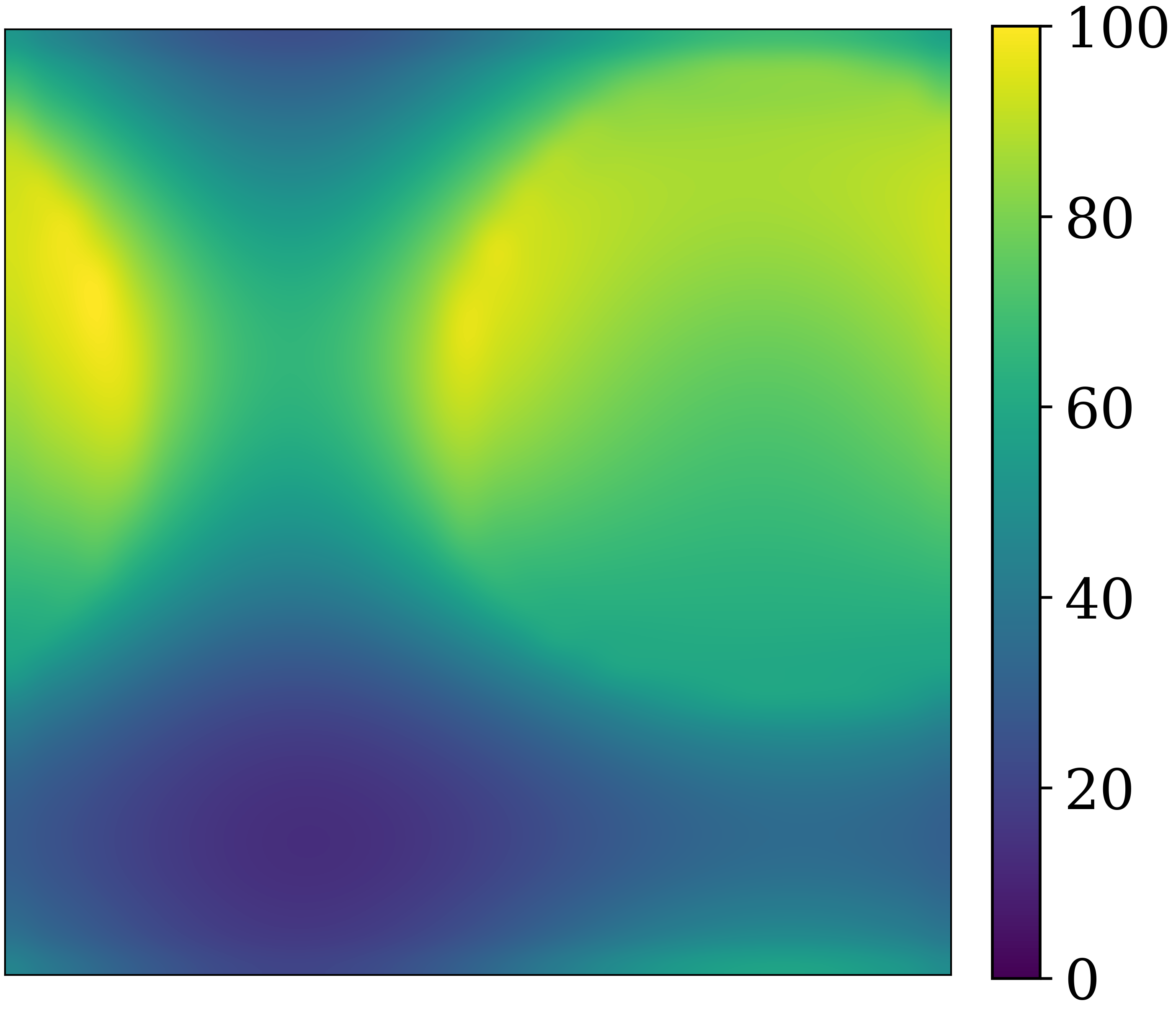}
};
\node[ptitle, above=2pt of delta] (delta_t) {Error reduction (\%)};

\begin{scope}[on background layer]
  \node[draw=black!35, line width=0.5pt, rounded corners=3pt,
        fit=(reaction)(delta)(reaction_t)(delta_t), inner sep=2pt] {};
\end{scope}

\draw[flowsmall] (input.south) -- ++(0,-0.28);

\draw[brace] (-2.75,-4.95) -- (-0.50,-4.95);
\draw[brace] (0.50,-4.95) -- (2.75,-4.95);

\draw[flowsmall] (-0.20,-6.48) -- (0.20,-6.48);

\end{tikzpicture}
\vspace{-0.4em} 
\caption{\textbf{AD$\rightarrow$ADR via dictionary enrichment (single illustrative sample).}
HyCOP pretrained on advection--diffusion (AD) misses reaction-driven changes in ADR. Adding a residual primitive and relearning the policy localizes corrections to regions with high reaction intensity, reducing error. Per-sample $L^2$ values shown; test-set averages are reported in Table~\ref{tab:regime_b}.}
\label{fig:adr_residual}
\end{figure}

\subsubsection{Path (b): Hybrid dictionary with a pretrained FNO-FiLM surrogate}
\label{app:ad_adr:path_b}

Path (b) models the scenario where a lab already possesses a pretrained surrogate for one sub-process and knows the missing sub-process analytically. HyCOP composes a learned primitive with a numerical primitive through a small policy---neither is retrained.

\paragraph{FNO-FiLM primitives.}
\label{app:fnofilm_details}
All learned primitives in this paper use the same FNO backbone as our FNO baseline (identical width, depth, and Fourier-mode count; see Section~\ref{app:common_setup} for the per-benchmark dimensions), augmented with FiLM conditioning~\citep{perez2018film} on the query time $\tau$. FiLM applies a feature-wise affine modulation $\gamma(\tau) \odot h + \beta(\tau)$ at each layer, where $\gamma, \beta$ are produced by a small 2-layer MLP (hidden size 64) from the scalar $\tau$. This lets a single FNO answer time-queried forward passes $\widehat{\Phi}^{(i)}_\tau(u)$ for a continuous range of durations $\tau$ rather than being tied to a fixed step, which is what HyCOP's variable-duration programs require.

\paragraph{Pretraining the FNO-FiLM AD surrogate.}
We pretrain a single FNO-FiLM on AD data with parameters $c_x,c_y\in[0.2,1.5]$, $D_x,D_y\in[0.05,0.2]$ and variable target time $T\in[0.02,0.2]$ (10{,}000 trajectories).
Training uses the common setup optimizer (Section~\ref{app:common_setup}). After pretraining, all FNO-FiLM parameters are frozen; only the policy is learned downstream.

\paragraph{Reaction primitive.}
The numerical reaction primitive implements $\mathcal{O}_{\mathrm{react}}^{\mathrm{num}}(u;\tau)$ by solving $\dot v = rv(1-v)$ pointwise for duration $\tau$ via an explicit RK4 step. It has zero learnable parameters.

\paragraph{HyCOP-Hyb adaptation.}
We train a 2-operator policy over the hybrid dictionary $\{\text{FNO-AD}, \mathcal{O}_{\mathrm{react}}^{\mathrm{num}}\}$ on 120 ADR samples from scratch (no AD policy warm-start, since the dictionary differs structurally from Path (a)). ES settings match Path (a): population size 50, noise std.\ 0.03, learning rate 0.005, weight decay 0.001, 20 generations. Neither the FNO-FiLM nor the reaction primitive is updated; only the ${\sim}$50-parameter policy is learned. Evaluation uses the same 120 held-out ADR test samples as Path (a).

\subsubsection{Ablation: fully learned primitive dictionary}
\label{app:ad_adr:fully_learned}

To stress-test the primitive-error term of the error decomposition (\S\ref{sec:theory}), we replace every numerical primitive with a per-process FNO-FiLM surrogate:
\[
\mathbb{D}_{\mathrm{Learned}}=\{\text{FNO-Adv}, \text{FNO-Diff}, \text{FNO-React}\}.
\]

\paragraph{Per-process pretraining.}
Each primitive is a separate FNO-FiLM (same backbone as in~\ref{app:fnofilm_details}) trained on single-process data for its own mechanism:
\begin{itemize}
  \item \textbf{FNO-Adv} on pure advection trajectories ($D_x{=}D_y{=}0$, $r{=}0$), $c_x,c_y\in[0.2,1.5]$.
  \item \textbf{FNO-Diff} on pure diffusion trajectories ($c_x{=}c_y{=}0$, $r{=}0$), $D_x,D_y\in[0.05,0.2]$.
  \item \textbf{FNO-React} on pure reaction trajectories ($c_x{=}c_y{=}D_x{=}D_y{=}0$), $r\in[0.1,1.0]$.
\end{itemize}
Each surrogate uses 10{,}000 single-process trajectories with variable target time $T\in[0.02,0.2]$ and the common-setup optimizer (Section~\ref{app:common_setup}). All three surrogates are frozen after pretraining.

\paragraph{HyCOP-Learned adaptation.}
We train a 3-operator policy over $\mathbb{D}_{\mathrm{Learned}}$ on the same 120 ADR samples used for Path (a) and Path (b), with identical ES settings. Only the ${\sim}$73-parameter policy is learned. This isolates the primitive-error term: the policy space is identical to HyCOP on ADR, so any gap traces to primitive quality, not to compositional structure.

\subsection{HyCOP on chaotic/multiscale PDEs}
\label{app:1d_chaos}
We use KS as a stress test to verify that the learned splitting policy remains stable and reproduces correct long-time statistics under simultaneous domain and horizon shift. KS exhibits spatiotemporal chaos whose intensity grows with domain width $W$ (more unstable modes); the stability constants in Definition~\ref{def:stable_split} depend on $W$ and uniform bounds may not hold across all domain sizes. Moreover, chaotic dynamics cause pointwise errors to grow exponentially (Lyapunov divergence), making the error bounds in Theorem~\ref{thm:fitting} uninformative at long times. We include KS as an empirical validation that HyCOP's learned splitting remains effective in this regime.

\subsubsection{Kuramoto--Sivashinsky (1D)}
\label{app:ks}
We consider KS on $x\in[0,2\pi W]$:
\[
u_t + uu_x + u_{xx} + u_{xxxx}=0,
\]
which exhibits spatiotemporal chaos for large $W$.

\paragraph{Data generation.}
We generate 10{,}000 trajectories with an ETDRK4 pseudospectral solver (Kassam--Trefethen contour integration), using $N=128$ and $\Delta t=0.02$.
Training samples $W\sim\mathrm{Unif}[24,40]$ and query times $T\sim\mathrm{Unif}[5,8]$.
Initial conditions are (i) two-mode sinusoids with random phases/amplitudes and (ii) low-frequency random Fourier series; all are mean-subtracted.

\paragraph{Dictionary.}
We use the canonical split
\[
\mathbb{D}_{\mathrm{KS}}=\{\mathcal{O}_{\mathrm{lin}},\mathcal{O}_{\mathrm{nl}}\},\quad
\mathcal{O}_{\mathrm{lin}}(u)=-u_{xx}-u_{xxxx},\quad
\mathcal{O}_{\mathrm{nl}}(u)=-\tfrac12(u^2)_x.
\]
$\mathcal{O}_{\mathrm{lin}}$ is exact in Fourier space via $\hat{u}(t+\Delta t)=e^{(k^2-k^4)\Delta t}\hat{u}(t)$; $\mathcal{O}_{\mathrm{nl}}$ uses SSPRK3 with $3/2$-rule dealiasing and CFL-based substepping.

\paragraph{Metrics and OOD.}
Pointwise errors diverge in chaotic systems, so we evaluate attractor-level statistics: spectrum error (SE) in the time-averaged log-energy spectrum and KL divergence between long-time state distributions (excluding the first 10\% as transient).
OOD shifts extrapolate domain and horizon: $W\in[40,50]$, $T\in[8,20]$, and the combined OOD-WT setting.
HyCOP achieves low SE and KL across shifts (Table~\ref{tab:ablations}, panel (e), and Figures~\ref{fig:1d-ks-id}--\ref{fig:1d_ks_ood_wt}), indicating stable learned splitting that preserves the chaotic attractor under individual and combined domain/horizon extrapolation.

\begin{figure}[t]
    \centering
    \includegraphics[width=0.6\textwidth]{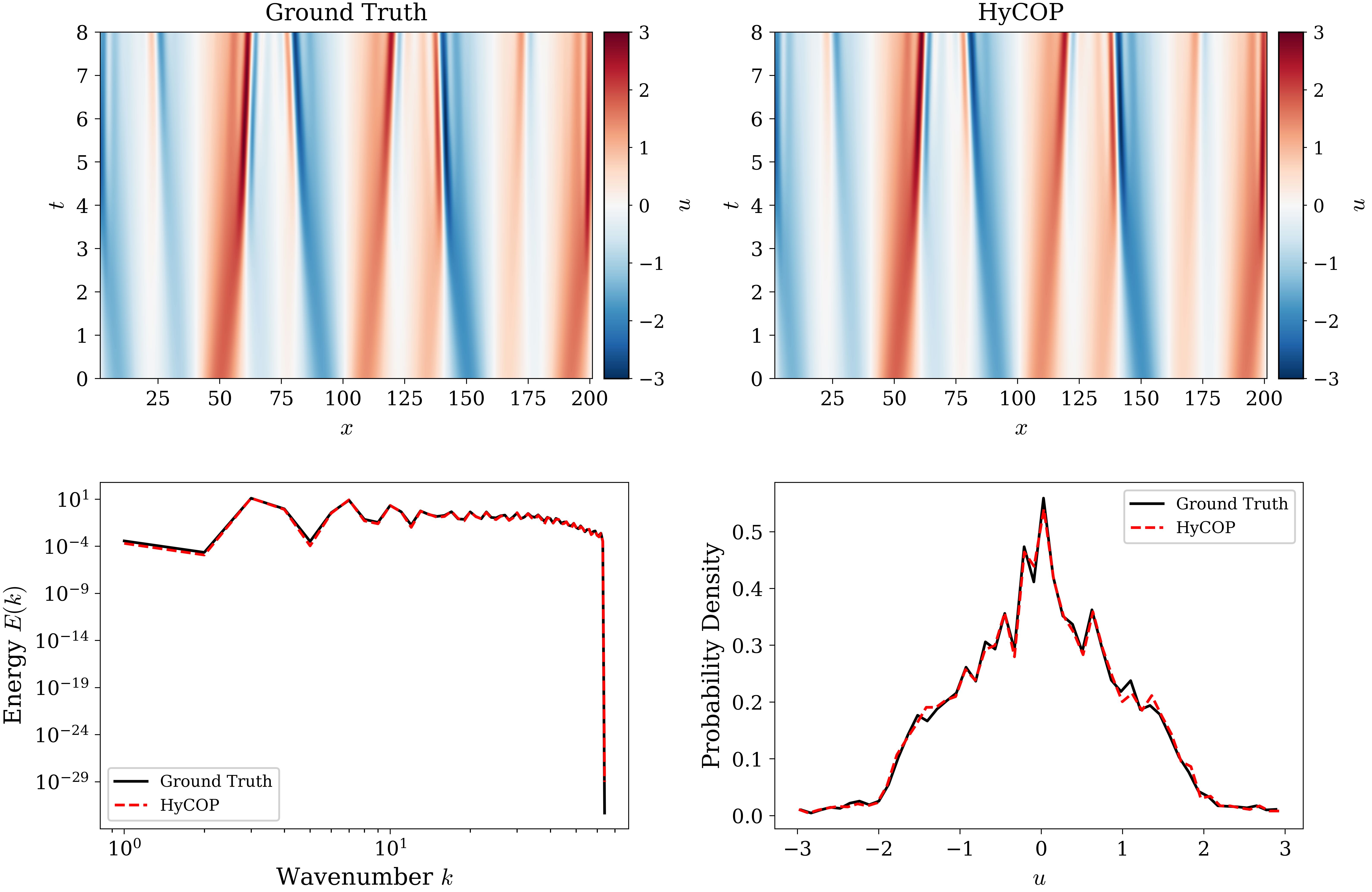}
    \caption{\textbf{1D KS (ID).} HyCOP reproduces the spatiotemporal structure (top), energy spectrum (bottom-left), and long-time state distribution (bottom-right) within the training regime ($W\in[24,40]$, $T\in[5,8]$).}
    \label{fig:1d-ks-id}
\end{figure}

\begin{figure}[t]
    \centering
    \includegraphics[width=0.6\textwidth]{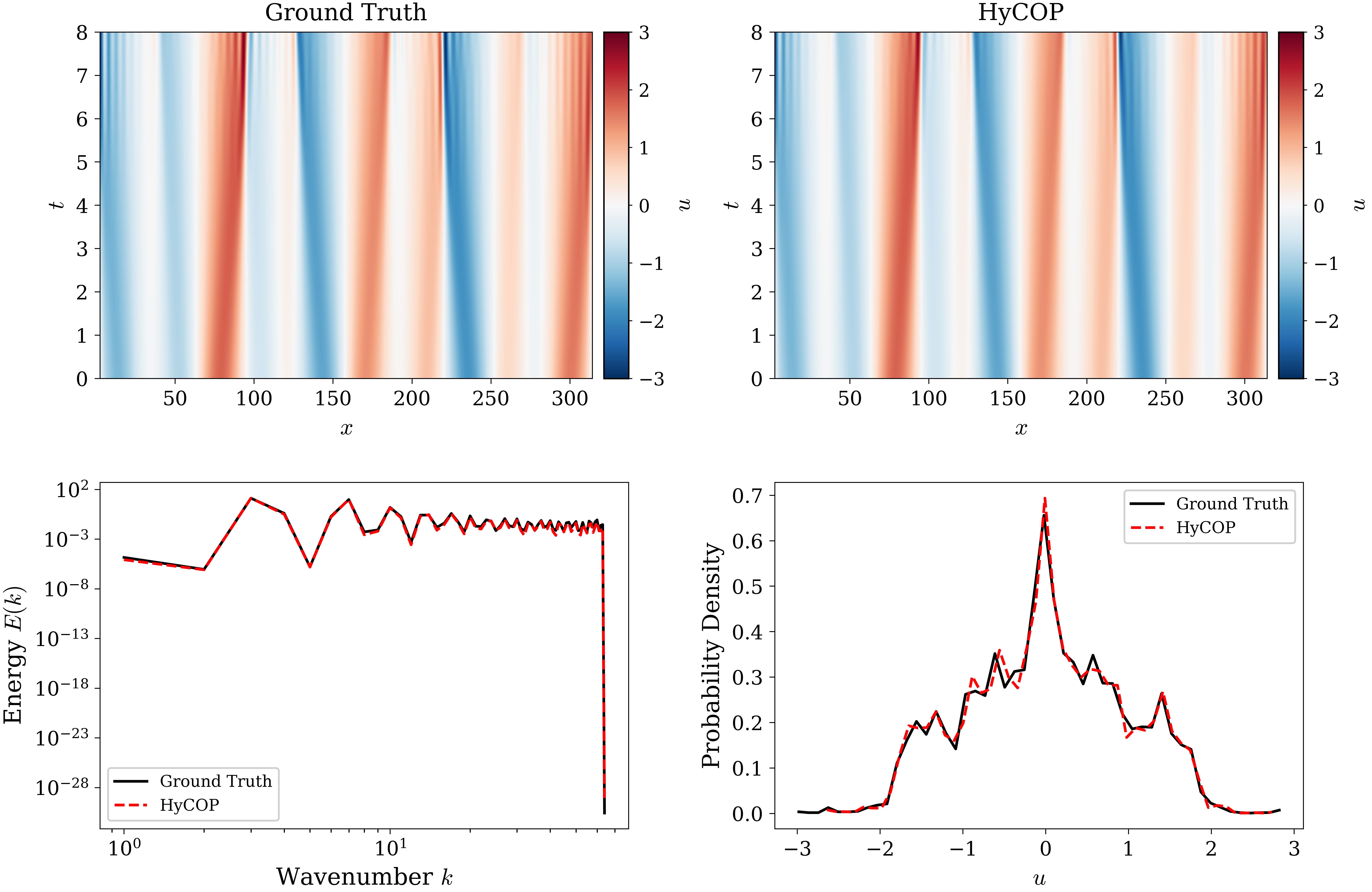}
    \caption{\textbf{1D KS (OOD: domain shift).} HyCOP generalizes to $W=50$ ($1.25\times$ training maximum) while preserving attractor statistics.}
    \label{fig:1d_ks_ood_w}
\end{figure}

\begin{figure}[t]
    \centering
    \includegraphics[width=0.6\textwidth]{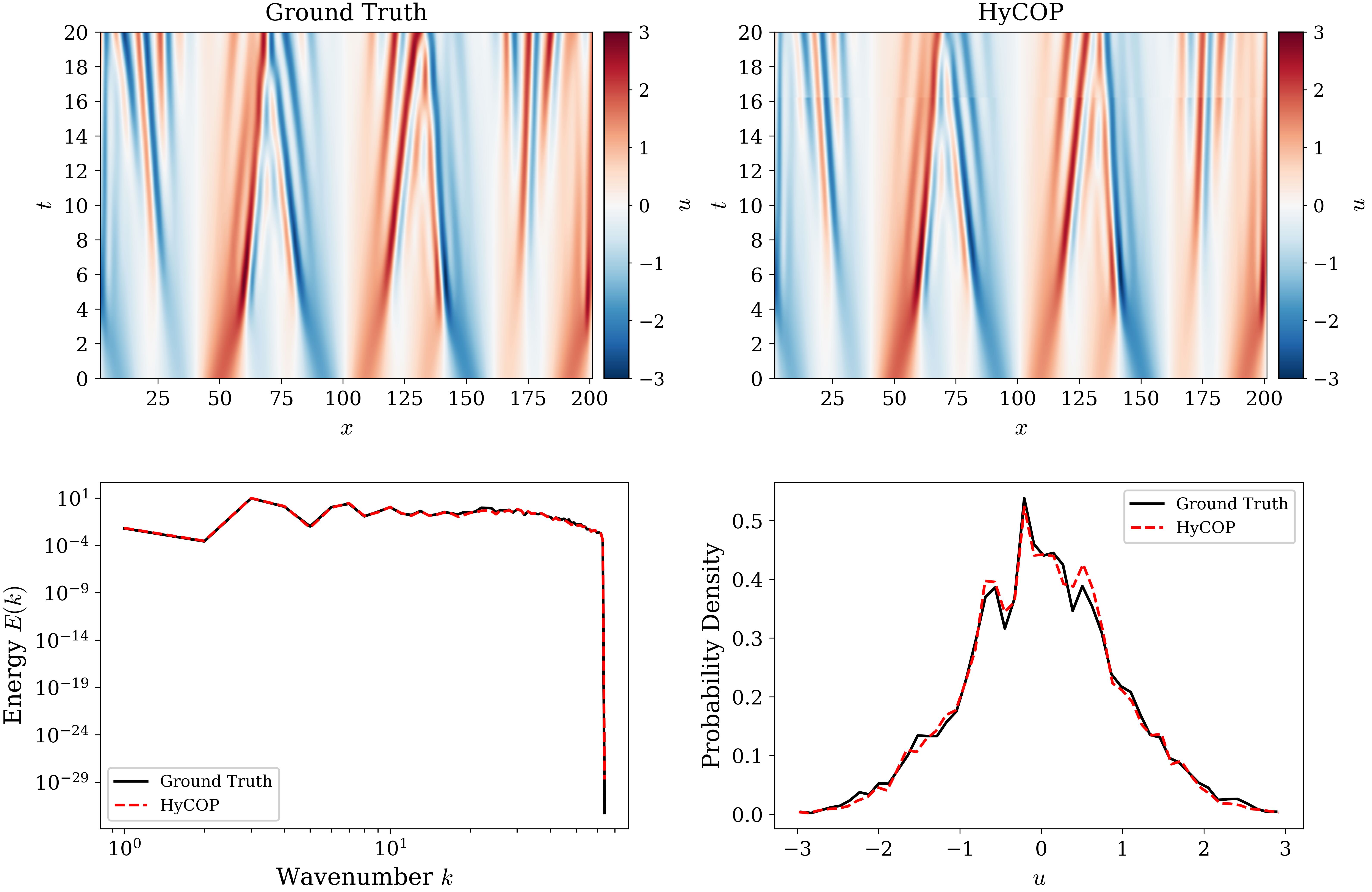}
    \caption{\textbf{1D KS (OOD: horizon shift).} HyCOP generalizes to $T=20$ ($2.5\times$ training maximum) while maintaining stable energy spectrum and state distribution.}
    \label{fig:1d_ks_ood_t}
\end{figure}

\begin{figure}[t]
    \centering
    \includegraphics[width=0.6\textwidth]{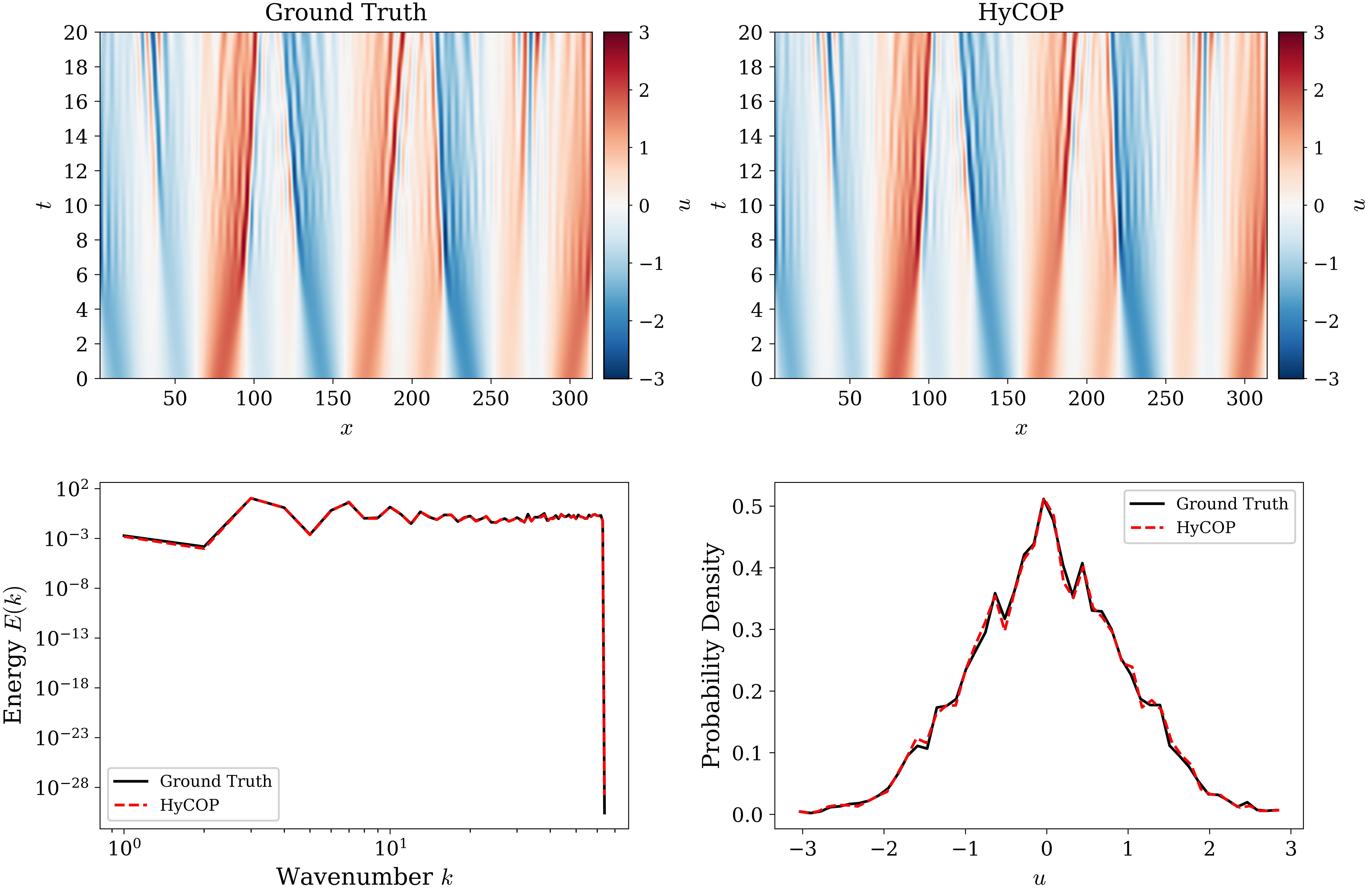}
    \caption{\textbf{1D KS (OOD: combined domain--horizon shift).} HyCOP generalizes to the joint extrapolation ($W=50$, $T=20$) while preserving attractor statistics: spatiotemporal structure (top), energy spectrum (bottom-left), and long-time state distribution (bottom-right). Pointwise errors are uninformative under Lyapunov divergence; HyCOP matches the reference solver at the level of invariant statistics.}
    \label{fig:1d_ks_ood_wt}
\end{figure}


\end{document}